\documentclass[journal,onecolumn]{IEEEtran}
\usepackage[utf8]{inputenc}
\usepackage[margin=1in]{geometry}
\usepackage{fancyhdr}

\usepackage{mathpazo}
\usepackage{amssymb}
\usepackage{amsmath}
\usepackage{mathtools}
\usepackage{relsize}
\usepackage[hidelinks]{hyperref}
\usepackage{amsthm}
\usepackage{physics}
\usepackage{enumitem}
\usepackage{mathrsfs}
\usepackage{stmaryrd}
\SetSymbolFont{stmry}{bold}{U}{stmry}{m}{n}

\usepackage{matlab-prettifier}

\usepackage{float} %For forcing figures somewhere
\allowdisplaybreaks %For making calculations go over multiple pages
\usepackage{hyperref} %citations
\usepackage{multirow} %tables
\usepackage{mwe}
\usepackage{listing}
\usepackage[en-US]{datetime2}
\usepackage{placeins}
\usepackage{capt-of}
\usepackage{tikz}
\usepackage[normalem]{ulem}

\usepackage{comment}
\usepackage[sort&compress,numbers]{natbib}
\newcommand{\ve}{\varepsilon}
\newcommand{\mrm}{\mathrm}
\newcommand{\mbb}{\mathbb}
\newcommand{\mbf}{\mathbf}

\newcommand{\wt}{\widetilde}
\newcommand{\ol}{\overline}

\newcommand{\supp}{\mrm{supp}}

\DeclarePairedDelimiterX{\inner}[2]{\langle}{\rangle}{#1,\,#2}

\renewcommand{\abs}[1]{\lvert #1 \rvert}

\newcommand{\cD}{\mathcal{D}}
\newcommand{\cE}{\mathcal{E}}
\newcommand{\cF}{\mathcal{F}}

\newcommand{\cP}{\mathcal{P}}

\newcommand{\cS}{\mathcal{S}}
\newcommand{\cT}{\mathcal{T}}

\newcommand{\cW}{\mathcal{W}}
\newcommand{\cX}{\mathcal{X}}
\newcommand{\cY}{\mathcal{Y}}

\newcommand{\Lin}{\mathrm{L}}

\newcommand{\Pos}{\mathrm{Pos}}
\newcommand{\Herm}{\mathrm{Herm}}
\newcommand{\Channel}{\mathrm{C}}

\newcommand{\Density}{\mathrm{D}}

\newcommand{\Pd}{\mathrm{Pd}}

\newenvironment{mylist}[1]{\begin{list}{}{
	\setlength{\leftmargin}{#1}
	\setlength{\rightmargin}{0mm}
	\setlength{\labelsep}{2mm}
	\setlength{\labelwidth}{8mm}
	\setlength{\itemsep}{0mm}}}
	{\end{list}}

\newcounter{problemcounter}

\theoremstyle{definition}
\newtheorem{theorem}{Theorem}
\newtheorem{lemma}[theorem]{Lemma}

\newtheorem{definition}[theorem]{Definition}

\newtheorem{corollary}[theorem]{Corollary}

\newtheorem{proposition}[theorem]{Proposition}
\newtheorem{example}[theorem]{Example}
\newtheorem{fact}[theorem]{Fact}

\begin{document}

\title{Divergence Inequalities with Applications in Ergodic Theory}
\author{Ian George, Alice Zheng, and Akshay Bansal
\thanks{Ian George is with the Department of Electrical and Computer Engineering, National University of Singapore, Singapore.}
\thanks{Alice Zheng and Akshay Bansal are with the Department of Computer Science, Virginia Polytechnic Institute and State University, Blacksburg, VA, USA 24061.}
}

\maketitle

\begin{abstract}
    The data processing inequality is central to information theory and motivates the study of monotonic divergences. However, it is not clear operationally we need to consider all such divergences. We establish a simple method for Pinsker inequalities as well as general bounds in terms of $\chi^{2}$-divergences for twice-differentiable $f$-divergences. 
    These tools imply new relations for input-dependent contraction coefficients. We  use these relations to show for many $f$-divergences the rate of contraction of a time homogeneous Markov chain is characterized by the input-dependent contraction coefficient of the $\chi^{2}$-divergence.
    This is efficient to compute and the fastest it could converge for a class of divergences. We show similar ideas hold for mixing times. Moreover, we extend these results to the Petz $f$-divergences in quantum information theory, albeit without any guarantee of efficient computation. These tools may have applications in other settings where iterative data processing is relevant.
\end{abstract}

%---Main Document---%
\section{Introduction}
A central idea in information theory is that of data processing. Conceptually, it captures the idea that if you have any object that encodes information and you alter it in any manner, you could only have lost information. This idea is so central because it is so general. To see that it is so general, consider the following two standard examples. First, if one has a file encoded in bits and they perform a function on the bit string, they can only have lost information contained in the original bit string. Second, if one has a physical system initialized in some known configuration and it then undergoes any allowed dynamics, one can only know less about the configuration of the physical system after the dynamics. The former of these results in the nuances of compression and communication \cite{Shannon-1948a} and the latter of these has been argued to be what leads to the second law of thermodynamics \cite{Merhav-2011a}. While these settings seem quite distinct, at the level of data processing, they are equivalent: there is an initial way things are that one has some (possibly partial) description of, and, without introducing new information from outside, any way things may change can only make that description less accurate. It is this generality that makes data processing so worthwhile to study.

To formally study data processing, one uses monotonic divergences. By a divergence $\mbb{D}$, we mean any measure of difference between a set of objects $\cS$, i.e. $\mbb{D}: \cS \times \cS \to \mbb{R}_{\geq 0}$. Such sets could be bit strings or the configuration of physical systems given the above examples. By monotonic, we mean that if $\cS$ represents a set of objects that encode information and $\Channel$ is a set of possible ways to process the information encoded in $\cS$, then $\mbb{D}(\cE(\rho) \Vert \cE(\sigma)) \leq \mbb{D}(\rho \Vert \sigma)$ for all $\rho,\sigma \in \Density$ and $\cE \in \Channel$. In this way, monotonic divergences are exactly the divergences that formally capture data processing--- indeed often being monotonic is referred to as `satisfying the data processing inequality' and we will use these terms interchangeably in this work. The standard example of these sets is in classical (resp.~quantum) information theory where $\cS$ is the set of distributions (resp.~quantum states) and $\Channel$ is the set of classical (resp.~quantum) channels, which will be the sets we primarily consider in this work.

Given the importance of data processing and monotonic divergences, people have studied the mathematical structure gives rise to monotonic divergences, e.g. \cite{Csiszar-1967a, Lindblad-1974a, Petz-1986a, Merhav-2011a, Beigi-2013a,Tomamichel-2016a,Gour-2021a} and references therein, as well as find various applications for them (see \cite{cover2006, Wilde-Book, Khatri-2020a, polyanskiy-2023a} for many such applications). In doing so, it has become of interest to relate the different divergences to each other in the sense of bounding one by another (See in particular \cite{Sason-2016-f-div-ineqs} and references therein). In this endeavour to better understand monotonic divergences, one of the most famous class of monotonic divergences is that of $f$-divergences introduced in \cite{Csiszar-1967a}, which inherit their monotonicity from convexity of the choice of $f$--- thereby making them a very large class of divergences. Much is known about how various $f$-divergences relate to each other under various settings and assumptions, e.g. \cite{Barnett-2002a,Dragomir-2016a, Sason-2016-f-div-ineqs,Sason-2018a, Dragomir-2019a} although this is by no means an exhaustive list of such results. Much of this research is interested in establishing results that generalize results known for the Kullback-Leibler (KL) divergence, $D(\mbf{p}\Vert \mbf{q}) = \sum_{x} p(x)\log(p(x)/q(x))$, which is obtained from the $f$-divergences via setting $f(t) = t\log(t)$. In particular, Pinsker-type inequalities, which lower bound $f$-divergences by the total variational distance, $\text{TV}(\mbf{p},\mbf{q}) \coloneq \frac{1}{2}\Vert \mbf{p} - \mbf{q}\Vert_{1}$ and reverse Pinsker inequalities that upper bound $f$-divergences by the total variational distance \cite{Pinsker-1960a, Vajda-1970a, Reid-2009, Gilardoni-2010a, Sason-2016-f-div-ineqs}. This research has been further generalized to relating arbitrary $f$-divergences to each other as the total variational distance is in fact also a specific $f$-divergence. From a practical standpoint, an application of studying monotonic divergences, including the $f$-divergences, has been contraction coefficients and `strong data processing inequalities,' which characterize when the monotonicity of a divergence is always strict for a given classical or quantum channel and possibly a specific initial distribution or quantum state, the latter case being called `input-dependent contraction coefficients.' These contraction coefficients have found uses in studying communication over networks, fault tolerance, and differential privacy \cite{Ahlswede-1976a,Cohen-1993a,Choi-1994a,Raginsky-2016a,Polyanskiy-2017a,Makur-2020a,Ordentlich-2021a,Hirche-2022a,Fawzi-2022a,Hirche-2023a,Dobrushin-1970a,Evans-2000a,Duchi-2013a,Polyanskiy-2015a,Hirche-2023-dp,Nuradha-2024a}.

In this work, we are interested in two basic questions within the framework of divergences and studying monotonicity:
\begin{enumerate}
    \item Are there simple general methods to relate divergences?
    \item Do we really need to consider all the $f$-divergences for information processing tasks limited by data processing?
\end{enumerate}
The former question is perhaps largely mathematically motivated, but it turns out to be quite useful. To answer this, we provide a simple method for establishing Pinsker inequalities for twice-differentiable $f$-divergences (Section~\ref{sec:f-div-Pinsker}) as well as an input-dependent method of relating arbitrary twice-differentiable $f$-divergences to the $\chi^{2}$-divergence (Section~\ref{sec:divergence-inequalities}). The latter of these implies reverse Pinsker inequalities and relations between arbitrary twice-differentiable $f$-divergences. These results will also serve as the technical tools for addressing the latter question.

The latter of our basic questions is more practically motivated. As mentioned earlier, contraction coefficients of $f$-divergences have been considered extensively and have applications in analyzing different information processing tasks. However, considering all $f$-divergences has clear disadvantages. First, when a problem has bounds for arbitrary $f$-divergences, it is not reasonable to compute the bounds for every $f$-divergence, so the applicability of such results is unclear. Similarly, it is not known how to efficiently compute contraction coefficients for many $f$-divergences. Furthermore, it is not even obvious that there is much information to be gained from considering a large set of $f$-divergences. Indeed, if we could instead relate many $f$-divergences to a single $f$-divergence whose properties are easy to compute with a manageable penalty, we could focus on that specific $f$-divergence. It is this proposed strategy that this work addresses by answering these two basic questions.

\paragraph{Summary of Results}
\begin{enumerate}
    \item In Section~\ref{sec:f-div-Pinsker} we establish a general method for obtaining Pinsker inequalities for twice-differentiable $f$-divergences using forms of Taylor's theorem (Theorems \ref{thm:pinsker-uni} and \ref{thm:pinsker-multi}). Moreover, the result can be tight, such as for the KL divergence. We remark that our method does not rely on unit normalization, which may be relevant in optimization theory.
    \item In Section~\ref{sec:divergence-inequalities} we establish a simple method for relating all twice-differentiable $f$-divergences to the $\chi^{2}$-divergence via Taylor's theorem (Theorem~\ref{thm:f-div-chi-squared-bounds}). This implies relations between all twice-differentiable $f$-divergences as well as reverse Pinsker inequalities. Moreover, it shows that, up to some penalty, working with the $\chi^{2}$-divergence could be sufficient. We also show how these methods work for Bregman divergences, which, while not generically monotonic, have uses in optimization theory and statistics.
    \item In Section~\ref{sec:classical-Input-Dependent-SDPI}, we apply our methods to studying input-dependent contraction coefficients to show that often times working with the $\chi^{2}$-divergence contraction coefficient is sufficient. In particular, we show that when considering the asymptotic behaviour of iterative applications of the same channel (equivalently a time-homogeneous Markov chain), most twice-differentiable $f$-divergence input-dependent contraction coefficients scale at a rate given by the $\chi^{2}$ input-dependent contraction coefficient. This is particularly convenient as this rate is known to be computable and is the fastest rate possible using contraction coefficients for this class of $f$-divergences \cite{Polyanskiy-2017a}. We also show how the related previous work \cite{Raginsky-2016a} implies bounds on computable mixing times for Markov chains both in terms of total variational distance, but also under more stringent conditions of measuring distance with $f$-divergences. In total, this strengthens our theory of ergodic Markov chains in terms of information-theoretic quantities. It is likely there are applications beyond ergodic systems for the tools we establish.
    \item In Section~\ref{sec:quantum-extensions}, we show how many of our results extend to quantum information theory, which we isolate to this section for ease of reading. First, we note that our $f$-divergence Pinsker inequalities trivially lift to arbitrary quantum $f$-divergences (Corollary~\ref{cor:quantum-f-divergence-pinsker}). We then extend our relations between $f$-divergences and the $\chi^{2}$-divergence and its applications to ergodicity to quantum systems in terms of Petz $f$-divergences \cite{Petz-1985a,Petz-1986a}. Unlike in the classical case, the Petz $\chi^{2}$ input-dependent contraction coefficient is not known to be efficient to compute, but our results nonetheless establish this as the key monotonic Petz $f$-divergence for studying time-homogeneous quantum Markov chains.
\end{enumerate}

\subsection{Relation to Previous Work}
A great deal of work has been done on $f$-divergences in both the classical and quantum settings. Here we summarize how our methods and results relate to previous work specifically.  
\begin{enumerate}
    \item \textit{Pinsker Inequalities:} Our basic approach will be to Taylor expand the $f$-divergence to second-order and then bound the expansion. We give a brief history to highlight this idea is not unique and compare it to previous work. The original Pinsker inequality was derived in \cite{Pinsker-1960a}, which was improved to its modern form, $D(\mbf{p}\Vert \mbf{q}) \geq \frac{1}{2}\Vert \mbf{p} - \mbf{q}\Vert^{2}$, independently in the works \cite{Csiszar-1967a,Kullback-1967a,Kemperman-1969}. Higher-order refinements for the KL divergence were obtained by Vajda and Fedotov et al. \cite{Vajda-1970a, Fedotov-2003a}. In \cite{Ordentlich-2005a}, the authors refined Pinsker's inequality in an input-dependent manner using univariate Taylor's theorem. This is in effect the methodology we use in deriving Theorem \ref{thm:pinsker-uni} except that they preserve the input-dependent aspect and only consider KL divergence. The first Pinsker inequalities for other $f$-divergences seem to have been considered by Gilardoni \cite{Gilardoni-2006a,Gilardoni-2010-corrigendum,Gilardoni-2010a}. The first work \cite{Gilardoni-2006a} finds a new integral representation of $f$-divergences such that it is reported in principle it finds the tightest Pinsker inequality for any $f$-divergence, but actually explicitly determining the constant is seemingly unmanageable. The second of his works \cite{Gilardoni-2010a} `simplifies' this to a set of infinite conditions on the function which, when satisfied, guarantee the tightest Pinsker inequality is given by $D_{f}(\mbf{p}\Vert \mbf{q}) \geq f''(1)/2 $. In effect, the idea is to Taylor expand the function $f$ to third-order and then select for the necessary conditions on the function to guarantee $f''(1)/2$ must be the tightest Pinsker inequality coefficient. In this respect, our method is very similar, although it is more flexible by not selecting for $f$ such that the optimal condition is $f''(1)/2$. On the other hand, Reid and Williamson established a method for obtaining Pinsker inequalities for $f$-divergences by approximating the curve of the variational distance, which often times can lead to tight Pinsker inequalities \cite{Reid-2009}. This is a very different method from ours, but also is derived from an integral representation--- though in terms of `statistical informations'  \cite{Liese-2006a}. We believe the method of \cite{Reid-2009} is generally tight as they show that, if one picks appropriate samples for constructing their approximating curve, they derive tight results. However, it does seem their method is restricted to considering only probability distributions unlike our methodology. 
    
    \item \textit{Divergence Inequalities:} Similar to our approach to Pinsker inequalities, our basic methodology is to use Taylor expansion to obtain a simple integral representation for our divergences and then bound our integral representations. While the explicit results we obtain are to the best of our knowledge new, the approach itself is not new and there exist similar results, which we wish to highlight. Obtaining $f$-divergence inequalities by first deriving integral representations is a common methodology. This has been done in terms of statistical information \cite{Gutenbrunner-1992a,Osterreicher-1993a,Liese-2006a}, hockey stick divergence \cite{Sason-2016-f-div-ineqs,Hirche-2023a}, and relative information spectrum \cite{Sason-2018a}. There has also been a long history of obtaining $f$-divergence inequalities via Taylor's expansions specifically by making various assumptions on the relative likelihood between the two measures and the convexity properties of the $f$-divergence, see e.g.~\cite{Barnett-2002a,Csiszar-2004a,Nielsen-2013a,Dragomir-2019a,polyanskiy-2023a}. Our results seem to primarily differ from the inequalities in these works due to our method of bounding the integral representation differing. Perhaps the closest to the method for establishing our Theorem~\ref{thm:f-div-chi-squared-bounds} is the derivation of \cite[Theorem 1]{Pardo-2003a}. The similarity to our methods is perhaps unsurprising as \cite{Pardo-2003a} studies the asymptotics of hypothesis testing as the distributions of the null and alternative hypothesis approach each other, which is what happens to distributions under iterative application of a contractive channel with a unique fixed point. We also note that our work and \cite{Pardo-2003a} are not alone in being interested in the importance of the $\chi^{2}$-divergence, see also \cite{Temme-2010a,Nielsen-2013a,Sason-2016-f-div-ineqs,Sason-2018a,Nishiyama-2020a,Hirche-2023a} and references therein. Lastly, we remark that within the theory of divergence inequalities is also that of reverse Pinsker inequalities, originally introduced for the KL divergence in \cite{Csiszar-2006a}. There exist reverse Pinsker inequalities for $f$-divergences established in \cite{Sason-2016-f-div-ineqs,Binette-2019a}, although these use different assumptions than the ones we derive, and in the quantum case in \cite{Hirche-2023a}.
    \item \textit{Contraction Coefficient Relations and Ergodic Theory:} While not a new topic \cite{Ahlswede-1976a,Cohen-1993a,Choi-1994a,Dobrushin-1970a,Evans-2000a} The study of contraction coefficients has received increased interest as of late given its application in network information theory, fault tolerance, and differential privacy \cite{Duchi-2013a,Polyanskiy-2015a,Raginsky-2016a,Polyanskiy-2017a,Makur-2020a,Ordentlich-2021a,Hirche-2022a,Fawzi-2022a, Hirche-2023a, Hirche-2023-dp, Nuradha-2024b,Cheng-2024a}. Our Theorem~\ref{thm:classical-contraction-rate} generalizes \cite[Proposition 7]{Makur-2020a} from the KL divergence to a large class of $f$-divergences, which follows from our general methods for $f$-divergence inequalities and Pinsker inequalities. With regards to our results on mixing times (Corollary~\ref{cor:Markov-chain-mixing-times} and Proposition~\ref{prop:f-div-mixing-time}), related results have been found previously in the ergodic literature \cite{Diaconis-1991a,Fill-1991a} and a related result for $f$-divergences was found in \cite{Raginsky-2016a}. The major difference of our result to that of \cite{Raginsky-2016a} is that ours is in terms of the $\chi^{2}$ input-dependent contraction coefficient and thus is efficient to compute.
    
    \item \textit{Quantum $f$-Divergences:} The study of divergences in the quantum setting becomes increasingly complex due to the non-commutative aspect of quantum probability theory. This results in a variety of quantum $f$-divergences e.g. \cite{Petz-1985a,Petz-1986a,Petz-1998a,Wilde-2018a,Hirche-2023a}. Contraction coefficients for quantum divergences have been considered in various works, e.g.~\cite{Ruskai-1994a,Lesniewski-1999a,Temme-2010a,Hirche-2022a, Fawzi-2022a, Hirche-2023-dp,Nuradha-2024a,Nuradha-2024b,Cheng-2024a,Belzig-2024a}. To the best of our knowledge, the contraction coefficients for Petz $f$-divergences have only been considered in \cite{Ruskai-1994a, Petz-1998a, Lesniewski-1999a}.\footnote{Technically, the authors consider a quantum extension of $D_{f}(Q \Vert P)$, but since for $\wt{f}(t) \coloneq tf(1/t)$, which also satisfies the conditions of an $f$-divergence, a direct calculation verifies $D_{\wt{f}}(P \Vert Q) = D_{f}(Q \Vert P)$, we don't take this as a significant distinction.} Mixing times measured in terms of trace distance have been considered in \cite{Temme-2010a} albeit under the assumption that the initial state is known. Our mixing times Proposition~\ref{prop:Petz-mixing-times} also consider other measures of dissimilarity under $f$-divergences. To the best of our knowledge, the only other quantum dissimilarity measure that has been considered is the 2-Sandwiched R\'{e}nyi Divergence in \cite{Muller-2018a}. We note that, unlike in the classical case, to the best of our knowledge, neither the results of \cite{Temme-2010a,Muller-2018a} nor ours (Theorem~\ref{thm:Petz-contraction-coefficient-rate} and Proposition~\ref{prop:Petz-mixing-times}) are known to be efficiently computable.
\end{enumerate}

\section{Notation and Background}
Here we briefly summarize notation and relevant facts for the majority of this work. We denote alphabets, $\cX,\cY,$ etc. Nearly this entire work focuses on finite dimensions. We will denote vectors with boldface, i.e. $\mbf{p} \in \mbb{R}^{|\cX|}$. We reference elements of a vector by its element in the finite alphabet, e.g. $p(x)$. We denote the simplex of probability distributions on an alphabet by $\cP(\cX)$. We denote the total variational distance via $\text{TV}(\mbf{p},\mbf{q}) \coloneq \frac{1}{2}\Vert \mbf{p} - \mbf{q}\Vert_{1}$. We denote a (classical) channel $\cW_{\cX \to \cY}$ and equivocate it with its matrix representation $W \in \mbb{R}^{\vert \cY \vert \times \vert \cX \vert}$. We stress to align with quantum information theory standards, the output of a channel, $\mbf{p}_{Y}$, is given by $\mbf{p}_{Y} = W\mbf{p}_{X}$, i.e. the matrix representation of the channel is defined by multiplying the vector on the left.

The two classes of divergences we focus on in this work are the Bregman and $f$-divergences.

\begin{definition}\label{def:BregDiv}
    Let $F:S \to \mbb{R}$ be a continuously differentiable function where $S \subset \mbb{R}^{n}$ is a convex set. Then the Bregman divergence is
    \begin{align}\label{eq:Bregman-Div}
        B_{F}(\mbf{x}||\mbf{y}) := F(\mbf{x}) - F(\mbf{y}) - \langle \grad F (\mbf{y}), \mbf{x} - \mbf{y} \rangle \ . 
    \end{align}
\end{definition}

\begin{definition}\label{def:f-divergence}
    Let $f:(0,+\infty) \to \mbb{R}$ be a convex function with $f(1)=0$. Let $\cX$ be a finite alphabet and $\mbf{p},\mbf{q} \in \mbb{R}^{|\cX|}_{\geq 0}$. The $f$-divergence of $\mbf{p}$ with respect to $\mbf{q}$ is given by
    \begin{align}\label{eq:f-div-def}
        D_{f}(\mbf{p}||\mbf{q}) := 
        \sum_{x \in \cX} q(x)f(p(x)/q(x))  \ ,
    \end{align}
    where we use the standard conventions: $0f(0/0) = 0$, $0f(a/0)= a f'(\infty)$ for $a>0$ where $f'(\infty) := \lim_{x \downarrow 0} xf(1/x)$, and we use the convention $f(0) := f(0^{+})$ as is standard in this definition. \\

    Moreover, if $f:(0+\infty) \to \mbb{R}$ is not guaranteed to satisfy the conditions above, we denote it as
    \begin{align}\label{eq:unrestricted-f-div}
        \cD_{f}(\mbf{p}||\mbf{q}) := \sum_{x \in \cX} q(x)f(p(x)/q(x)) \ , 
    \end{align}
    which may be thought of as an `unrestricted' $f$-divergence.
\end{definition}

A useful property of $f$-divergences over finite alphabets under our conventions is that we can restrict to summing over the support of the reference vector $\mbf{q}$.
\begin{proposition}\label{prop:classical-f-div-restrict-to-q-support}
    Let $\mbf{p},\mbf{q} \in \mbb{R}^{|\cX|}_{\geq 0}$ such that $\mbf{p} \ll \mbf{q}$. Then, by standard convention, without loss of generality one may restrict to the support of $\mbf{q}$ when calculating $\cD_{f}(\mbf{p}\Vert \mbf{q})$, i.e. $\cD_{f}(\mbf{p} \Vert \mbf{q}) = \sum_{x \in \supp(\mbf{q})} q(x) f(p(x)/q(x))$.
\end{proposition}
\begin{proof}
    This simply follows from the convention that $0f(0/0) = 0$ and our assumption $\mbf{p} \ll \mbf{q}$: 
    \begin{align*}
        \cD_{f}(\mbf{p} \Vert \mbf{q}) = \sum_{x \in \supp(\mbf{q})} q(x)f(p(x)/q(x)) + \sum_{x \not \in \supp(\mbf{q})} 0f(0/0) = \sum_{x \in \supp(\mbf{q})} q(x)f(p(x)/q(x)) \ . 
    \end{align*}
\end{proof}

The most important properties of $f$-divergences for this work are the following (see \cite{polyanskiy-2023a} for a more in-depth summary).
\begin{fact}\label{fact:f-div-properties} ~
\begin{enumerate}
    \item (\textit{Data-Processing}) For all (classical) channels $\cW_{X \to Y}$ and distributions $\mbf{p},\mbf{q} \in \cP(\cX)$, $D_{f}(W\mbf{p} \Vert W\mbf{q}) \leq D_{f}(\mbf{p} \Vert \mbf{q})$, i.e. $f$-divergences are monotonic under classical channels.
    \item (\textit{Non-negativity}) For all distributions $\mbf{p},\mbf{q} \in \cP(\cX)$, $D_{f}(\mbf{p}\Vert \mbf{q}) \geq 0$. Moreover, if $f$ is strictly convex at unity, then $D_{f}(\mbf{p} \Vert \mbf{q}) = 0$ if and only if $\mbf{p} = \mbf{q}$. 
    \item (\textit{Invariance of generating functions to $c(t-1)$}) If $f$ induces an $f$-divergence, then for $c \in \mbb{R}$ and $\wt{f}(x) \coloneq f(x) + c(x-1)$, $D_{f}(\mbf{p}\Vert \mbf{q}) = D_{\wt{f}}(\mbf{p}\Vert \mbf{q})$ for all $\mbf{p},\mbf{q} \in \cP(\cX)$.
\end{enumerate}
\end{fact}

The most common $f$-divergence is the KL divergence, which is obtained via $f(t) \coloneq t\log(t)$, i.e., when $\mbf{p} \ll \mbf{q}$, $D_{t\log(t)}(P \Vert Q) = \sum_{x} p(x)\log(p(x)/q(x))$. The most important $f$-divergence for this work will be the $\chi^{2}$ divergence, which is induced by $f(t) \coloneq t^{2}-1$ or $f(t) \coloneq (t-1)^{2}$:
\begin{align}\label{eq:chi-squared-div}
    \chi^{2}(\mbf{p}||\mbf{q}) = \sum_{x \in \cX} \frac{\left(p(x) -q(x) \right)^{2}}{q(x)} \ .
\end{align}
Other examples of $f$-divergences are later enumerated in Tables~\ref{tab:uni} and \ref{tab:multi}.

One critical property of the $f$-divergences is that they all `look' like the $\chi^{2}$-divergence locally or as $\mbf{p}$ approaches $\mbf{q}$. That is, if $\mathscr{P},\mathscr{Q}$ are probability measures, Csiszar and Shields \cite{Csiszar-2004a} showed
\begin{align}\label{eq:csiszar-loc-behavior}
    \lim_{\mathscr{P} \to \mathscr{Q}} \frac{D_{f}(\mathscr{P} \vert \mathscr{Q})}{\chi^{2}(\mathscr{P}\Vert \mathscr{Q})} \to \frac{f''(1)}{2} \ ,  \ 
\end{align}
and Sason \cite{Sason-2018a} showed for an $f$-divergence as defined in Definition~\ref{def:f-divergence} that also satisfies $f''$ is continuous at $1$, whenever probability measures $\text{ess} \sup \frac{d\mathscr{P}}{d\mathscr{Q}} < +\infty$,
\begin{align}
    \lim_{\lambda \downarrow 0} \frac{1}{\lambda^{2}} D_{f}(\lambda \mathscr{P} + (1-\lambda)\mathscr{Q} \Vert \mathscr{Q}) = \frac{f''(1)}{2}\chi^{2}(\mathscr{P}\Vert \mathscr{Q}) \ , 
\end{align}
where we omit the formal definitions of $f$-divergences for probability measures as we only reference them in these two previous results and in the appendix. This shows that $\chi^{2}$-divergence is in some sense particularly central and is what much of our work is further establishing.

\section{\texorpdfstring{$f$}{f}-Divergence Pinsker Inequalities}\label{sec:f-div-Pinsker}

In this section, we present a method for obtaining input-independent Pinsker inequalities for twice-continuously-differentiable $f$-divergences.
We first derive two special cases that follow from univariate Taylor's theorem, and subsequently generalize the procedure and obtain tighter bounds via multivariate Taylor's theorem.
We showcase several derived results, which in some cases (e.g. Kullback-Leibler and squared Hellinger divergences) match known sharp inequalities.
Our method is more general than the previous most general method for input-independent Pinsker inequalities \cite{Gilardoni-2010a}.

\subsection{\texorpdfstring{$f$}{f}-Divergence Pinsker Inequalities from Univariate Taylor's Theorem}

We now present a means of lower-bounding $f$-divergences by total variation distance via Taylor's theorem, compare to an existing method, and derive several consequent bounds summarized in Table~\ref{tab:uni}.
This basic methodology was effectively previously used to establish Pinsker's inequality in \cite{polyanskiy-2023a} as well as an input-dependent refinement in \cite{Ordentlich-2005a}, but was not considered more generally.

Recall the univariate Taylor's theorem with integral remainder (see, e.g., \cite{Apostol-1967a}).
\begin{lemma}[Univariate Taylor's theorem]\label{lem:univariate-Taylor}
    Let $f: [a,b] \to \mbb{R}$ such that $f \in C^{2}([a,b])$. Then
    $$ f(b) = f(a) + f'(a)(b-a) + \int_{a}^{b} f''(t)(b-t) dt \ . $$
\end{lemma}

We will also make use of the following identity.
\begin{proposition}\label{prop:tv}
    Let $c > 0$ and define
    \begin{equation}\label{eq:same-norm-binary-vectors}
        \mbf{p} = \begin{bmatrix} p \\ c - p \end{bmatrix}\ , \quad \mbf{q} = \begin{bmatrix} q \\ c -q \end{bmatrix} \ ,
    \end{equation}
    for some $p, q \in [0, c]$.
    Then,
    \begin{equation*}
        \mrm{TV}(\mbf{p},\mbf{q})^{2} = (p-q)^{2} \ .
    \end{equation*}
\end{proposition}
\begin{proof}
    We have
    \begin{align*}
        \mrm{TV}(\mbf{p},\mbf{q}) = \frac{1}{2}\left[|p-q| + |(c-p)-(c-q)| \right] = |p-q| \ .
    \end{align*}
    Squaring both sides completes the proof.
\end{proof}

We are now in a position to establish our Pinsker inequalities.
\begin{theorem}\label{thm:pinsker-uni}
    Let $f:(0,+\infty) \to \mbb{R}$ with $f(1) = 0$ be convex and twice continuously differentiable. Suppose $L_{f}$ satisfies one of the following conditions
    \begin{align}
        L_f &\leq \frac{1}{y}\ f''\!\left(\frac{x}{y}\right) + \frac{1}{1-y}\ f''\!\left(\frac{1-x}{1-y}\right)\hspace*{-7.5em} &\forall\, x,y \in (0,1)\ ,\label{eq:lambda0}\\
        L_f &\leq \frac{x^2}{y^3} f''\!\left(\frac{x}{y}\right) + \frac{(1-x)^2}{(1-y)^3} f''\!\left(\frac{1-x}{1-y}\right)\hspace*{-7.5em} &\forall\, x,y \in (0,1)\label{eq:lambda1}\ .
    \end{align}
    Then, for all $\mbf{p},\mbf{q} \geq 0$ such that $\Vert\mbf{p}\Vert_{1} = \Vert\mbf{q}\Vert_{1} = c > 0$, we have
    \begin{equation*}
        D_{f}(\mbf{p}\Vert \mbf{q}) \geq \frac{L_f}{2c} \mrm{TV}(\mbf{p},\mbf{q})^2 \ .
    \end{equation*}
\end{theorem}
\begin{proof}

    We first prove Theorem~\ref{thm:pinsker-uni} for two-dimensional non-negative vectors and later provide an argument to extend it to any finite dimension.
    
    We restrict $\mbf{p}$ and $\mbf{q}$ to two-dimensional vectors of the form \eqref{eq:same-norm-binary-vectors} and define a mapping $g$ from $p$ and $q$ to the given binary $f$-divergence as
    \begin{align}\label{eq:binary-f-div-mapping}
        g(p, q) = q f\left(\frac{p}{q}\right) + (c-q) f\left(\frac{c-p}{c-q}\right) = D_{f}(\mbf{p}||\mbf{q}).
    \end{align}
    If $g$ is undefined at the boundary of $[0,c]$, we define it there via its limits, which exist since $f$ is continuous.
    We can also consider a mapping $g_q : p \mapsto g(p, q)$ based on a fixed choice of $q$.
    We determine its first derivative via the chain rule.
    \begin{equation*}
        \frac{d}{dp} g_{q}(p)
        = f'\left(\frac{p}{q}\right) - f'\left(\frac{c-p}{c-q}\right)\ .
    \end{equation*}
    Thus, $g_{q}(q) = c f(1) = 0$ and $\frac{d}{dp} g_{q}(p) \big \vert_{p=q} = f'(1) - f'(1) = 0$.
    Since $g_{q} \in C^2([0, c])$ and $p,q \in [0,c]$,
    \begin{equation}\label{eq:binary-f-div-int-rep-p}
        D_{f}(\mbf{p}||\mbf{q})
        = \int^{p}_{q} \frac{d^{2}}{dt^{2}} g_{q}(t) (p-t) \, dt
    \end{equation}
    by applying Lemma~\ref{lem:univariate-Taylor} with $f(t) = g_{q}(t)$, $a = q$ and $b = p$.
    We compute the second derivative of $g_{q}$ as
    \begin{equation*}
        \frac{d^{2}}{dp^{2}} g_{q}(p) = \frac{1}{q} f''\left(\frac{p}{q}\right) + \frac{1}{c-q} f''\left(\frac{c-p}{c-q}\right)\ .
    \end{equation*}
    Let $p = cx$ and $q = cy$ with $x, y \in [0, 1]$, we then additionally have
    \begin{equation*}
        c\ \frac{d^{2}}{dp^{2}} g_{q}(p) = \frac{1}{y} f''\left(\frac{x}{y}\right) + \frac{1}{1-y} f''\left(\frac{1-x}{1-y}\right)\ .
    \end{equation*}
    Supposing the bound \eqref{eq:lambda0} holds on $(0,1)$, we know it must hold on $[0,1]$ by the continuity of $f''$.
    It may be substituted into the integral representation \eqref{eq:binary-f-div-int-rep-p} to obtain
    \begin{equation*}
        D_{f}(\mbf{p} || \mbf{q})
        = \int_q^p \frac{d^2}{dt^2} g_q(t) (p-t) dt
        \geq \frac{L_{f}}{c} \int_q^p (p-t) dt
        = \frac{L_{f}}{2c}(p - q)^2
        = \frac{L_{f}}{2c} \mathrm{TV}(\mbf{p},\mbf{q})^2\ ,
    \end{equation*}
    where the identity $\mrm{TV}(\mbf{p},\mbf{q})^2 = (p-q)^2$ is by Proposition~\ref{prop:tv}.
    
    Repeat the above procedure, differentiating with respect to $q$.
    Fix $p$ in \eqref{eq:binary-f-div-mapping} to obtain $g_p : q \mapsto g(p,q)$.
    We determine the first derivative of this mapping via the chain rule and product rule.
    \begin{equation*}
        \frac{d}{dq} g_{p}(q)
        = f\left(\frac{p}{q}\right) - \frac{p}{q} f'\left(\frac{p}{q}\right) - f\left(\frac{c-p}{c-q}\right) + \frac{c-p}{c-q} f'\left(\frac{c-p}{c-q}\right)\ .
    \end{equation*}
    Thus, $g_{p}(p) = c f(1) = 0$ and $\frac{d}{dq} g_{p}(q) \big \vert_{q=p} = f(1) - f'(1) - f(1) + f'(1) = 0$.
    Since $g_p \in C^2([0,c])$ and $p,q \in [0,c]$,
    \begin{equation}\label{eq:binary-f-div-int-rep-q}
        D_{f}(\mbf{p}||\mbf{q})
        = \int^{q}_{p} \frac{d^{2}}{dt^{2}} g_{p}(t) (q-t) \, dt
    \end{equation}
    by Lemma~\ref{lem:univariate-Taylor}.
    We compute the second derivative of $g_p$ as
    \begin{equation*}
        \frac{d^{2}}{dq^{2}} g_{p}(q)
        = \frac{p^2}{q^3} f''\left(\frac{p}{q}\right) + \frac{(c-p)^2}{(c-q)^3} f''\left(\frac{c-p}{c-q}\right)\ .
    \end{equation*}
    Let $p = cx$ and $q = cy$ with $x, y \in [0, 1]$, we then additionally have
    \begin{equation*}
        c\ \frac{d^{2}}{dq^{2}} g_{p}(q) = \frac{x^2}{y^3} f''\left(\frac{x}{y}\right) + \frac{(1-x)^2}{(1-y)^3} f''\left(\frac{1-x}{1-y}\right)\ .
    \end{equation*}
    Supposing the bound \eqref{eq:lambda1} holds on $(0,1)$, we know it must hold on $[0,1]$ by the continuity of $f''$.
    It may be substituted into the integral representation \eqref{eq:binary-f-div-int-rep-q} to obtain
    \begin{equation*}
        D_{f}(\mbf{p} || \mbf{q})
        = \int_p^q \frac{d^2}{dt^2} g_p(t) (q-t) dt
        \geq \frac{L_{f}}{c} \int_p^q (q-t) dt
        = \frac{L_{f}}{2c}(p - q)^2
        = \frac{L_{f}}{2c} \mathrm{TV}(\mbf{p},\mbf{q})^2\ ,
    \end{equation*}
    where the identity $\mrm{TV}(\mbf{p},\mbf{q})^2 = (p-q)^2$ is by Proposition~\ref{prop:tv}.

    We now argue that Theorem~\ref{thm:pinsker-uni} holds any finite dimension vectors $\mbf{p}$ and $\mbf{q}$.
    Let $\cW$ be a classical-to-classical sum-preserving channel that maps $\mbf{p}$ and $\mbf{q}$ to two-dimensional vectors $\cW(\mbf{p})$ and $\cW(\mbf{q})$ such that $\mrm{TV}(\mbf{p},\mbf{q}) = \mrm{TV}(\cW(\mbf{p}),\cW(\mbf{q}))$. It is well-known that such a channel always exists.     
    Now under the application of $\cW$, $D_{f}(\mbf{p}\Vert \mbf{q})$ is lower bounded by $D_{f}(\cW(\mbf{p})\Vert \cW(\mbf{q}))$ (due to data-processing), and furthermore, as Theorem~\ref{thm:pinsker-uni} holds for two-dimensional vectors, $D_{f}(\cW(\mbf{p})\Vert \cW(\mbf{q}))$ is lower bounded by $\frac{L_f}{2c} \mrm{TV}(\cW(\mbf{p}),\cW(\mbf{q}))^2$, or equivalently to $\frac{L_f}{2c} \mrm{TV}(\mbf{p},\mbf{q})^2$ (due to the nature of the channel $\cW$), implying an overall lower bound of $\frac{L_f}{2c} \mrm{TV}(\mbf{p},\mbf{q})^2$ on $D_{f}(\mbf{p}\Vert \mbf{q})$.
\end{proof}

We first note that the above result only appeals to $f$ being twice continuously differentiable, and thus applies to any $f$-divergence induced by a twice continuously differentiable $f$.
This, however, is a rather standard property of $f$-divergences to consider -- for example, they are the set of $f$-divergences that admit an integral representation in terms of Hockey stick divergences \cite{Sason-2016-f-div-ineqs,Hirche-2023-dp}.
We also note that an alternate approach involves recovering similar inequalities via strong convexity.

In finite dimensions, Theorem \ref{thm:pinsker-uni} is more manageable than the previous most general case \cite{Gilardoni-2010a}, which we state here for comparison.
\begin{proposition}\label{prop:Gilardoni}
    \cite[Theorem 3]{Gilardoni-2010a} Suppose $f$ is convex, thrice differentiable at unity with $f''(1) > 0$, and
    \begin{align}\label{eq:Gilardoni-pinsker-requirement}
         (f(t)-f'(1)(t-1))\left[1-\frac{1}{3}\frac{f'''(1)}{f''(1)}(t-1) \right] \geq \frac{f''(1)}{2}(t-1)^{2} \ \quad \forall t \in (0,+\infty)\ .
    \end{align}
    Then, $D_{f}(\mbf{p}\Vert\mbf{q}) \geq \frac{f''(1)}{2}\mrm{TV}(\mbf{p},\mbf{q})^{2}$ where $\mbf{p},\mbf{q}$ are probability measures.
    Moreover, this bound is optimal.
\end{proposition}

To see this proposition is less general than Theorem~\ref{thm:pinsker-uni} in finite dimensions even when only considering probability mass functions, we may verify that \eqref{eq:Gilardoni-pinsker-requirement} is not satisfied for R\'{e}nyi's information gain of order $\alpha \not\in [-1,2]$.
Indeed, \cite{Gilardoni-2010a} only established a Pinsker inequality in the case of $\alpha \in [-1,2]$ for this divergence.
On the contrary, we establish non-trivial Pinsker inequalities for R\'{e}nyi's information gain for all values of $\alpha$ in Propositions~\ref{prop:pinsker-renyi-uni} and~\ref{prop:pinsker-renyi-multi}.

Moreover, our method is more general than any other that requires thrice-differentiability.
Consider the example below, in which we obtain a non-trivial lower bound on an $f$-divergence generated by a function that is not thrice differentiable.
\begin{proposition}\label{prop:example}
    Let $f : (0, +\infty) \to \mbb{R}$ be defined as
    \begin{equation*}
        f(t) = \begin{cases}
            \frac{1}{2} t (t-1) & \text{if}\ t \leq 1\ ,\\
            t \ln t - \frac{1}{2}(t-1) & \text{otherwise}\ .
        \end{cases}
    \end{equation*}
    Note that $f$ is twice continuously differentiable but not thrice differentiable, has $f(1) = 0$ and $f(0) < +\infty$, and $\frac{f(t)}{t}$ and $\frac{f(t) - f(0)}{t}$ are concave.
    Additionally, $f$ satisfies \eqref{eq:lambda0} with $L_f = 2$.
\end{proposition}

The rest of this subsection provides example applications of Theorem~\ref{thm:pinsker-uni}.
For brevity, some related proofs (including that of Proposition~\ref{prop:example} above) are postponed to Appendix~\ref{appx:bounds}.
Results are summarized in Table~\ref{tab:uni}.

As a first step, we recover Pinsker's inequality via either KL divergence or reverse KL divergence.
\begin{proposition}[KL-divergence]\label{prop:pinsker-KL}
    Let $f(t) = t \ln t$.
    Then, $L_f = 4$ satisfies \eqref{eq:lambda0}.
\end{proposition}
\begin{proof}
    We have $f''(t) = t^{-1}$.
    Thus,
    $$ 4 \leq \frac{1}{y} \frac{y}{x} + \frac{1}{1-y} \frac{1-y}{1-x} = \frac{1}{x} + \frac{1}{1-x}\ ,$$
    where equality is attained at $x = 1/2$.
\end{proof}

\begin{proposition}[Reverse KL-divergence]\label{prop:pinsker-reverseKL}
    Let $f(t) = -\ln t$.
    Then, $L_f = 4$ satisfies \eqref{eq:lambda1}.
\end{proposition}
\begin{proof}
    We have $f''(t) = t^{-2}$.
    Thus,
    $$ 4 \leq \frac{x^2}{y^3} \frac{y^2}{x^2} + \frac{(1-x)^2}{(1-y)^3} \frac{(1-y)^2}{(1-x)^2} = \frac{1}{y} + \frac{1}{1-y}\ ,$$
    where equality is attained at $y = 1/2$.
\end{proof}

With $\mbf{p},\mbf{q} \in \mbb{R}^{\cX}_{\geq 0}$ such that $\Vert\mbf{p}\Vert_{1} = \Vert\mbf{q}\Vert_{1} = c$, Propositions~\ref{prop:pinsker-KL} and~\ref{prop:pinsker-reverseKL} both imply
\begin{equation*}
    D(\mbf{p}||\mbf{q}) \geq \frac{2}{c \ln(2)}\mrm{TV}(\mbf{p},\mbf{q})^{2} \ ,
\end{equation*}
where $\ln(2)$ accounts for the conversion to base $2$.
In the case $c=1$, this recovers Pinsker's inequality.

We next consider R\'{e}nyi's information gain of order $\alpha$, also known as Vajda's relative information.\footnote{We note that our definition has an extra term $-\frac{\alpha}{\alpha-1}(t-1)$ from the definition in other works, e.g. \cite{Gilardoni-2010a}. Note this is irrelevant for the $f$-divergence by Item 3 of Fact \ref{fact:f-div-properties}.}
\begin{proposition}[R\'{e}nyi's information gain]\label{prop:pinsker-renyi-uni}
    Let $\alpha \in \mbb{R}$ and define $f_{\alpha} : (0, +\infty) \to \mbb{R}$ as
    \begin{equation}\label{eq:fa}
        f_{\alpha}(t) = \begin{cases}
            -\ln t + t - 1 & \text{if}\ \alpha = 0\ ,\\
            t \ln t - t + 1 & \text{if}\ \alpha = 1\ ,\\
            \frac{t^\alpha - 1 - \alpha(t-1)}{\alpha (\alpha-1)} & \text{otherwise}\ .
        \end{cases}
    \end{equation}
    Then, $L_{f_\alpha} = 1$ satisfies \eqref{eq:lambda0}.
    Moreover, $L_{f_\alpha} = 4$ satisfies \eqref{eq:lambda1} when $\alpha \in [-1, 0]$ and \eqref{eq:lambda0} when $\alpha \in [1, 2]$.
\end{proposition}
In the above, $L_{f_\alpha} = 1$ similarly satisfies \eqref{eq:lambda1}.
The inequalities for KL divergence in Propositions~\ref{prop:pinsker-KL} and~\ref{prop:pinsker-reverseKL} are equivalent to $L_{f_\alpha} = 4$ with $\alpha \in \{1,0\}$ in Proposition~\ref{prop:pinsker-renyi-uni}, and may alternatively be viewed as direct corollaries of it.
Similarly, $\chi^2$-divergence is a special case of R\'{e}nyi's information gain with $\alpha \in \{-1, 2\}$.

\begin{corollary}[Hellinger divergence]\label{cor:pinsker-hellinger-div-uni}
    Let $\alpha \in (0, +\infty)$ and define $f_{\alpha} : (0, +\infty) \to \mbb{R}$ as
    \begin{equation}\label{eq:fah}
        f_{\alpha}(t) = \begin{cases}
            t \ln t & \text{if}\ \alpha = 1\ ,\\
            \frac{t^\alpha - 1 }{\alpha - 1} & \text{otherwise}\ .
        \end{cases}
    \end{equation}
    Then, \eqref{eq:lambda0} is satisfied by $L_{f_\alpha} = \alpha$ and $L_{f_\alpha} = 4\alpha$ when $\alpha \in [1, 2]$.
\end{corollary}
\begin{proof}
    This follows from Proposition~\ref{prop:pinsker-renyi-uni} by multiplying by $\alpha$ and noting the extra term $\alpha(t-1)$ is irrelevant for the $f$-divergence by Item 3 of Fact \ref{fact:f-div-properties}.
\end{proof}

\begin{corollary}[Pearson $\chi^2$-divergence]\label{cor:pinsker-chisq}
    Let $f(t) = t^2 - 1$.
    Then, $L_f = 8$ satisfies \eqref{eq:lambda0}.
\end{corollary}
\begin{proof}
    This follows from Proposition~\ref{prop:pinsker-renyi-uni} by setting $\alpha = 2$ and multiplying by $2$.
\end{proof}

\begin{corollary}[Neyman $\chi^2$-divergence]\label{cor:pinsker-neyman}
    Let $f(t) = t^{-1} - 1$.
    Then, $L_f = 8$ satisfies \eqref{eq:lambda1}.
\end{corollary}
\begin{proof}
    This follows from Proposition~\ref{prop:pinsker-renyi-uni} by setting $\alpha = -1$ and multiplying by $2$.
\end{proof}

Corollaries~\ref{cor:pinsker-chisq} and~\ref{cor:pinsker-neyman} recover the input independent lower bound in \cite[Proposition 7.15]{polyanskiy-2023a} when $c = 1$.
\begin{equation*}
    \chi^{2}(\mbf{p}\Vert \mbf{q}) \geq \frac{4}{c}\, \mrm{TV}(\mbf{p},\mbf{q})^{2} \ .
\end{equation*}

We conclude this subsection with two more examples where Theorem~\ref{thm:pinsker-uni} is applied directly.
\begin{proposition}[Symmetric $\chi^2$-divergence]\label{prop:pinsker-symchisq}
    Let $f(t) = \frac{(t-1)^2(t+1)}{t}$.
    Then, $L_f = 16$ satisfies \eqref{eq:lambda0}.
\end{proposition}
Alternatively, Proposition~\ref{prop:pinsker-symchisq} above may be recovered from Corollaries~\ref{cor:pinsker-chisq} and~\ref{cor:pinsker-neyman}.
\begin{proposition}[Arithmetic-geometric mean]\label{prop:pinsker-mean}
    Let $f(t) = \left( \frac{t+1}{2} \right) \ln \left( \frac{t+1}{2 \sqrt{t}} \right)$.
    Then, $L_f = 1$ satisfies \eqref{eq:lambda0}.
\end{proposition}
In both Propositions~\ref{prop:pinsker-symchisq} and \ref{prop:pinsker-mean}, we have that \eqref{eq:lambda1} is additionally satisfied via a similar proof.

\begin{table}[H]
    \centering
    \bgroup
    \def\arraystretch{1.5}
    \begin{tabular}{rlrll}
        Divergence & $f(t)$ & $L_f$ & Conditions & Reference \\
        \hline
        KL-divergence & $t \ln t$ & $4$ && Proposition~\ref{prop:pinsker-KL} \\
        Reverse KL-divergence & $-\ln t$ & $4$ && Proposition~\ref{prop:pinsker-reverseKL}\\
        \multirow{2}{*}{R\'{e}nyi's information gain} & \multirow{2}{*}{$\frac{t^\alpha-1}{\alpha(\alpha-1)}$} & $4$ & $\alpha \in [-1, 0] \cup [1, 2]$ & \multirow{2}{*}{Proposition~\ref{prop:pinsker-renyi-uni}}\\
        && $1$ & otherwise &\\
        \multirow{2}{*}{Hellinger divergence} & \multirow{2}{*}{$\frac{t^\alpha-1}{\alpha-1}$} & $4\alpha$ & $\alpha \in [1, 2]$ & \multirow{2}{*}{Corollary~\ref{cor:pinsker-hellinger-div-uni}}\\
        && $\alpha$ & $\alpha \in (0,1) \cup (2, +\infty)$ &\\
        Pearson $\chi^2$-divergence & $t^2 - 1$ & $8$ && Corollary~\ref{cor:pinsker-chisq}\\
        Neyman $\chi^2$-divergence & $\frac{1}{t} - 1$ & $8$ && Corollary~\ref{cor:pinsker-neyman}\\
        Symmetric $\chi^2$-divergence & $\frac{(t-1)^2(t+1)}{t}$ & $16$ && Proposition~\ref{prop:pinsker-symchisq}\\
        Arithmetic-geometric mean & $\left( \frac{t+1}{2} \right) \ln \left( \frac{t+1}{2 \sqrt{t}} \right)$ & $1$ && Proposition~\ref{prop:pinsker-mean}\\
    \end{tabular}
    \egroup\vspace*{1em}
    \caption{Summary of bounds obtained by Theorem~\ref{thm:pinsker-uni}}
    \label{tab:uni}
\end{table}

\subsection{\texorpdfstring{$f$}{f}-Divergence Pinsker Inequalities from Multivariate Taylor's Theorem}

Despite the multitude of successful applications of Theorem~\ref{thm:pinsker-uni} showcased in Table~\ref{tab:uni}, there exist divergences for which it fails to obtain the tightest bound.
One simple example is Jeffrey's divergence, generated by $f(t) = (t-1)\ln t$.
This is a sum of KL and reverse KL divergences and as such its lower bound should be the sum of their bounds.
On the contrary, $L_f = 8$ satisfies neither \eqref{eq:lambda0} nor \eqref{eq:lambda1}.

To obtain tighter bounds in such cases, we generalize Theorem~\ref{thm:pinsker-uni} to accept a larger parametrized family of conditions, of which \eqref{eq:lambda0} and \eqref{eq:lambda1} are the extreme points.
In particular, we note that the proof of \ref{thm:pinsker-uni} considered differentiation in the direction along a single variable, whereas we can instead traverse in an arbitrary direction via multivariate Taylor's theorem.
We present the theorem and its proof, several example applications, and a discussion of the limitations of this method.

We state the first-order multivariate Taylor's theorem with the integral remainder form, which is all we will need. This form of multivariate Taylor's theorem may be found in \cite{Marsden-2003a} and, up to a change in parameterization, the more general case may be found in \cite{Hormander-2003a}.
\begin{lemma}\label{lem:first-order-Taylor's}
    Let $f: \mbb{R}^{n} \to \mbb{R}$ be $C^{2}$ on an open convex set $S$. If $\mbf{a} \in S$ and $\mbf{a} + \mbf{h} \in S$ then 
    \begin{align*}
        f(\mbf{a}+\mbf{h}) = f(\mbf{a}) + \sum_{i \in [n]} h_{i} \frac{\partial f(\mbf{a})}{x_{i}} + R_{1}(\mbf{a},\mbf{h}) \ , 
    \end{align*}
    where the remainder in integral form is
    $$ R_{1}(\mbf{a},\mbf{h}) = \sum_{i,j \in [n]} \int_{0}^{1} (1-t) h_{i}h_{j} \frac{\partial^{2}f(\mbf{a}+t\mbf{h})}{\partial x_{i} \partial x_{j}}  \, dt \ . $$
\end{lemma}
The previous lemma may be easier expressed in terms of the Hessian matrix of $f$, $H_{f(t)} := \sum_{i,j} \frac{\partial^{2} f(t)}{\partial x_{i}\partial x_{j}} E_{i,j}$. 
\begin{corollary}\label{cor:first-order-taylor}
     Let $F: \mbb{R}^{n} \to \mbb{R}$ be $C^{2}$ on an open convex set $S$. If $\mbf{a} \in S$ and $\mbf{a} + \mbf{h} \in S$ then 
     \begin{align*}
         & F(\mbf{a}+\mbf{h}) = F(\mbf{a}) + \langle \grad F (\mbf{a}), \mbf{h}\rangle  + \int^{1}_{0} (1-t) \mbf{h}^{T} H_{F}\vert_{\mbf{a}+t\mbf{h}} \mbf{h} \, dt \ ,
     \end{align*}
     where $x^{\odot n}$ applies power to the $n$ in an entry-wise fashion.
\end{corollary}

We are now in a position to generalize our Pinsker inequalities.
\begin{theorem}\label{thm:pinsker-multi}
    Let $f:(0,+\infty) \to \mbb{R}$ with $f(1) = 0$ be convex and twice continuously differentiable.
    Suppose there exists $\lambda \in [0, 1]$ such that $L_{f}$ satisfies for all $x,y \in (0, 1)$
    \begin{align}\label{eq:lambda}
        L_f &\leq \left[ (1-\lambda) + \lambda \frac{x}{y} \right]^2 \frac{1}{y} f''\left( \frac{x}{y} \right) + \left[ (1-\lambda) + \lambda \frac{1-x}{1-y} \right]^2 \frac{1}{1-y} f''\left( \frac{1-x}{1-y} \right) \ .
    \end{align}
    Then for all $\mbf{p},\mbf{q} \geq 0$ such that $\Vert\mbf{p}\Vert_{1} = \Vert\mbf{q}\Vert_{1} = c > 0$, we have
    \begin{equation}
        D_{f}(\mbf{p}\Vert \mbf{q}) \geq \frac{L_f}{2c} \mrm{TV}(\mbf{p},\mbf{q})^2 \ .
    \end{equation}
\end{theorem}

\begin{proof}
    We show Theorem~\ref{thm:pinsker-multi} for two-dimensional vectors and use the same argument discussed in the proof of Theorem~\ref{thm:pinsker-uni} to extend to it any finite dimensional vectors $\mbf{p}$ and $\mbf{q}$.
    
    We define a mapping $g$ from $p$ and $q$ to the given binary $f$-divergence.
    \begin{equation*}
        g(p,q) = q f\left(\frac{p}{q}\right) + (c-q) f\left(\frac{c-p}{c-q}\right) = D_{f}(\mbf{p}||\mbf{q})\ .
    \end{equation*}
    If $g$ is undefined at the boundary of $p,q \in [0,c]$, we define it there via its limits, which exist since $f$ is continuous.
    We may compute its gradient and Hessian as follows.
    \begin{align*}
        \grad g = \begin{bmatrix} \frac{\partial g}{\partial p} \\ \frac{\partial g}{\partial q} \end{bmatrix} &=
        \begin{bmatrix}
            f'\!\left(\frac{p}{q}\right) - f'\!\left(\frac{c-p}{c-q}\right) \\
            f\!\left(\frac{p}{q}\right) - \frac{p}{q}f'\!\left(\frac{p}{q}\right) - f\!\left(\frac{c-p}{c-q}\right) + \frac{c-p}{c-q}f'\!\left(\frac{c-p}{c-q}\right)
        \end{bmatrix} \\
        H_g = \begin{bmatrix} \frac{\partial^2 g}{\partial p^2} \!&\! \frac{\partial^2 g}{\partial p \partial q} \\ \frac{\partial^2 g}{\partial q \partial p} \!&\! \frac{\partial^2 g}{\partial q^2}\end{bmatrix} &=
        \begin{bmatrix}
            \frac{1}{q}f''\!\left(\frac{p}{q}\right) + \frac{1}{c-q}f''\!\left(\frac{c-p}{c-q}\right) &
            \frac{-p}{q^2}f''\!\left(\frac{p}{q}\right) - \frac{c-p}{(c-q)^2}f''\!\left(\frac{c-p}{c-q}\right) \\
            \frac{-p}{q^2}f''\!\left(\frac{p}{q}\right) - \frac{c-p}{(c-q)^2}f''\!\left(\frac{c-p}{c-q}\right) &
            \frac{p^2}{q^3}f''\!\left(\frac{p}{q}\right) + \frac{(c-p)^2}{(c-q)^3}f''\!\left(\frac{c-p}{c-q}\right)
        \end{bmatrix}
    \end{align*}
    For an arbitrary $\lambda \in [0, 1]$, define the vectors
    \begin{align*}
        \mbf{a} &= \left( (1-\lambda) q + \lambda p \right) \begin{bmatrix} 1 \\ 1 \end{bmatrix} \ ,\hspace*{-7.5em}&
        \mbf{h} &= (p-q) \begin{bmatrix} 1-\lambda \\ -\lambda \end{bmatrix}\ ,
    \end{align*}
    which ensures that both elements of $\mbf{a} + t\mbf{h}$ are between $p$ and $q$, and hence $g(\mbf{a} + t\mbf{h})$ is well-defined for $t \in [0, 1]$.
    Note also that $g(\mbf{a} + \mbf{h}) = D_{f}(\mbf{p}\Vert\mbf{q})$.
    Since $g \in C^2([0,c]^2)$ and $p,q \in [0,c]$, this lets us express $f$-divergence using multivariate Taylor's theorem (Corollary~\ref{cor:first-order-taylor}).
    \begin{equation*}
        D_{f}(\mbf{p}\Vert\mbf{q}) = (p-q)^2 \int_0^1 (1-t) \left[ (1-\lambda)^2 \frac{\partial^2 g}{\partial p^2} - 2(1-\lambda)\lambda \frac{\partial^2 g}{\partial p \partial q} + \lambda^2 \frac{\partial^2 g}{\partial q^2} \right]_{\mbf{a}+t\mbf{h}} dt\ ,
    \end{equation*}
    where $g(\mbf{a}) = c f(1) = 0$ and $\grad g(\mbf{a}) = 0$.
    We can write the mapping inside the square brackets as
    \begin{align}
        h_\lambda : \begin{bmatrix} p \\ q \end{bmatrix} \mapsto &\left[(1-\lambda)^2 + 2(1-\lambda)\lambda \frac{p}{q} + \lambda^2 \frac{p^2}{q^2} \right] \frac{1}{q}f''\!\left(\frac{p}{q}\right) \nonumber\\
        + &\left[ (1-\lambda)^2 + 2(1-\lambda)\lambda \frac{c-p}{c-q} + \lambda^2 \frac{(c-p)^2}{(c-q)^2} \right] \frac{1}{c-q}f''\!\left(\frac{c-p}{c-q}\right) \nonumber\\
        = &\left[(1-\lambda) + \lambda \frac{p}{q} \right]^2 \frac{1}{q}f''\!\left(\frac{p}{q}\right)
        + \left[ (1-\lambda) + \lambda \frac{c-p}{c-q} \right]^2 \frac{1}{c-q}f''\!\left(\frac{c-p}{c-q}\right)\ .\label{eq:binary-f-div-int-rep-multi}
    \end{align}
    Let $p = cx$ and $q = cy$ with $x, y \in [0, 1]$, the same mapping may now be written as
    \begin{equation*}
        h : \begin{bmatrix} p \\ q \end{bmatrix}
        \mapsto \left[(1-\lambda) + \lambda \frac{x}{y} \right]^2 \frac{1}{cy}f''\!\left(\frac{p}{q}\right)
        + \left[ (1-\lambda) + \lambda \frac{1-x}{1-y} \right]^2 \frac{1}{c(1-y)}f''\!\left(\frac{1-x}{1-y}\right)\ .
    \end{equation*}
    This is almost identical to the right-hand side of condition \eqref{eq:lambda}, which when satisfied guarantees
    \begin{equation*}
        \frac{L_f}{c} \leq h \left(\begin{bmatrix} p \\ q \end{bmatrix} \right) \qquad \forall\, p,q \in [0,c]\ .
    \end{equation*}
    due to the continuity of $f''$.
    Note that both elements of $\mbf{a} + t\mbf{h}$ are in $[0, c]$ for all $t \in [0, 1]$, and this condition holds.
    We can substitute this mapping back into the integral \eqref{eq:binary-f-div-int-rep-multi} to obtain the desired bound.
    \begin{equation*}
        D_{f}(\mbf{p}\Vert\mbf{q})
        = (p-q)^2 \int_0^1 (1-t)\, h(\mbf{a} + t\mbf{h})\, dt
        \geq \frac{L_{f}}{c} (p-q)^2 \int_0^1 (1-t) dt
        = \frac{L_{f}}{2c} (p-q)^2\ .
    \end{equation*}
    Using the identity $\mrm{TV}(\mbf{p},\mbf{q})^2 = (p-q)^2$ in Proposition~\ref{prop:tv} completes the proof.
\end{proof}

We first note that Theorem~\ref{thm:pinsker-multi} is a strict generalization of Theorem~\ref{thm:pinsker-uni}.
Choosing $\lambda \in \{0, 1\}$ recovers \eqref{eq:lambda0} and \eqref{eq:lambda1}, yet \eqref{eq:lambda} may alternatively be satisfied for other values of $\lambda \in (0, 1)$.
For instance, $\lambda = 1/2$ yields
\begin{equation}\label{eq:lambda12}
    L_f \leq \left[ 1 + \frac{x}{y} \right]^2 \frac{1}{4 y} f''\left( \frac{x}{y} \right) + \left[ 1 + \frac{1-x}{1-y} \right]^2 \frac{1}{4 (1-y)} f''\left( \frac{1-x}{1-y} \right) \ =: \, h_{1/2}(x,y)\ .
\end{equation}
This value of $\lambda$ typically suffices where Theorem~\ref{thm:pinsker-uni} fails.
We may group divergences based on which values of $\lambda$ satisfy their tightest bound.
\begin{itemize}
    \item Each of $0, 1/2, 1$ (symmetric $\chi^2$-divergence, algebraic-geometric mean, R\'{e}nyi information gain with $\alpha \not\in [-1, 2]$ and $L_{f_\alpha} = 1$),
    \item One of $0$ or $1$, and also $1/2$ (KL and reverse KL divergence),
    \item One of $0$ or $1$, but not $1/2$ (R\'{e}nyi information gain with $\alpha \in [-1,0] \cup [1,2]$ and $L_{f_\alpha} = 4$, and consequently Pearson and Neyman $\chi^2$-divergence),
    \item Only $1/2$ (the examples in this section).
\end{itemize}

We begin by showing that Theorem~\ref{thm:pinsker-multi} (unlike Theorem~\ref{thm:pinsker-uni}) is directly applicable to Jeffrey's divergence.
\begin{proposition}[Jeffrey's divergence]\label{prop:pinsker-jeffrey}
    Let $f(t) = (t-1) \ln t$.
    Then, $L_f = 8$ satisfies \eqref{eq:lambda} with $\lambda = 1/2$.
\end{proposition}

Theorem~\ref{thm:pinsker-multi} also allows tightening a specific case in comparison to Proposition~\ref{prop:pinsker-renyi-uni}.
\begin{proposition}[R\'{e}nyi's information gain]\label{prop:pinsker-renyi-multi}
    Let $\alpha \in (0, 1)$ and $f_{\alpha}(t) = \frac{t^\alpha - 1}{\alpha (\alpha-1)}$, as in \eqref{eq:fa}.
    Then, $L_{f_\alpha} = 4$ satisfies \eqref{eq:lambda} with $\lambda = 1/2$.
\end{proposition}

\begin{corollary}[Hellinger divergence]\label{cor:pinsker-hellinger-div-multi}
    Let $\alpha \in (0, 1)$ and $f_{\alpha}(t) = \frac{t^\alpha - 1}{\alpha-1}$, as in \eqref{eq:fah}.
    Then, $L_{f_\alpha} = 4\alpha$ satisfies \eqref{eq:lambda} with $\lambda = 1/2$.
\end{corollary}
\begin{proof}
    This follows from Proposition~\ref{cor:pinsker-hellinger-div-multi} by multiplying by $\alpha$.
\end{proof}

\begin{corollary}[Squared Hellinger distance]\label{cor:pinsker-squared-hellinger}
    Let $f(t) = \frac{1}{2}(\sqrt{t} - 1)^2$.
    Then, $L_f = 1$ satisfies \eqref{eq:lambda} with $\lambda = 1/2$.
\end{corollary}
\begin{proof}
    This follows from Proposition~\ref{prop:pinsker-renyi-multi} by setting $\alpha = 1/2$ and dividing by $4$.
\end{proof}

In the case $c=1$, this immediately recovers a known lower bound on Hellinger distance.
\begin{equation*}
    \mrm{TV}(\mbf{p},\mbf{q}) \leq \sqrt{2} H(\mbf{p},\mbf{q})\ .
\end{equation*}

Below are a few more example applications of Theorem~\ref{thm:pinsker-multi}.
\begin{proposition}[Lin's measure]\label{prop:lins}
    Let $\theta \in [0, 1]$ and $f(t) = \theta t\ln t - (\theta t + 1 - \theta) \ln (\theta t + 1 - \theta)$.
    Then, $L_f = 4 \theta (1-\theta)$ satisfies $\eqref{eq:lambda}$ with $\lambda = 0.5$.
\end{proposition}
Note also that the bound holds on the domain of $D_f$ for all $\theta$ once the sign flip is accounted for, e.g. via dividing $f$ by $\theta (1-\theta)$.
However, in this case $f$ does not form a proper $f$-divergence since its domain does not cover all of $(0, +\infty)$.

\begin{proposition}[Jensen-Shannon divergence]\label{prop:js}
    Let $f(t) = \frac{1}{2}(t \ln t - (t+1) \ln(\frac{t+1}{2}))$.
    Then, $L_f = 1$ satisfies \eqref{eq:lambda} with $\lambda = 1/2$.
\end{proposition}

\begin{proposition}[Triangular discrimination]\label{prop:pinsker-triangle}
    Let $f(t) = \frac{(t-1)^2}{t+1}$.
    Then, $L_f = 4$ satisfies \eqref{eq:lambda} with $\lambda = 1/2$.
\end{proposition}
\begin{proof}
    We have $f''(t) = \frac{8}{(t+1)^3}$, which implies
    \begin{equation*}
        h_{1/2}(x,y) = \frac{2}{x+y} + \frac{2}{2 - x - y}
        = \frac{4}{s (2 - s)}
        \quad \text{with} \quad
        s = x + y\ .
    \end{equation*}
    This function is convex on $[-2, 2]$ and its minimum is $4$ at $s = 1$.
\end{proof}

\begin{table}[H]
    \centering
    \bgroup
    \def\arraystretch{1.5}
    \begin{tabular}{rlrll}
        Divergence & $f(t)$ & $L_f$ & Conditions & Reference \\
        \hline
        Jeffrey's divergence & $(t-1) \ln t$ & $8$ && Proposition~\ref{prop:pinsker-jeffrey} \\
        \multirow{2}{*}{R\'{e}nyi's information gain} & \multirow{2}{*}{$\frac{t^\alpha-1}{\alpha(\alpha-1)}$} & $4$ & $\alpha \in [-1, 2]$ & \multirow{2}{*}{Propositions~\ref{prop:pinsker-renyi-uni} and \ref{prop:pinsker-renyi-multi}}\\
        && $1$ & otherwise &\\
        \multirow{2}{*}{Hellinger divergence} & \multirow{2}{*}{$\frac{t^\alpha-1}{\alpha-1}$} & $4\alpha$ & $\alpha \in (0, 2]$ & \multirow{2}{*}{Corollaries~\ref{cor:pinsker-hellinger-div-uni} and \ref{cor:pinsker-hellinger-div-multi}}\\
        && $\alpha$ & $\alpha \in (2, +\infty)$ &\\
        Squared Hellinger distance & $\frac{1}{2}(\sqrt{t} - 1)^2$ & $1$ && Corollary~\ref{cor:pinsker-squared-hellinger}\\
        Lin's measure & $\theta t \ln t - (\theta t + 1 - \theta) \ln (\theta t + 1 - \theta)$ & $4 \theta (1 - \theta)$ & $\theta \in [0,1]$ & Proposition~\ref{prop:lins}\\
        Jensen-Shannon divergence & $\frac{1}{2} \left( t \ln t - (t + 1) \ln \left( \frac{t+1}{2} \right) \right)$ & $1$ && Proposition~\ref{prop:js}\\
        Triangular discrimination & $\frac{(t-1)^2}{t+1}$ & $4$ && Proposition~\ref{prop:pinsker-triangle}\\
    \end{tabular}
    \egroup\vspace*{1em}
    \caption{Summary of additional bounds obtained by Theorem~\ref{thm:pinsker-multi}}
    \label{tab:multi}
\end{table}

We conclude by discussing the limitations of Theorem~\ref{thm:pinsker-multi}.
For one, the theorem only applies to $f$-divergences generated by $f$ that are twice-differentiable.
When the theorem does apply, it is not guaranteed to retrieve a non-trivial lower bound.
A prominent example of both cases is $\chi^\alpha$ divergence with $f_\alpha(t) = \abs{t-1}^\alpha$ and $\alpha \geq 1$, which is not twice-differentiable for $\alpha \in [1, 2)$ and Theorem~\ref{thm:pinsker-multi} yields the trivial bound when $\alpha > 2$ (though notably a non-trivial lower bound does not appear possible here). Even when Theorem~\ref{thm:pinsker-multi} attains a non-trivial bound, it is not guaranteed to be the tightest possible of the kind.
An example of such a mismatch is the $f$-divergence in Proposition~\ref{prop:example} -- although a tighter bound is possible with $\lambda = 1/2$, it still falls short of the numerical evidence of the best constant.

On the other hand, the approach still appears to yield meaningful bounds in cases where it is not technically applicable.
An example of this is
\begin{equation*}
    f(t) = \begin{cases}
        (t-1)^2 & \text{if}\ t \leq 1\ ,\\
        0 & \text{otherwise}\ ,
    \end{cases}
\end{equation*}
for which the method, if applicable, would give $L_f = 2$ via $\lambda = 0$.
This suggests there may be a means of replicating this approach that does not rely on $f$ being twice continuously differentiable, though this would likely require an alternative to Taylor's theorem.
One such candidate is given for arbitrary convex functions in \cite{Liese-2006a}.

\section{Divergence Inequalities}\label{sec:divergence-inequalities}
In this section, we derive new divergence inequalities. In particular, we derive input-dependent bounds on $f$-divergences in terms of the $\chi^{2}$-divergence. This allows us to establish new reverse Pinsker inequalities for $f$-divergences. The key technical lemma is a simple integral representation for Bregman divergences which recovers an integral representation for $f$-divergences as a special case. At the end of the section, we also show our methodology can obtain reverse Pinsker inequalities for Bregman divergences, which may have applications in learning theory and statistics.

\subsection{Integral Representation for Bregman Divergences and \texorpdfstring{$f$}{}-Divergences} We begin by establishing our main technical lemma--- a multivariate integral representation for Bregman and $f$-divergences obtained via the multivariate Taylor's theorem in integral form.

Recall multivariate Taylor's theorem with integral remainder, as stated in Lemma~\ref{lem:first-order-Taylor's} and Corollary~\ref{cor:first-order-taylor}. The latter nearly immediately implies an integral representation for Bregman divergences defined in terms of a twice-differentiable $F$.
\begin{lemma}\label{lem:integral-rep-for-breg-div}
    Let $S \subset \mbb{R}^{n}$ be an open convex set and $F: S \to \mbb{R}$ be convex and twice-differentiable over $S$. Define $\pmb{\lambda}_{t} := (1-t)\mbf{q} + t\mbf{p}$. Then,
    \begin{equation}\label{eq:bregman-div-int-rep} 
    \begin{aligned}
    B_{F}(\mbf{p}||\mbf{q})
    =& \int^{1}_{0} (1-t) \left[(\mbf{p}-\mbf{q})^{T} H_{F\vert_{\pmb{\lambda}_{t}}} (\mbf{p} -\mbf{q}) \right] \, dt \ . 
    \end{aligned}
    \end{equation}
\end{lemma}
\begin{proof}
    Choose $\mbf{a} = \mbf{q}$, $\mbf{p} = \mbf{a} + \mbf{h}$, which means $\mbf{h} = \mbf{p}-\mbf{q}$. This means $\mbf{a}+t\mbf{h} = \mbf{q} + t(\mbf{p}-\mbf{q}) = (1-t)\mbf{q} + t\mbf{p}$. Applying Corollary~\ref{cor:first-order-taylor},
    \begin{align*}
        & F(\mbf{p}) = F(\mbf{q}) + \langle \grad F(\mbf{q}), \mbf{h} \rangle + \int^{1}_{0} (1-t) \left[\mbf{h}^{T} H_{F}\vert_{\mbf{a}+t\mbf{h}} \mbf{h} \right] dt \ .
    \end{align*}
    Re-arranging the linear terms onto the left hand side and using $\mbf{h} = \mbf{p}-\mbf{q}$ so that it is in the form of \eqref{eq:Bregman-Div} establishes \eqref{eq:bregman-div-int-rep}.
\end{proof}
First, note the above proof only used the convexity of $F$ to align with the definition of Bregman divergence (Definition~\ref{def:BregDiv}). We also remark the above result is in retrospect obvious: the Bregman divergence is defined as the the difference between $F(p)$ and the first-order Taylor expansion of $F$ around $\mbf{q}$ evaluated at $\mbf{p}$. It follows that so long as $F$ is twice-differentiable, $B_{F}(\mbf{p} \Vert \mbf{q})$ should admit an integral representation that is the integral remainder in Taylor's theorem. 

\paragraph{\texorpdfstring{$f$}{f}-Divergence}
The above result simplifies when we consider $f$-divergences. To do this, we define the following functions,
\begin{definition}
    Let $f:S \to \mbb{R}$ for some open interval $(0,+\infty) \subseteq S \subset \mbb{R}$. For any $\mbf{r} \geq \mbf{0}$, we define the function
\begin{align}
    f_{\mbf{r}}(\mbf{x}) := \sum_{i \in \cX} r_{i}f(x_{i}/r_{i}) = \cD_{f}(\mbf{x}||\mbf{r}) \ ,
\end{align}
where we remind the reader that $\cD_{f}(\cdot \Vert \cdot)$ is the `unconstrained' $f$-divergence (See Definition~\ref{def:f-divergence}).
\end{definition}

In the subsequent results, we will require that $\mbf{q} > 0$ and $\mbf{p} \ll \mbf{q}$. However, note by Proposition~\ref{prop:classical-f-div-restrict-to-q-support}, the assumption $\mbf{q} > 0$ is effectively irrelevant so long as $\mbf{p} \ll \mbf{q}$.

\begin{lemma}\label{lem:f-div-first-order-taylor}
    Let $\mbf{q}>\mbf{0}$ and $\mbf{0} \leq \mbf{p} \ll \mbf{q}$. Let $(0,+\infty) \subseteq S \subseteq \mbb{R}$ be an interval and $f:S \to \mbb{R}$ be twice continuously differentiable on open interval $(0,+\infty) \subseteq I \subseteq \mbb{R}$, $f(1) = 0$, and either (i) $f'(1) = 0$ or (ii) $\langle \mbf{p} \rangle = \langle \mbf{q} \rangle$, then
    \begin{equation}\label{eq:f-div-taylor-first-order-exp}
    \begin{aligned}
        \cD_{f}(\mbf{p}||\mbf{q}) =  \int_{0}^{1} (1-t)\langle \pmb{\partial}^{2}f_{\mbf{q}}((1-t)\mbf{q}+t\mbf{p}), (\mbf{p}-\mbf{q})^{\odot 2} \rangle  dt \ ,
    \end{aligned}
    \end{equation}
    where $\pmb{\partial}^{2}f(\mbf{x}) := \sum_{i} \frac{\partial^{2}f(\mbf{x})}{\partial x_{i}^{2}} e_{i}$.
\end{lemma}
\begin{proof}
    We first get a simplified expression for $f_{\mbf{r}}(\mbf{p})$ from Corollary~\ref{cor:first-order-taylor} under the conditions $\mbf{r},\mbf{p},\mbf{q} > \mbf{0}$. Note the following holds for all $i,j \in [n]$,
    \begin{align}\label{eq:f-div-partials}
        \frac{\partial f_{\mbf{r}}}{\partial x_{i}} = f'\left(\frac{x_{i}}{r_{i}}\right)\ ,\quad \quad \frac{\partial^{2} f_{\mbf{r}}}{\partial x_{i} \partial x_{j}} = \delta_{i,j} r_{i}^{-1} f''\left(\frac{x_{i}}{r_{i}}\right) \ .
    \end{align}
    This implies all terms in the Hessian are zero except the diagonal terms:
    $$\frac{\partial^{2}f_{\mbf{r}}(\mbf{x})}{\partial^{2}x_{i}}(p_{i}-q_{i})^{2} \ . $$
    Applying Corollary~\ref{cor:first-order-taylor} and simplifying with $\mbf{r} = \mbf{q}$ establishes
    \begin{align*}
        f_{\mbf{q}}(\mbf{p}) =& f_{\mbf{q}}(\mbf{q}) + \langle \grad f_{\mbf{q}}(\mbf{q}), \mbf{p}-\mbf{q} \rangle + \int_{0}^{1}  (1-t) \langle\pmb{\partial}^{2}f_{\mbf{q}}((1-t)\mbf{q}+t\mbf{p}), (\mbf{p}-\mbf{q})^{\odot 2} \rangle  dt.
    \end{align*}
    The first term goes away by assumption $f(1) = 0$. Since $\frac{\partial f_{\mbf{r}}}{\partial x_{i}}(\mbf{r}) = f'\left(\frac{r_{i}}{r_{i}}\right) = f'(1) =: c_{1}$ for all $i$, then $\langle \grad f_{\mbf{q}}, \mbf{p}-\mbf{q} \rangle = c_{1} \langle \mbf{p} - \mbf{q} \rangle$, which is zero by either of our conditions. Finally, we note $\cD_{f}(\mbf{p}||\mbf{q}) = f_{\mbf{q}}(\mbf{p})$, which completes the proof.
\end{proof}

\subsection{Inequalities from Integral Representations}
We now derive inequalities from our integral representations.

\subsubsection{Bounding \texorpdfstring{$f$}{}-Divergences in terms of \texorpdfstring{$\chi^{2}$}{}-Divergences}
The first set of inequalities shows that under many conditions we can bound $f$ divergences in terms of the $\chi^{2}$-divergence, the second derivative of $f$, and the ratios of the vector's elements, which are the likelihood ratios in the case that the vectors are probability distributions.

We remind the reader that the $\chi^{2}$-divergence is known to approximate other $f$-divergences when two distributions are sufficiently similar \cite{Csiszar-2004a} and that all $f$-divergences behave like the $\chi^{2}$-divergence `locally' \cite{Sason-2018a} as do a class of quantum $f$-divergences with respect to the quantum $\chi^{2}$-divergence \cite{Hirche-2023a}. In effect the following theorem provides a quantitative method of extending the behaviour in \eqref{eq:csiszar-loc-behavior} to the case where the distributions are no longer necessarily close. We note that there has been special cases of $f$-divergences that have upper bounds in terms of $\chi^{2}$-divergence (See \cite{Raginsky-2016a} and a particularly clear proof of the same result in \cite{Makur-2019a}), but the conditions on $f$ in that case are much more stringent than ours.

\begin{theorem}\label{thm:f-div-chi-squared-bounds}
    Let $\mbf{p} \ll \mbf{q}$. Let $f$ be twice continuously differentiable on open interval $(0,+\infty)$ and assume either  (i) $\vert f''(0) \vert < +\infty$ or (ii)  $\mbf{p} > 0$. Define
    \begin{align} 
        \kappa_{f}^{\uparrow}(\mbf{p}, \mbf{q}) &:= \max_{i \in \supp(\mbf{q}), t \in [0,1]} f''\left(1+t\left(\frac{p_{i}}{q_{i}}-1\right)\right) \label{eq:kappa-up-arrow-defn} \\
        \kappa_{f}^{\downarrow}(\mbf{p}, \mbf{q}) &:= \min_{i \in \supp(\mbf{q}), t \in [0,1]} f''\left(1+t\left(\frac{p_{i}}{q_{i}}-1\right)\right)  
         \ . \label{eq:kappa-down-arrow-defn}
    \end{align}
    Then, we have
    \begin{align}\label{eq:chi-squared-upper-bound}
        \frac{\kappa_{f}^{\downarrow}(\mbf{p}, \mbf{q})}{2}\chi^{2}(\mbf{p}||\mbf{q}) \leq \cD_{f}(\mbf{p}||\mbf{q}) \leq \frac{\kappa_{f}^{\uparrow}(\mbf{p}, \mbf{q})}{2}\chi^{2}(\mbf{p}||\mbf{q}) \ .
    \end{align} Moreover, if $f$ is strictly convex, the lower bound is strictly greater than zero unless $\mbf{p} = \mbf{q}$.
\end{theorem}

\begin{proof}
First, by Proposition~\ref{prop:classical-f-div-restrict-to-q-support}, we may restrict the $f$-divergence to being calculated on the support of $\mbf{q}$ without loss of generality. Now, starting from \eqref{eq:f-div-taylor-first-order-exp} and using \eqref{eq:f-div-partials},
    \begin{align*}
        \cD_{f}(\mbf{p}||\mbf{q}) 
        =&  \int_{0}^{1} (1-t) \left[\sum_{i \in \supp(\mbf{q})} q_{i}^{-1}f''\left(\frac{(1-t)q_{i} + tp_{i}}{q_{i}}\right) (p_{i}-q_{i})^{2} \right] dt \\
        =&  \int_{0}^{1} (1-t) \left[\sum_{i \in \supp(\mbf{q})}f''\left(1+t\left(\frac{p_{i}}{q_{i}}-1\right)\right)    q_{i}^{-1}(p_{i}-q_{i})^{2} \right] dt \\
        =& \sum_{i \in \supp(\mbf{q})} \int_{0}^{1} (1-t)f''\left(1+t\left(\frac{p_{i}}{q_{i}}-1\right)\right)    q_{i}^{-1}(p_{i}-q_{i})^{2}  dt \ ,
    \end{align*}
where the second equality is just rearranging terms and the third uses our assumptions to interchange the integral and sum, which, since we consider a finite sum, holds so long as the individual integrals exist by Fubini's theorem. To see the integrals exist, note we need to consider the closed interval $\{1+t\left(\frac{p_{i}}{q_{i}}-1\right) \}_{t \in [0,1]}$ for each $i \in \supp(\mbf{q})$. This interval is contained in $(0,+\infty)$ when $\mbf{p} > 0$, which $f''$ is continuous over by assumption, so the integrals exist. In the case $\mbf{p} \geq 0$, the integral still exists under the assumption $\vert f''(0) \vert < +\infty$ as we then may apply the Riemann-Lebesgue theorem. Note that the intervals being closed justifies the $\max$ and $\min$ in \eqref{eq:kappa-up-arrow-defn},\eqref{eq:kappa-down-arrow-defn}. 

Continuing from the above equality,
    \begin{align*}
        \cD_{f}(\mbf{p}||\mbf{q}) 
         =& \sum_{i \in \supp(\mbf{q})} \int_{0}^{1} (1-t)f''\left(1+t\left(\frac{p_{i}}{q_{i}}-1\right)\right)    q_{i}^{-1}(p_{i}-q_{i})^{2}  dt  \ ,\nonumber \\
        \leq& \kappa_{f}^{\uparrow}(\mbf{p}, \mbf{q}) \sum_{i \in \supp(\mbf{q})} \int_{0}^{1} (1-t) q_{i}^{-1} (p_{i}-q_{i})^{2}  dt \\
        =& \kappa_{f}^{\uparrow}(\mbf{p}, \mbf{q}) \int_{0}^{1} (1-t) \left[ \sum_{i \in \supp(\mbf{q})} q_{i}^{-1} (p_{i}-q_{i})^{2} \right]  dt \\
        =& \frac{\kappa_{f}^{\uparrow}(\mbf{p}, \mbf{q})}{2}\chi^{2}(\mbf{p}||\mbf{q}) \ ,
    \end{align*}
    where we used the definition of $\kappa_{f}^{\uparrow}(\mbf{p}, \mbf{q})$ in the inequality and the definition of $\chi^{2}(\mbf{p}||\mbf{q})$ (See Eq.~\eqref{eq:chi-squared-div}).The lower bound is by an identical argument using $\kappa_{f}^{\downarrow}(\mbf{p}, \mbf{q})$.
    
    Finally, note that if $f$ is strictly convex, by the twice-differentiable condition for strict convexity, $f''(y) > 0$ for all $y \in (0,+\infty)$, so the lower bound is zero if and only if $\chi^{2}(\mbf{p}||\mbf{q})$, which is zero if and only if $\mbf{p}=\mbf{q}$ (See e.g. \cite{polyanskiy-2023a}). This completes the proof.
\end{proof}

Before continuing on, we make some remarks about our result. First, we note that Theorem~\ref{thm:f-div-chi-squared-bounds} implies \eqref{eq:csiszar-loc-behavior} in the finite-dimensional setting because 
\begin{align}\label{eq:kappa-limits}
\lim_{\mbf{p} \to \mbf{q}}\kappa_{f}^{\uparrow}(\mbf{p},\mbf{q}) = f''(1) = \lim_{\mbf{p} \to \mbf{q}}\kappa_{f}^{\downarrow}(\mbf{p},\mbf{q}) 
\end{align}
as may be verified via \eqref{eq:kappa-up-arrow-defn},\eqref{eq:kappa-down-arrow-defn}. This is a good general sanity check. 

Second, a standard presentation for $f$-divergence inequalities is to make some assumption on the likelihood ratio of the distributions, e.g. \cite{Barnett-2002a,Taneja-2004a, Dragomir-2016a,Dragomir-2019a}. We state our versions of such results as a corollary.
\begin{corollary}
    Consider interval $[m,M]$ where $0 < m \leq 1 \leq M$. Let $f$ be twice continuously differentiable such that $f(1) = 0$ and $f'(1) = 0$. Let $\mbf{p},\mbf{q}$ such that $m \leq p(x)/q(x) \leq M$ for all $x \in \cX$. Then,
    \begin{align}
        \frac{l(m,M)}{2} \chi^{2}(\mbf{p} \Vert \mbf{q}) \leq \cD_{f}(\mbf{p}\Vert \mbf{q}) \leq \frac{u(m,M)}{2}\chi^{2}(\mbf{p} \Vert \mbf{q}) \ , 
    \end{align}
    where $l(m,M) \coloneq \min_{a \in [m,M]} f''(a)$ and $u(m,M) \coloneq \max_{a \in [m,M]} f''(M)$.
\end{corollary} 
To the best of our ability, we have not been able to find a work that explicitly states these specific bounds, though we remark they are morally similar to results in \cite{Taneja-2004a,Dragomir-2016a,Dragomir-2019a}.

\subsubsection{Reverse Pinsker Inequalities and Different Lower bounds in terms of \texorpdfstring{$\chi^{2}$}{}-Divergence}The same proof method as the previous theorem allows us to establish reverse Pinsker inequalities for the twice continuously differentiable $f$-divergences.
\begin{corollary}\label{cor:Reverse-Pinsker}
    Let $0 \leq \mbf{p}$ and $\mbf{p} \ll \mbf{q}$. Let $(0,+\infty) \subseteq S \subseteq \mbb{R}$ be an interval and $f:S \to \mbb{R}$ be twice continuously differentiable on open interval $(0,+\infty) \subseteq I \subseteq \mbb{R}$. Then,
    \begin{align}
        \cD_{f}(\mbf{p}||\mbf{q}) \leq& \frac{\kappa_{f}^{\uparrow}(\mbf{p},\mbf{q})}{2\wt{q}_{\min}} \Vert\mbf{p}-\mbf{q}\Vert_{2}^{2} \label{eq:Lipschitz-like} \\
        \leq& \frac{2\kappa_{f}^{\uparrow}(\mbf{p},\mbf{q})}{\wt{q}_{\min}}\mrm{TV}(\mbf{p},\mbf{q})^{2} \ , \label{eq:Reverse-Pinsker}
    \end{align}
    where $\kappa_{f}^{\uparrow}(\mbf{p},\mbf{q})$ is defined in Eq.~\eqref{eq:kappa-up-arrow-defn}, $\wt{q}_{\min} := \min_{i: q_{i} > 0} q_{i}$, and this bound is finite under the same conditions as given in Theorem~\ref{thm:f-div-chi-squared-bounds}.
\end{corollary}
\begin{proof}
    As we assume $\mbf{p} \ll \mbf{q}$, we restrict to the support of $\mbf{q}$  without loss of generality. Then, as established at the start of the proof of Theorem~\ref{thm:f-div-chi-squared-bounds}, we have
    \begin{align*}
        \cD_{f}(\mbf{p}||\mbf{q})
        = \sum_{i \in \supp(\mbf{q})} \int_{0}^{1} (1-t)f''\left(1+t\left(\frac{p_{i}}{q_{i}}-1\right)\right)    q_{i}^{-1}(p_{i}-q_{i})^{2}  dt \ .
    \end{align*}
    Noting that on the support of $\mbf{q}$, $q_{i}^{-1} \leq \wt{q}^{-1}_{\min}$, we have
    \begin{align*}
        \cD_{f}(\mbf{p}||\mbf{q})
        \leq & \wt{q}_{\min}^{-1} \sum_{i \in \supp(\mbf{q})} \int_{0}^{1} (1-t)f''\left(1+t\left(\frac{p_{i}}{q_{i}}-1\right)\right) (p_{i}-q_{i})^{2}  dt \\
        \leq & \frac{\kappa_{f}^{\uparrow}(\mbf{p},\mbf{q})}{\wt{q}_{\min}}  \int^{1}_{0} (1-t) \left[ \sum_{i \in \supp(\mbf{q})} (p_{i} - q_{i})^{2} \right] dt \\
        =& \frac{\kappa_{f}^{\uparrow}(\mbf{p},\mbf{q})}{2\wt{q}_{\min}} \Vert\mbf{p}-\mbf{q}\Vert_{2}^{2} \ ,
    \end{align*}
    where the second inequality is our definition of $\kappa_{f}^{\uparrow}(\mbf{p},\mbf{q})$ and the equality uses the definition of the $\ell^{2}$-norm and that $\mbf{p} \ll \mbf{q}$. Recalling that $\Vert\mbf{x}\Vert_{2} \leq \Vert\mbf{x}\Vert_{1}$ and using $\mrm{TV}(\mbf{p},\mbf{q}) := \frac{1}{2}\Vert \mbf{p} - \mbf{q}\Vert_{1}$ establishes \eqref{eq:Reverse-Pinsker} in the theorem statement.
\end{proof}
We note that there exist reverse Pinsker inequalities for $f$-divergences established in \cite{Sason-2016-f-div-ineqs}. In our case we require $f$ to be twice continuously differentiable, whereas there they require $f$ to be $L$-Lipschitz over the interval they consider.

To get a better gauge on these inequalities, we provide two special cases of the above: the KL divergence and the $\chi^{2}$-divergence. We note our reverse Pinsker inequality is not restricted to probability distributions like other known reverse Pinsker inequalities for the KL divergence, e.g. \cite{Csiszar-2006a,Sason-2016-f-div-ineqs}.
\begin{corollary}
    For $\mbf{p} \ll \mbf{q}$,
    \begin{align}
        \chi^{2}(\mbf{p} \Vert \mbf{q}) \leq& \frac{4}{\wt{q}_{\min}}\mrm{TV}(\mbf{p},\mbf{q})^{2} \label{eq:chi-squared-reverse-pinsker} \ .
    \end{align}
    Similarly, if $\mbf{p} \ll \gg \mbf{q}$,
    \begin{align}\label{eq:reverse-Pinsker-non-normalized}
        D(\mbf{p} \Vert \mbf{q}) \leq& \frac{2\ln(2) r}{q_{\min}}\mrm{TV}(\mbf{p},\mbf{q})^{2} , 
    \end{align}
    where $r := \max\{1,\max_{i \in \supp(\mbf{q})} q_{i}/p_{i}\}$.
\end{corollary}
\begin{proof}
    For the $\chi^{2}$-divergence, recall it is defined via $f(t) := t^{2} -1$, so $f''(t) = 2$ for all $t \in \mbb{R}$. Thus, $\kappa_{f}^{\uparrow}(\mbf{p},\mbf{q}) = 2$.
    
    For the KL divergence, recall it is defined via $f(t) := t\log(t)$, so $f''(t) = \frac{1}{t \ln(2)}$. Thus, we are interested in maximizing $\frac{1}{1+t(p_{i}/q_{i}-1)}$ over $i \in \supp(\mbf{q})$ as we may restrict to its support. In this case, the maximum is achieved by either choosing $i^{\star} := \text{argmin}_{i \in \supp(\mbf{q})} \frac{p_{i}}{q_{i}} < 1$ if such a point exists and letting $t = 1$ and otherwise setting $t = 0$. Note we assume $\mbf{q} \ll \mbf{p}$ as otherwise the bound becomes trivial.
\end{proof}

It is not hard to see that the our reverse Pinsker for KL divergence is incomparable for probability distributions to the previous best method \cite{Csiszar-2006a}. To see this, we may rewrite our result as $D(\mbf{p} \Vert \mbf{q}) \leq \frac{r}{4q_{\min} \log(e)}\Vert \mbf{p} - \mbf{q} \Vert_{1}^{2}$ whereas for probability measures on a finite alphabet, the known result \cite{Csiszar-2006a} is $D(\mbf{p} \Vert \mbf{q}) \leq \frac{\log(e)}{q_{\min}}\Vert \mbf{p} - \mbf{q} \Vert_{1}^{2}$. It follows that when $r \leq 4\log^{2}(e) \approx 8.325$, our result is tighter. In other words, when $\mbf{p},\mbf{q}$ are rather similar in an entry-wise manner, our result is tighter, but when they are further apart, our result is more loose. 

In contrast to our discussion on our reverse Pinsker for KL divergence, the $\chi^{2}$ reverse Pinsker can be universally tightened using extra structure from its definition. For completeness, and to use it later, we note the following which is a direct generalization of a claim in \cite{Makur-2020a} and a known inequality from \cite{Sason-2014a}.
\begin{proposition}(\cite{Sason-2014a,Makur-2020a})\label{prop-chi-squared-lower-bound-tv}
    Let $0 \leq \mbf{p} \ll \mbf{q}$.
    \begin{align}
        \chi^{2}(\mbf{p}\Vert \mbf{q}) \leq \frac{\Vert \mbf{p} - \mbf{q} \Vert_{\infty}}{\wt{q}_{\min}} \Vert \mbf{p} - \mbf{q} \Vert_{1} \ . 
    \end{align}
    Moreover, if $\mbf{p}, \mbf{q} \in \cP(\cX)$ such that $\mbf{p} \neq \mbf{q}$ as well, then
    \begin{align}\label{eq:chi-squared-reverse-Pinsker-Sason}
        \chi^{2}(\mbf{p} \Vert \mbf{q}) \leq \frac{\Vert \mbf{p} - \mbf{q}\Vert_{1}^{2}}{2 \wt{q}_{\min}}
    \end{align}
\end{proposition}
\begin{proof}
First we provide an immediate generalization of an inequality given in \cite{Makur-2020a}:
    \begin{align}
        \chi^{2}(\mbf{p}\Vert \mbf{q}) 
        =& \sum_{x \in \supp(\mbf{q})} \vert p_{x} - q_{x} \vert \frac{\vert p_{x} - q_{x}\vert }{\vert q_{x} \vert} \\
        \leq&  \sum_{x \in \supp(\mbf{q})} \vert p_{x} - q_{x} \vert \frac{\Vert \mbf{p} - \mbf{q} \Vert_{\infty}}{\wt{q}_{\min}} \\
        =&   \frac{\Vert \mbf{p} - \mbf{q} \Vert_{\infty}}{\wt{q}_{\min}} \Vert \mbf{p} - \mbf{q} \Vert_{1} \ ,
    \end{align}
    where the inequality uses that $\vert p_{x} - q_{x} \vert \leq \max_{x} \vert p_{x} - q_{x}\vert = \Vert \mbf{p} - \mbf{q} \Vert_{\infty}$ and $\vert q_{x} \vert$ can only be larger than the minimal entry on the support. In the case $\mbf{p},\mbf{q} \in \cP(\cX)$ then one may simplify to \eqref{eq:chi-squared-reverse-Pinsker-Sason} following the argument in \cite[Eqn. 72]{Makur-2020a}, although the relation was already established in \cite[Proposition 3]{Sason-2014a}.
\end{proof}

We may combine this with our Pinsker inequalities to get lower bounds on $f$-divergences in terms of $\chi^{2}$-divergences that do not depend on $\kappa^{\downarrow}$ as will be useful later.
\begin{lemma}\label{lem:f-div-lb-by-chi-squared}
    Let $f:(0,+\infty) \to \mbb{R}$ be convex and twice continuously differentiable and $\mbf{p},\mbf{q} \in \cP(\cX)$. Then,
    \begin{align}
         D_{f}(\mbf{p} \Vert \mbf{q}) \geq \frac{L_{f}\wt{q}_{\min}}{4}\chi^{2}(\mbf{p} \Vert \mbf{q}) \ .
    \end{align}
\end{lemma}
\begin{proof}
    If $L_{f} = 0$, we have nothing to prove, so we assume $L_{f} > 0$. By Theorem \ref{thm:pinsker-uni} or \ref{thm:pinsker-multi} followed by \eqref{eq:chi-squared-reverse-Pinsker-Sason}, we have
    \begin{align}
        D_{f}(\mbf{p} \Vert \mbf{q}) \geq \frac{L_{f}}{8}\Vert \mbf{p} - \mbf{q} \Vert_{1}^{2} \geq \frac{L_{f}\wt{q}_{\min}}{4}\chi^{2}(\mbf{p} \Vert \mbf{q}) \ .
    \end{align}
\end{proof}

In total, what the above examples tell us is that, arguably unsurprisingly, using a more general method comes at a cost can come at the cost of tightness in either some regime or, even worse, in all regimes. However, it is this generality in the choice of $f$ that we are interested in and which will be used in Section~\ref{sec:classical-Input-Dependent-SDPI}.

\subsubsection{Pinsker and Reverse Pinsker Inequalities for Bregman Divergences}
We now use our integral representations of Bregman divergences to obtain (input-dependent) Pinsker and Reverse Pinsker inequalities. Technically, these are more general than the cases we considered previously, but they seem less insightful in general. Note that we could not get this from the univariate Taylor's theorem methodology of Section~\ref{sec:f-div-Pinsker} as the Bregman divergence can involve a function $F: S \to \mbb{R}$ where $S \subset \mbb{R}^{n}$ (See Definition~\ref{def:BregDiv}). We also note in the specific case that the Bregman divergence was defined in terms of the negative of the entropy of a vector, $F \to -S(\mbf{x}) := \sum_{i \in \cX} x_{i}\log(x_{i})$, the (non-reverse) Pinsker inequality was derived in \cite{Bansal-2024a} using the integral representation where it was used for online learning regret bounds. 

\begin{corollary}
    Let $S \subset \mbb{R}^{n}$ be a \textit{closed} convex set and $F: S \to \mbb{R}$ be convex and twice-differentiable over $S$. Define $\pmb{\lambda}_{t} := (1-t)\mbf{q} + t\mbf{p}$,
    \begin{align*}
        \frac{4\gamma_{F}^{\downarrow}}{|\supp(\mbf{p}-\mbf{q})|^{2}} \mrm{TV}( \mbf{p},\mbf{q})^{2} \leq \gamma_{F}^{\downarrow}(\mbf{p},\mbf{q}) \Vert \mbf{p} - \mbf{q} \Vert_{2}^{2} \leq & B_{F}(\mbf{p}\Vert \mbf{q}) \\
        \leq& \frac{\gamma_{F}^{\uparrow}(\mbf{p},\mbf{q})}{2} \Vert \mbf{p}-\mbf{q}\Vert_{2}^{2} \leq 2 \gamma^{\uparrow}_{F}(\mbf{p},\mbf{q}) \mrm{TV}(\mbf{p},\mbf{q})^{2} \ ,
    \end{align*}
    where 
    \begin{align*}
        \gamma^{\uparrow}_{F}(\mbf{p},\mbf{q}) :=& \max_{t \in [0,1]} \lambda_{\max}(H_{F_{\vert \pmb{\lambda}_{t}}}) \ , \\
        \gamma^{\downarrow}_{F}(\mbf{p},\mbf{q}) :=&  \min_{t \in [0,1]} \lambda_{\min}(H_{F_{\vert\pmb{\lambda}_{t})}})  \ ,
    \end{align*}
    $\pmb{\lambda}_{t}$ is defined in Lemma~\ref{lem:integral-rep-for-breg-div}, and $\lambda_{\max}(\cdot)$ (resp.~$\lambda_{\min}(\cdot)$) is the max (resp.~minimum) eigenvalue function. Moreover, the lower bounds are non-trivial if $F$ is strictly convex.
\end{corollary}
\begin{proof}
We start from Lemma~\ref{lem:integral-rep-for-breg-div}:
\begin{align*}
    B_{F}(\mbf{p}||\mbf{q})
    =& \int^{1}_{0} (1-t) \left[(\mbf{p}-\mbf{q})^{T} H_{F\vert_{\pmb{\lambda}_{t}}} (\mbf{p} -\mbf{q}) dt \right] \\
    \leq& \gamma^{\uparrow}(\mbf{p},\mbf{q}) \Vert \mbf{p} - \mbf{q} \Vert_{2}^{2} \int^{1}_{0} (1-t) \, dt \\
    =& \frac{\gamma^{\uparrow}(\mbf{p},\mbf{q})}{2} \Vert \mbf{p} - \mbf{q} \Vert_{2}^{2} \ ,
\end{align*}
where the inequality is our definition of $\gamma$. One converts the above results to a total variation bound in the manner we have done previously.

In the reverse case, by an identical argument.
\begin{align*}
    B_{F}(\mbf{p}||\mbf{q})
    =& \int^{1}_{0} (1-t) \left[(\mbf{p}-\mbf{q})^{T} H_{F\vert_{\pmb{\lambda}_{t}}} (\mbf{p} -\mbf{q}) dt \right] \geq \frac{\gamma^{\downarrow}}{2}\Vert \mbf{p} - \mbf{q} \Vert_{2}^{2} \ .
\end{align*}
To convert this to total variation, one must use $\|\mbf{x}\|_{2} \geq \frac{1}{|\supp(\mbf{x})|}\|\mbf{x}\|_{1}$, which introduces an unfortunate scaling. 

The moreover statement simply follows from a twice-differentiable convex function having a strictly positive Hessian. This completes the proof.
\end{proof}
We remark that the above results can be made dependent on other properties than specific inputs. For example, we can make the coefficients depend on just the normalization of the vectors as is the case for the traditional Pinsker inequality. We note one can recover the methodology for establishing Pinsker inequalities for $f$-divergences developed in Section~\ref{sec:f-div-Pinsker} through this method, but it obfuscates the key idea, which only relies on univariate Taylor's theorem.

\section{Input-Dependent Strong Data Processing and Markov Chains}\label{sec:classical-Input-Dependent-SDPI}
In this section, we provide our main application of our $f$-divergence bounds in the classical setting: bounds on contraction coefficients for $f$-divergences and their implications for time homogeneous (classical) Markov chains. We begin with some background on contraction coefficients and their relations to notions in discrete time Markov chains where the key points are that the $\chi^{2}$ input-dependent contraction coefficient is the smallest (i.e. fastest contraction speed) and is known to be efficiently computable. We then establish non-linear bounds between contraction coefficients of $f$-divergences, which we use to show the rate of contraction of a time homogeneous Markov chain is the $\chi^{2}$ contraction coefficient if we were to measure it in terms of most $f$-divergences (Theorem~\ref{thm:classical-contraction-rate}), i.e. from an iterative contraction perspective, the input-dependent $\chi^{2}$-contraction is the `correct' notion of contraction. This is an operational answer to the second basic question we posed in the introduction. We then show how to straightforwardly extend previous results to get computable mixing times for a class of $f$-divergences in terms of the input-dependent $\chi^{2}$ contraction coefficient and the choice of $f$ (Proposition~\ref{prop:f-div-mixing-time}).

\subsection{Background and Some Lemmata}
We denote a (classical) channel $\cW_{\cX \to \cY}$ and equivocate it with its matrix representation $W \in \mbb{R}^{\vert \cY \vert \times \vert \cX \vert}$. We stress to align with quantum information theory standards, the output of a channel, $\mbf{p}_{Y}$, is given by $\mbf{p}_{Y} = W\mbf{p}_{X}$, i.e. the matrix representation of the channel is defined by multiplying the vector on the left. We denote the composition of a channel iteratively via the notation $\cW^{n} \coloneq \otimes_{i \in [n]} \cW$. We begin by defining contraction coefficients, which are measures of the strength of data processing of a channel.
\begin{definition}\label{def:classical-contraction-coefficients}
    Consider a channel $\cW_{X \to Y}$ and an $f$-divergence $D_{f}$. The \textit{input-independent contraction coefficient} is
    \begin{align}
       \eta_{f}(\cW) := \sup_{\mbf{p},\mbf{q} \in \cP(\cX) \, : \, 0 < D_{f}(\mbf{p}||\mbf{q}) < + \infty } \frac{D_{f}(W\mbf{p}||W\mbf{q})}{D_{f}(\mbf{p}||\mbf{q})}  \ .
    \end{align}
    The \textit{input-dependent contraction coefficient} is
    \begin{align}\label{eq:classical-input-dependent-contraction-coeff}
        \eta_{f}(\cW,\mbf{q}) := \sup_{\mbf{p} \in \cP(\cX) \, : \, 0 < D_{f}(\mbf{p}||\mbf{q}) < + \infty } \frac{D_{f}(W\mbf{p}||W\mbf{q})}{D_{f}(\mbf{p}||\mbf{q})}  \ .
    \end{align}
\end{definition}

We now define two properties of a matrix which will be relevant for understanding contraction coefficients.
\begin{definition}\label{def:scrambling}
     A non-negative matrix is \textit{scrambling} if and only if no two columns are orthogonal vectors. We say a channel $\cW$ is scrambling if its matrix representation $W$ is scrambling.
\end{definition}
We stress this definition is given in \cite{Hajnal-1958a,Seneta-2006a} where they require that no two \textit{rows} are orthogonal vectors. The reason we use columns is that we consider multiplying vectors on the left by matrices.

\begin{definition}\label{def:indecomposable-joint-dist}
    A joint distribution $p_{XY} \in \cP(\cX \times \cY)$ is decomposable if there exists $A \subset \cX$, $B \subset \cY$ such that $0 < \Pr[x \in A], \Pr[y \in B] < 1$ and $x \in A$ if and only if $y \in B$. Otherwise, it is indecomposable. 
\end{definition}
The above definition is saying that the joint outcome space can be partitioned into at least two pieces. Indeed, it is noted in \cite{Makur-2020a} that indecomposability can be expressed in a graph-theoretic manner as the bipartite graph with vertices over disjoint vertex sets $\cX,\cY$ defined via its edge set $\cE = \{(x,y) : p_{Y|X}(y|x)>0\}$ has two or more connected components.

\begin{proposition}[Facts about Classical Contraction Coefficients]\label{prop:facts-about-contraction-coefficients} Let $\cW_{X \to Y}$ be channel.
    \begin{enumerate}
        \item \cite{Cohen-1993a,Raginsky-2016a} For any $f$-divergence, $\eta_{f}(\cW) \leq \eta_{\text{TV}}(\cW)$.
        \item \cite{Choi-1994a,Raginsky-2016a} For $f$-divergence such that $f$ is \textit{operator} convex, $\eta_{\chi^{2}}(\cW) = \eta_{f}(\cW)$.
        \item \cite{Cohen-1993a} For any $f$-divergence, $\eta_{f}(\cW) < 1$ if and only if $\cW$ is scrambling.
        \item \cite{Polyanskiy-2017a} For any $f$ that is twice-differentiable and $f''(1) > 0$, 
        $$ \eta_{f}(\cW,\mbf{p}) \geq \eta_{\chi^{2}}(\cW,\mbf{p}) = \rho_{m}(\mbf{p},\cW(\mbf{p}))^{2} \ , $$
        where $\eta_{\chi^{2}}$ is the contraction coefficient for $\chi^{2}$ and $\rho_{m}(\cdot,\cdot)$ is the maximal correlation coefficient, which is defined as
        \begin{align}
            \rho_{m}(X,Y) := \sup_{f,g} \mbb{E}[f(X)g(Y)] \ , 
        \end{align} 
        where the supremum is measurable $f:\cX \to \mbb{R}$, $g : \cY \to \mbb{R}$ such that $\mbb{E}[f(X)]=0$, $\mbb{E}[g(Y)]=0$, $\mbb{E}[f(X)^{2}]=1$, and $\mbb{E}[g(Y)^{2}] = 1$. Moreover, the maximal correlation coefficient is efficiently computable in finite dimensions \cite{Witsenhausen-1975a,Kang-2010a} (See \cite{Makur-2020a} for a direct proof) In particular, it is the second eigenvalue of the matrix $\wt{p}_{XY}$ where $p_{XY}$ is the joint distribution of $(W\mbf{p})_{Y}$ and $\mbf{p}_{X}$  
        $$\wt{p}_{XY}(x,y) = \begin{cases} \frac{p_{XY}(x,y)}{\sqrt{p(x)}\sqrt{p(y)}} & p(x),p(y) >0 \\
        0 & \text{otherwise}
        \end{cases}
            \ . $$
    \end{enumerate}
\end{proposition}
It follows that in terms of contraction coefficients, the `fastest' a system can contract is in terms of the $\chi^{2}$-contraction coefficient, $\eta_{\chi^{2}}$. We also note that contraction coefficients are submultiplicative, i.e. $\eta_{f}$, but we prove this also is true for quantum channels so we omit it until that section (Proposition~\ref{prop:sub-multiplicativity}).

\paragraph{Markov Chains} Note that if we consider a channel $\cW_{X \to X}$, i.e. from the alphabet into itself, then we may identify its matrix representation $W$ as representing a discrete time Markov chain. We refer Then one might expect that the properties of a channel that result in an input-independent or input-dependent contraction coefficient are on some level really about the `mixing' properties of the corresponding Markov chain. It follows it will be useful to have some further background on Markov chains. A (finite-space) time-homogeneous Markov chain is a process under which $W$ is applied to the input iteratively.\footnote{Traditionally $W$ is called the transition matrix, but since it is in fact the matrix representation of a channel, we do not make the distinction.} We will conflate the stochastic matrix $W$ corresponding to a time homogenenous Markov chain by talking of a ``Markov chain $W$." The common properties of interest of such Markov chains are the following.
\begin{definition} (See e.g. \cite{Levin-2017a}.)
    \begin{enumerate}
        \item A Markov chain $W$ is irreducible if for all $x,y \in \cX$ there exists time step $t \in \mbb{N}$ such that $\cW^{t} := \circ_{i \in [t]} \cW$ has $W_{X'|X}^{t}(y|x) > 0$.
        \item For a Markov chain $W$, the period of $x \in \cX$, $d(x)$, is the greatest common divisor of $\cT(x) \coloneq \{t \geq 1: W^{t}(x,x) > 0\}$.
        \item A Markov chain $W$ is aperiodic if for all $x \in \cX$, $d(x) =1$, i.e. the period of is unity for all. 
    \end{enumerate}
\end{definition}

The reason these properties are generally considered important is they induce the following mathematical claims.
\begin{proposition}\label{prop:MC-convergence} (See e.g. \cite{Levin-2017a}.)
\begin{enumerate}
    \item An irreducible Markov chain $\cW$ on finite alphabet $\cX$ admits a unique stationary distribution $\pmb{\pi}$ such that $\pmb{\pi}_{x} = \left(\mbb{E}_{x} \tau_{x}^{+} \right)^{-1} > 0$, which is the inverse of the expected return time and is known to be finite for all $x \in \cX$ under these conditions. 
    \item (\textit{Convergence Theorem}) Let $W$ be irreducible and aperiodic with stationary distribution $\pmb{\pi}$. Then there exist constants $\alpha \in (0,1)$, $C>0$ such that for all 
    \begin{align}
        \max_{x \in \cX} \Vert W^{t}e_{x} - \pmb{\pi} \Vert_{1} \leq C \alpha^{t} \ . 
    \end{align}
    That is to say, every input distribution converges exponentially fast to the stationary distribution in total variation at some rate $\alpha$.
\end{enumerate}
\end{proposition}
It is known that being irreducible is not necessary to admit a unique stationary distribution, and is rather captured by a more technical property.
\begin{fact}\label{fact:unique-stationary-distribution}
    (\cite[Proposition 1.29]{Levin-2017a}) A Markov chain $W$ has a unique stationary distribution if and only if it has a `unique essential communicating class.'
\end{fact}
\noindent We do not formally define what is a `unique essential communicating class' in this work as all we appeal to is the existence of the unique stationary distribution. Similarly, it is known that having a unique stationary distribution does not guarantee convergence to it without the aperiodicity condition.\footnote{The standard example is $W(1 \vert 0) = 1$, $W(0 \vert 1) = 1$, which has the uniform distribution as a stationary distribution, but has periodicity of two.} 

The above remarks show that the convergence theorem is very convenient but does not cover all natural cases. To further handle this, we will combine the Markov chain and contraction properties. This will require new definitions. First, we will say a Markov chain $W$ is scrambling if $W$ is scrambling (See Definition~\ref{def:scrambling}). To handle more nuanced cases than scrambling, we introduce the following definition motivated by Definition~\ref{def:indecomposable-joint-dist}.
\begin{definition}\label{def:indecomposable-channel}
    Let $W$ admit a unique stationary distribution $\pmb{\pi}$. We say $W$ is indecomposable if the joint distribution $p_{XY}$ defined by $\pmb{\pi}$ and the output $W\pmb{\pi}$ is indecomposable.
\end{definition}

For clarity, we provide examples that highlight how these conditions may differ. 
\begin{example}[Scrambling and Indecomposable but not Irreducible]\label{ex:scrambling-not-irreducible}
    Consider $W$ such that $W(y \vert x) = \delta_{y',y}$ for all $y,y',x \in \cX$. That is, all inputs are mapped to a constant value $y'$. This is scrambling as every column of $W$ is the same (not orthogonal). It is not irreducible by definition as $W^{t}(y\vert x) > 0$ if and only if $y = x$. 
\end{example}

\begin{example}[Irreducible, Aperiodic, and Scrambling]\label{ex:irred-scramb-not-aperiod}
    Fix a finite alphabet $\cX$ such that $|\cX| \geq 3$. Let $W(x,x') = \frac{1}{|\cX|-1}(1-\delta_{x,x'})$ for all $x,x' \in \cX$, i.e. every input is uniformly randomly assigned to any value but itself. However, note that for any $x \neq x'$, $W(x,x') > 0$ and for any $x \in \cX$, $W^{2}(x,x) > 0$. Thus, by definition it is irreducible, as for any $x,y \in \cX$ there is $t$ such that $W^{t}(x,y) > 0$. It is aperiodic as $\cT(x)$ is all integers greater than two. It is scrambling as for any inputs $x \neq x' \in \cX$, there exists $y \in \cX$ such that $W(y|x),W(y|x') > 0$. 
\end{example}

\begin{example}[Irreducible, Aperiodic, Indecomposable, but not Scrambling]
    Consider the Markov chain $W$ on $\cX = [4]$
    $$W = \frac{1}{2} \begin{bmatrix} 1 & 1 & 0 & 0 \\ 0 & 1 & 1 & 0 \\ 0 & 0 & 1 & 1 \\ 1 & 0 & 0 & 1 \end{bmatrix} \ . $$
    This means that it maps its input either to itself or increases the value by one modulo 3.\footnote{This is an example of the `noisy typerwriter' channel \cite{cover2006}.} One may verify it is irreducible by verifying $W^{3}$ has all strictly positive values. By direct calculation, its stationary distribution is the uniform distribution. That it is indecomposable may be verified via the graph-theoretic equivalence given below Definition~\ref{def:indecomposable-joint-dist}. However, it is not scrambling as as the first and third columns are orthogonal.
\end{example}

We remark that it is claimed at the very end of \cite{GZB-2024a} that scrambling is a strictly weaker condition than being aperiodic in all dimensions greater than 2. We believe this claim does not make sense given the definition of period and examples such as Example~\ref{ex:irred-scramb-not-aperiod}. We believe the correct claim is that the joint conditions of irreducibility and aperiodicity are incomparable to the joint conditions of there existing a unique stationary distribution and being scrambling.

\subsection{Non-Linear Bounds for Contraction Coefficients and a Generalized Convergence Theorem}
In this section we establish that the rate of contraction does not vary for most $f$-divergences. This will make use of the following lemma which are non-linear bounds on the contraction coefficient.

\begin{lemma}\label{lem:classical-input-dependent-contraction-coeff-bounds}
Let $f$ be twice continuously differentiable over $(0,+\infty)$ and such that $L_{f} > 0$. Let either $f'(+\infty) = +\infty$ or $\mbf{q} > 0$. Let either $\vert f''(0) \vert < +\infty$ or $W\mbf{p}$ have full support. Then, for any channel $\cW$, the input-dependent contraction coefficient of the relevant $f$-divergence admits the following bounds:
$$ \eta_{f}(\cW,\mbf{q}) \leq \frac{4}{L_{f}\wt{q}_{\min}} \left[\sup_{\substack{\mbf{p} \in \cP(\cX):\\ 0 < D_{f}(\mbf{p}||\mbf{q}) < + \infty }}\kappa_{f}^{\uparrow}(W\mbf{p},W\mbf{q})\right] \cdot \eta_{\chi^{2}}(\cW,\mbf{q}) < +\infty \ . $$
\end{lemma}

Before providing the proof, we remark that the result may be seen as a non-linear generalization of bounds between contraction coefficients derived in \cite{Makur-2020a}. The reason we are able to generalize this is because our method for obtaining Pinsker inequalities generalizes the relevant conditions for applying the $f$-divergence Pinsker inequality from \cite{Gilardoni-2010a} and generalizing the conditions of upper bounding an $f$-divergence in terms of $\chi^{2}$-divergence from \cite{Raginsky-2016a} both of which \cite{Makur-2020a} relies upon.

\begin{proof}[Proof of Lemma~\ref{lem:classical-input-dependent-contraction-coeff-bounds}]
    The demands on the functions and support of the distributions are to guarantee we have the support conditions so that we can apply Theorem~\ref{thm:f-div-chi-squared-bounds}. If $f'(+\infty) = +\infty$, then $\mbf{p} \ll \mbf{q}$ is enforced by demanding $D_{f}(\mbf{p} \Vert \mbf{q}) < +\infty$. Similarly, if $\mbf{q} > 0$, then $\mbf{p} \ll \mbf{q}$. Then the further assumption $W\mbf{p} >0$ or $\vert f''(0) \vert < +\infty$ completes the conditions to apply Theorem~\ref{thm:f-div-chi-squared-bounds}. With this addressed, for any of the assumed conditions and the appropriate choice of $\mbf{p} \in \cP(\cX)$ we have
    \begin{align*}
        \frac{D_{f}(W\mbf{p}\Vert W\mbf{q})}{D_{f}(\mbf{p}\Vert \mbf{q})}
        \leq \frac{\kappa_{f}^{\uparrow}(W\mbf{p},W\mbf{q})\chi^{2}(W\mbf{p}\Vert W\mbf{q})}{\frac{L_{f}\wt{q}_{\min}}{4}\chi^{2}(\mbf{p}||\mbf{q})} 
        = \frac{4\kappa_{f}^{\uparrow}(W\mbf{p},W\mbf{q})\chi_{2}(W\mbf{p}\Vert W\mbf{q})}{L_{f,1}\wt{q}_{\min}\chi^{2}(\mbf{p}||\mbf{q})}  \ ,
    \end{align*}
    where we used the upper bound of Theorem~\ref{thm:f-div-chi-squared-bounds} in the numerator and the lower  bound in Lemma~\ref{lem:f-div-lb-by-chi-squared} in the denominator. We can now supremize over $\mbf{p}$ such that $0 < D_{f}(\mbf{p}\Vert\mbf{q}) < +\infty$.

    \begin{align*}
        & \sup_{\substack{\mbf{p} \in \cP(\cX):\\ 0 < D_{f}(\mbf{p}||\mbf{q}) < + \infty }} \frac{D_{f}(W\mbf{p}\Vert W\mbf{q})}{D_{f}(\mbf{p}\Vert \mbf{q})}\\
        \leq & \sup_{\substack{\mbf{p} \in \cP(\cX):\\ 0 < D_{f}(\mbf{p}||\mbf{q}) < + \infty }} \frac{4\kappa_{f}^{\uparrow}(W\mbf{p},W\mbf{q})\chi_{2}(W\mbf{p}\Vert W\mbf{q})}{L_{f,1}\wt{q}_{\min}\chi^{2}(\mbf{p}||\mbf{q})} \\
        \leq & \frac{4}{L_{f,1}\wt{q}_{\min}} \sup_{\substack{\mbf{p} \in \cP(\cX):\\ 0 < D_{f}(\mbf{p}||\mbf{q}) < + \infty }}\left[ \kappa_{f}^{\uparrow}(W\mbf{p},W\mbf{q})  \right] \cdot  \sup_{\substack{\mbf{p} \in \cP(\cX):\\ 0 < D_{f}(\mbf{p}||\mbf{q}) < + \infty }}\left[ \frac{\chi_{2}(W\mbf{p}\Vert W\mbf{q})}{\chi^{2}(\mbf{p}||\mbf{q})} \right] \\ 
        =&  \frac{4}{L_{f,1}\wt{q}_{\min}} \sup_{\substack{\mbf{p} \in \cP(\cX):\\ 0 < D_{f}(\mbf{p}||\mbf{q}) < + \infty }}\left[ \kappa_{f}^{\uparrow}(W\mbf{p},W\mbf{q}) \right] \cdot \eta_{\chi^{2}}(\cW,\mbf{q}) \ , 
    \end{align*}
    where the second inequality uses that both functions we supremize over are non-negative and the last equality is by definition of contraction coefficient. Using the definition of contraction coefficient and recalling the conditions in Theorem~\ref{thm:f-div-chi-squared-bounds} that guarantee $\kappa_{f}^{\uparrow}(W\mbf{p},W\mbf{q})$ is finite completes the proof.
\end{proof}

We now use our bounds on input-dependent contraction coefficients to when a time-homogeneous Markov chain $W$ admits a unique stationary distribution $\pmb{\pi}$ it converges to the stationary distribution at a rate of in terms of the $\chi^{2}(\cW,\pmb{\pi})$. This theorem significantly extends \cite[Proposition 7]{Makur-2020a} from being about the KL divergence to being about a large class of $f$-divergences. It also considers more general contraction settings. We are able to achieve this because of Lemma~\ref{lem:classical-input-dependent-contraction-coeff-bounds} being for a more general class of $f$-divergences than in \cite{Makur-2020a} as we explained above. 

\begin{theorem}\label{thm:classical-contraction-rate}
Consider any twice continuously differentiable convex function $f:(0,\infty) \to \mbb{R}$ such that $f(1)=0$, $f''(1) >0$, and $L_{f}>0$ is finite. Let $\cW$ have any matrix representation $W$ that has a unique stationary distribution $\pmb{\pi}$. Let $\vert f''(0) \vert < +\infty$ or there exist $n_{0} \in \mbb{N}$ such that $W^{n}\mbf{p} > 0$ for all $\mbf{p} \in \cP(\cX)$ and $n \geq n_{0}$. Let $W$ satisfy one of the following conditions:
     \begin{enumerate}
         \item $W$ is irreducible and aperiodic,
         \item $W$ is scrambling with either (a) $\pmb{\pi}$ full support or (b) $f'(+\infty) = +\infty$, 
         \item $W$ is indecomposable and $\pmb{\pi}$ is full support.
     \end{enumerate}
Then any distribution $\mbf{p} \in \cP(\cX)$ converges to the stationary state at a rate of at most $\eta_{\chi^{2}}(\cW,\pmb{\pi})$, i.e.
        \begin{align}
            \lim_{n \to \infty} \eta_{f}(W^{n},\pmb{\pi})^{1/n} \leq \eta_{\chi^{2}}(W,\pmb{\pi}) \ .
        \end{align}
        Moreover, if $\cW$ is reversible, the above bound is known to be tight.
\end{theorem}
We remark that at least condition 2 is not covered by condition 1 given Example~\ref{ex:scrambling-not-irreducible}. {We also remark that if the distribution is full support, then the $n_{0}$ exists as shown more generally in Lemma \ref{lem:full-rank-unique-means-strongly-mixing}.
\begin{proof}[Proof of Theorem~\ref{thm:classical-contraction-rate}]
    Let $W$ be the transition matrix representation of $\cW$ so that $\cW^{n}$ is represented by $W^{n}$. First, applying Lemma~\ref{lem:classical-input-dependent-contraction-coeff-bounds},
    \begin{align}
        \eta_{f}(\cW^{n},\pmb{\pi}) \leq & \frac{4}{L_{f}\pi_{\min}} \left[\sup_{\substack{\mbf{p} \in \cP(\cX):\\ 0 < D_{f}(W^{n}\mbf{p}||W^{n}\pmb{\pi}) < + \infty }}\kappa_{f}^{\uparrow}(W^{n}\mbf{p},W^{n}\pmb{\pi})\right] \cdot \eta_{\chi^{2}}(\cW^{n},\pmb{\pi})  \\
        \leq & \frac{4}{L_{f}\pi_{\min}} \left[\sup_{\substack{\mbf{p} \in \cP(\cX):\\ 0 < D_{f}(W^{n}\mbf{p}||W^{n}\pmb{\pi}) < + \infty }}\kappa_{f}^{\uparrow}(W^{n}\mbf{p},W^{n}\pmb{\pi})\right] \cdot \eta_{\chi^{2}}(\cW,\pmb{\pi})^{n}  \ ,
    \end{align}
    where the inequality is sub-multiplicativity of contraction coefficients (Proposition~\ref{prop:sub-multiplicativity}). We stress that for this to be finite, we must apply the conditions $\vert f''(0) \vert$ or $n \geq n_{0}$ as in the theorem statement so that we satisfy the conditions of Lemma~\ref{lem:classical-input-dependent-contraction-coeff-bounds}.
    
    At this point we just need to deal with the limiting behavior of the supremum term. In the case $W$ is irreducible and aperiodic, by the convergence theorem (Proposition~\ref{prop:MC-convergence}), $\lim_{n \to \infty} \text{TV}(W^{n}\bf{p},W^{n}\pmb{\pi}) = 0$. In the case $W$ is scrambling, by Item 2 of Proposition~\ref{prop:facts-about-contraction-coefficients}, we know that $\eta_{\text{TV}}(\cW) < 1$. Therefore, $0 \leq \lim_{n \to \infty}\mrm{TV}(W^{n}\mbf{p},W^{n}\pmb{\pi}) \leq \lim_{n \to \infty} \eta_{\text{TV}}(\cW)^{n}\text{TV}(\mbf{p},\pmb{\pi}) = 0$. In the case $W$ is indecomposable and $\pmb{\pi}$ is full support,
    \begin{align}
        \text{TV}(W^{n}\mbf{p},W^n \pmb{\pi}) 
        \leq & \frac{1}{2}\sqrt{\chi^{2}(W^{n}\mbf{p} \Vert W^{n} \pmb{\pi})} \\
        \leq & \frac{1}{2}\sqrt{\eta_{\chi^{2}}(\cW^{n},\pmb{\pi}) \chi^{2}(\mbf{p} \Vert \pmb{\pi})} \\
        \leq &  \frac{1}{2}\eta_{\chi^{2}}(\cW,\pmb{\pi})^{n/2} \sqrt{\chi^{2}(\mbf{p} \Vert \pmb{\pi})} < +\infty \label{eq:bound-on-TV}
    \end{align}
    where the first inequality used Corollary \ref{cor:pinsker-chisq}, the second inequality is by definition of the contraction coefficient, and the third uses Proposition~\ref{prop:sub-multiplicativity} and that $\pmb{\pi}$ is stationary. The strict inequality uses that $\pmb{\pi}$ is full support. Taking a limit with respect to $n$ again guarantees convergence in variational distance as before. Thus, in all cases, $\lim_{n \to \infty} W^{n}\mbf{p} = \pmb{\pi}$. It follows
    \begin{align*}
        \lim_{n \to \infty} \kappa_{f}^{\uparrow}(W^{n}\mbf{p},W^{n}\pmb{\pi}) =& \kappa_{f}^{\uparrow}(\pmb{\pi},\pmb{\pi}) = f''(1) \ , 
    \end{align*}
    where we used \eqref{eq:kappa-limits}.

    Finally, we take the power $1/n$ and take the limit with respect to $n$:
    \begin{align}
        \lim_{n \to \infty} \eta_{f}(\cW^{n},\pmb{\pi})^{n} \leq& \lim_{n \to \infty} \left(\frac{4}{L_{f}\pi_{\min}} \right)^{1/n} \left[\sup_{\substack{\mbf{p} \in \cP(\cX):\\ 0 < D_{f}(W^{n}\mbf{p}||W^{n}\pmb{\pi}) < + \infty }}\kappa_{f}^{\uparrow}(W^{n}\mbf{p},W^{n}\pmb{\pi})\right]^{1/n} \cdot \eta_{\chi^{2}}(\cW,\pmb{\pi}) \\
        \leq & \eta_{\chi^{2}}(\cW,\pmb{\pi}) \lim_{n \to \infty} \left[\sup_{\substack{\mbf{p} \in \cP(\cX):\\ 0 < D_{f}(W^{n}\mbf{p}||W^{n}\pmb{\pi}) < + \infty }}\kappa_{f}^{\uparrow}(W^{n}\mbf{p},W^{n}\pmb{\pi})\right]^{1/n} \\
        =& \eta_{\chi^{2}}(\cW,\pmb{\pi}) \ ,
    \end{align}
    where the second inequalities and the equality are standard limit laws for products and compositions of functions and that $\lim_{n \to \infty} c^{1/n} = 1$ for any $c > 0$. Note the equality relies on our assumption $f''(1) > 0$.  This completes the upper bound for Item 1.

    To show the moreover statement for Item 1, we first use that $\eta_{f}(\cW^{n},\pmb{\pi}) \geq \eta_{\chi^{2}}(\cW^{n},\pmb{\pi})$ by Item 3 of Proposition~\ref{prop:facts-about-contraction-coefficients}. It is then just a question of when the $\eta_{\chi^{2}}(\cW^{n},\pmb{\pi})$ coefficient is multiplicative rather than submultiplicative. \cite{Makur-2020a} showed this to be the case when $\cW$ is a reversible Markov chain. Thus, for reversible Markov chains, the above bound is tight. 
\end{proof}

We remark that Theorem~\ref{thm:classical-contraction-rate} not only shows that under a large class of $f$-divergences the rate at which one contracts is as fast as possible, but in fact this rate is \textit{computable} due to Proposition~\ref{prop:facts-about-contraction-coefficients}. Implicit in the above proof is also what may be seen as a non-trivial refinement of the convergence theorem in various settings as it provides explicit values of $\alpha,C$. 
\begin{corollary}[Refined Convergence Theorem]\label{cor:refined-convergence-theorem}
    For any time homogeneous Markov chain $W$ with a (possibly non-unique) stationary distribution $\pmb{\pi}$,
    \begin{align}
        4\text{TV}(W^{n}\mbf{p},\pmb{\pi})^{2} \leq  \eta_{\chi^{2}}(\cW,\pmb{\pi})^{n} \chi^{2} (\mbf{p} \Vert \pmb{\pi}) \ ,
    \end{align}
    which is non-trivial so long as $\mbf{p} \ll \pmb{\pi}$ and $\eta_{\chi^{2}}(\cW,\pmb{\pi}) < 1$. Moreover, if $\pmb{\pi}$ is full support,
    \begin{align}\label{eq:convergence-for-full-rank}
        4\text{TV}(W^{n}\mbf{p},\pmb{\pi})^{2} \leq\frac{2}{ \pi_{\min}} \eta_{\chi^{2}}(\cW,\pmb{\pi})^{n} \ .
    \end{align}
\end{corollary}
\begin{proof}
    Noting that Eq.~\eqref{eq:bound-on-TV} solely relies on the guarantee $\pmb{\pi}$ is a stationary distribution under $\cW$ and Corollary \ref{cor:pinsker-chisq} shows for all $\mbf{p}$, 
    \begin{align}
        \text{TV}(W^{n}\mbf{p},\pmb{\pi})^{2} 
        \leq \frac{1}{4} \eta_{\chi^{2}}(\cW,\pmb{\pi})^{n} \chi^{2} (\bf{p} \Vert \pmb{\pi}) \ ,
    \end{align}
    which is non-trivial so long as $\chi^{2}(\mbf{p} \Vert \pmb{\pi}) < +\infty$. Moreover, when $\pmb{\pi}$ is full support, we may apply \eqref{eq:chi-squared-reverse-Pinsker-Sason} to obtain $\chi^{2} (\mbf{p} \Vert \pmb{\pi}) \leq \frac{2}{\pi_{\min}}$. This gives us our bound.
\end{proof}
We remark that for irreducible, aperiodic Markov chains, which guarantees $\pmb{\pi}$ is full support, upper bounds of the form of \eqref{eq:convergence-for-full-rank} except with an improvement $\frac{2}{\pi_{\min}} \to \frac{1}{\pi_{\min}}$ have been found through different means \cite{Diaconis-1991a, Fill-1991a}. Moreover, those results themselves have even been extended to a class of quantum generalizations of the $\chi^{2}$-divergence \cite{Temme-2010a}. Similar results are also known for convergence of matrix multiplication \cite{Gross-1993a}. However, we highlight here that we have derived such a result without appealing to irreducibility directly and solely through divergence inequalities, i.e. information-theoretic methods.

\subsection{\texorpdfstring{$f$}{}-Divergence Mixing Times from Linear Contraction Bounds}
Mixing times tell us how many iterations of a Markov chain with a unique stationary distribution $\pmb{\pi}$ are necessary to be $\delta$-close to the stationary distribution $\pmb{\pi}$ under some measure of closeness $\Delta(\cdot,\cdot)$:
\begin{align}\label{eq:classical-mixing-time-def}
    t^{\Delta}_{\text{mix}}(\cW,\delta) \coloneq \min\{n \in \mbb{N}: \max_{\mbf{p} \in \cP(\cX)} \Delta(\cW^{n}(\mbf{p}),\pmb{\pi}) \leq \delta \} \ . 
\end{align}
Total variational distance is a standard measure of closeness for mixing times. For a statement about a mixing time to be truly useful, it ought to be computable. For total variational distance, Eq.~\ref{eq:convergence-for-full-rank} of Corollary~\ref{cor:refined-convergence-theorem} already implies that, given a Markov chain with a unique full support stationary distribution, we can do this by solving $\delta \geq \sqrt{\frac{1}{2\pi_{\min}}} \eta_{\chi^{2}}(\cW,\pmb{\pi})^{n/2}$ as a function of $n$ and then computing $\chi^{2}(\cW,\pmb{\pi})$. We state this as a corollary.
\begin{corollary}\label{cor:Markov-chain-mixing-times}
    Given a Markov chain $W$ with a full support stationary distribution $\pmb{\pi}$, whenever $\eta_{\chi^{2}}(W,\pmb{\pi}) < 1$,
    $$ t_{\text{mix}}^{\text{TV}}(\cW,\delta) \leq \frac{2\log(1/\sqrt{2\pi_{\min}}\delta)}{\log(1/\eta_{\chi^{2}}(\cW,\pmb{\pi}))} \ . $$
    Moreover, this is efficient to compute as $\chi^{2}(\cW,\pmb{\pi})$ is efficient to compute.
\end{corollary} 

However, it would be natural to ask if similar claims can be said for more stringent measures on the Markov chain. One such case would be to consider measuring distance under KL divergence, which we know can be strictly larger than one, unlike trace distance. More generally, in principle, one could measure Markov chain mixing times under arbitrary $f$-divergences. Here we will show that, for a certain class of $f$-divergences, the mixing times only mildly vary as a function of the choice of $f$, which is an immediate corollary of \cite[Lemma A.2]{Raginsky-2016a} (See \cite[Lemma 6]{Makur-2020a} for a notationally clearer proof).
\begin{proposition}\label{prop:f-div-mixing-time}
    $f:(0,\infty) \to \mbb{R}$ be convex, differentiable at unity with $f(1) = 0$, $f(0) < \infty$, and $g(t) \coloneq \frac{f(t)-f(0)}{t}$ concave on $(0,\infty)$. Let $W$ be a Markov chain with unique full support stationary distribution $\pmb{\pi}$. Then, whenever $\eta_{\chi^{2}}(W,\pmb{\pi}) < 1$ ,
    \begin{align}
       t_{\text{mix}}^{D_{f}}(\cW,\delta) \leq \frac{\log(2/[\delta \pi_{\min}]) + \log(f'(1)+f(0))}{\log(1/\eta_{\chi^{2}}(W,\pmb{\pi}))}  \ .
    \end{align}
\end{proposition}
\begin{proof}
    We are interested in $D_{f}(W^{n}\mbf{p} \Vert \pmb{\pi})$ for arbitrary $\mbf{p}$. Now,
    \begin{align}
        D_{f}(W^{n}\mbf{p} \Vert \pmb{\pi}) \leq& [f'(1) + f(0)]\chi^{2}(W^{n}\mbf{p} \Vert \pmb{\pi}) \\
        \leq& [f'(1) + f(0)]\chi^{2}(W^{n}\mbf{p} \Vert W^{n}\pmb{\pi}) \\
        \leq& [f'(1) + f(0)]\eta_{\chi^{2}}(W,\pmb{\pi})^{n}\chi^{2}(\mbf{p} \Vert \pmb{\pi}) \\
        \leq& \frac{2 [f'(1) + f(0)]}{\pi_{\min}}\eta_{\chi^{2}}(W,\pmb{\pi})^{n} \ . \label{eq:linear-upper-bound-on-f-dist-from-stationary} \ ,
    \end{align}
    where the first inequality is \cite[Lemma A.2]{Raginsky-2016a}, the second is that $\pmb{\pi}$ is stationary, the third is multiplicativity, and the fourth is \eqref{eq:chi-squared-reverse-Pinsker-Sason} where we have relied on $\pmb{\pi}$ being full support. It then suffices to do the arithmetic to find $n$ such that \eqref{eq:linear-upper-bound-on-f-dist-from-stationary} is upper bounded by $\delta$, which always exists so long as $\eta_{\chi^{2}}(W,\pmb{\pi}) < 1$.
\end{proof}
We note that mixing times for $f$-divergences were considered in \cite[Proposition 5.1]{Raginsky-2016a} where for $f$ such that $f(0) < +\infty$ and irreducible, aperiodic $W$ where $W(x\vert x) > 0$, 
\begin{align}\label{eq:Raginsky-f-mixing-time}
    \log(\omega_{f,\cX}/\delta)/\log(1/\eta_{f}(W,\pmb{\pi})) \ ,
\end{align}
where $\omega_{f,\cX} := f(\vert \cX \vert)/|\cX| + (1-1/\vert \cX \vert)f(0)$.
Clearly this result and Proposition~\ref{prop:f-div-mixing-time} are not directly comparable. First, the conditions on both $W$ and $f$ differ. Second, the denominator in \cite[Proposition 5.1]{Raginsky-2016a} is scales as $\eta_{f}(W,\mbf{p})$, which means the denominator can be smaller than that of Proposition~\ref{prop:f-div-mixing-time}. On the other hand, the numerator is parameterized in terms of $f$ and is a function of the size of the alphabet. Of course, the major advantage of Proposition~\ref{prop:f-div-mixing-time} over \cite[Proposition 5.1]{Raginsky-2016a} is that it is efficient to compute.

\paragraph{Linear Contraction Coefficient Bounds}
Given that our mixing time bounds for total variation distance were deeply related to Theorem~\ref{thm:classical-contraction-rate}, it would be reasonable to ask if a contraction coefficient relation also may be established. Indeed, this was established in \cite[Theorem 3]{Makur-2020a} in both statement and proof method. We provide a slight variation of this result that follows from our Pinsker's inequality methodology.
\begin{proposition}\label{prop:linear-contraction-bounds}
    Let $f:(0,\infty) \to \mbb{R}$ be convex, differentiable at unity with $f(1) = 0$, $f(0) < \infty$, and $g(t) \coloneq \frac{f(t)-f(0)}{t}$ concave on $(0,\infty)$. Let $\mbf{q}$ be full support. Then,
    \begin{align}
        \eta_{f}(W,\mbf{q}) \leq \frac{4[f'(1)+f(0)]}{L_{f} q_{\min}} \eta_{\chi^{2}}(W,\mbf{q}) \ ,
    \end{align}
    where $L_{f}$ is defined as in Theorems \ref{thm:pinsker-uni} and \ref{thm:pinsker-multi}.
\end{proposition}
\begin{proof}
    In the case $L_{f} = 0$, this bound is trivial if we define $a/0 \coloneq +\infty$, so we assume $L_{f} > 0$.
    \begin{align}
        \frac{D_{f}(W\mbf{p} \Vert W\mbf{q})}{D_{f}(\mbf{p}\Vert \mbf{q})} \leq \frac{(f'(1) + f(0))\chi^{2}(\mbf{p} \Vert \mbf{q})}{\frac{L_{f}\wt{q}_{\min}}{4} \chi^{2}(\mbf{p} \Vert \mbf{q})^{2}} \ ,
    \end{align}
    where the numerator makes use of \cite[Lemma A.2]{Raginsky-2016a} (See also \cite[Lemma 6]{Makur-2020a} for a direct proof) and Lemma~\ref{lem:f-div-lb-by-chi-squared} for the denominator. As this is independent of $\mbf{p} \in \cP(\cX)$, supremizing completes the proof.
\end{proof}
We stress that this is exactly the result of \cite[Theorem 3]{Makur-2020a} except that we have relaxed the demands on $f$ as we have replaced their application of Pinsker's inequality using Proposition~\ref{prop:Gilardoni} to using our Pinsker's inequality. The advantage of using Proposition~\ref{prop:Gilardoni} is that $4/L_{f}$ may be replaced with $f''(1)$ which is known to be a tight Pinsker's inequality. The advantage of our method is generality. Of course, when both methods apply, if $L_{f} = f''(1)$, then the results are equivalent.

\begin{example}
    Consider by the binary symmetric channel with parameter $p$, $W_{p} = \begin{bmatrix} 1-p & p \\ p & 1-p \end{bmatrix}$. Its unique stationary distribution is the uniform distribution as is direct to verify. It is well-known that $\eta_{\chi^{2}}(W_{p},\pmb{\pi}) = (1-2p)^{2}$ as may be verified by the singular value characterization. For the Hellinger divergences, $f_{\alpha}(t) \coloneq \frac{t^{\alpha}-1}{\alpha-1}$, of $\alpha \in (1,2)$. It follows $g(t) = [f_{\alpha}(t) - f_{\alpha}(0)]/t = \frac{t^{\alpha}}{\alpha-1}$ and thus $g(t)$ is concave over $(0,\infty)$ for $\alpha \in (1,2)$ by the second order criterion. Then we may apply Proposition~\ref{prop:f-div-mixing-time} to obtain have 
    \begin{align}
        \eta_{f_{\alpha}}(W_{p},\pmb{\pi}) \leq 2\eta_{\chi^{2}}(W_{p},\pmb{\pi}) = 2(1-2p)^{2} \quad \forall \alpha \in (1,2) \, \forall p \in [0,1] \ , 
    \end{align}
    where we used Table \ref{tab:uni}.
\end{example}

\section{Extending Results to Quantum \texorpdfstring{$f$}{}-Divergences}\label{sec:quantum-extensions}
In this section we explain how various results we have established extend to classes of quantum $f$-divergences. In particular, we show that for time homogeneous quantum Markov chains, the rate of contraction for many monotonic Petz $f$-divergences is the Petz $\chi^{2}$-divergence, which we denote by $\ol{\chi}^{2}$ (Theorem~\ref{thm:Petz-contraction-coefficient-rate}). We also establish mixing times under various Petz $f$-divergences in terms of the input-dependent Petz $\chi^{2}$ contraction coefficient (Proposition~\ref{prop:Petz-mixing-times}). Before establishing these results, for clarity to those unfamiliar, we begin with a background on quantum divergences that will be sufficient for understanding the rest of the section. We refer the reader to standard texts \cite{Wilde-Book,WatrousBook,Tomamichel-2016a,Khatri-2020a} for further background.

\subsection{Background: Quantum Divergences}
Quantum divergences extend the notion of divergence from probability distributions and vectors to quantum states and positive semidefinite operators. We denote the space of positive semidefinite operators on a Hilbert space $A$ by $\Pos(A)$. For two Hermitian oeprators $R,S \in \Herm(A)$, we write $R \ll S$ if the kernel of $R$ contains the kernel of $S$. We also write $R \leq S$ if $R-S \in \Pos(A)$, which is known as the L\"{o}wner order. The space of density matrices (quantum states) is denoted by $\Density(A) \coloneq \{\rho \in \Pos(A) : \Tr[\rho] = 1\}$. The generalization of total variation to the quantum setting is the trace distance: $\text{TD}(P,Q) \coloneq \frac{1}{2}\Vert P - Q \Vert_{1}$ where $\Vert \cdot \Vert_{1}$ is the Schatten $1-$norm. Quantum channels are completely positive, trace-preserving (CPTP) linear maps. We denote the set of CPTP maps from $\Lin(A)$ to $\Lin(B)$ by $\Channel(A,B)$. We will denote the iterative composition of a channel by $\cE^{n} \coloneq \circ_{i \in [n]} \cE$.

The most well known quantum divergence is the Umegaki relative entropy, which is often just called the relative entropy, and is defined as follows.
\begin{definition}
    Let $P,Q \in \Pos(A)$. The Umegaki relative entropy is
    \begin{align*}
        D(P\Vert Q) := \begin{cases}
            \Tr[P\log(P)] - \Tr[P\log(Q)] & P \ll Q \\
            +\infty & \text{otherwise} \ .
        \end{cases}
    \end{align*}
\end{definition}
In the same way that $f$-divergences are motivated by the KL divergence, so too are quantum $f$-divergences. However, because operators do not commute, there are a variety of generalizations of both $f$-divergences, e.g. \cite{Petz-1985a,Petz-1986a,Petz-1998a,Wilde-2018a,Hirche-2023a}. In this section we will primarily be focused on quantum divergences in two manners. The first is at a very abstract level, where we only care about the divergence satisfying two properties: data-processing and the ability to reduce to a classical $f$-divergence. The many families of such quantum $f$-divergences listed above all satisfy these properties.
\begin{definition}\label{def:gen-quantum-f-divergence}
    We say a functional $\Pos(A) \times \Pos(B) \to \mbb{R}$ is a quantum $f$-divergence, denoted $\mbb{D}_{f}(P\Vert Q)$, if it satisfies the following two properties:
    \begin{enumerate}
        \item It is monotonic under the action of a quantum channel on quantum states, e.g. for all Hilbert spaces $A,B$, all quantum states $ \rho,\sigma \in \Density(A)$, and all quantum channels $\cE \in \Channel(A,B)$, 
        \begin{align}
            \mbb{D}_{f}(\cE(\rho)\Vert \cE(\sigma)) \leq \mbb{D}_{f}(\rho \Vert \sigma) \ .
        \end{align}
        \item On all classical (i.e. diagonal) states, $\rho_{X} = \sum_{x} p_{x} \dyad{x} ,\sigma_{X} = \sum_{x} q_{x} \dyad{x}$ where $\{\ket{x}\}_{x \in \cX}$ form a basis of $A$ and $\mbb{p},\mbb{q} \in \cP(\cX)$, the $f$-divergence simplifies to the classical $f$-divergence,
        \begin{align}
            \mbb{D}_{f}(\rho_{X}\Vert \sigma_{X}) = D_{f}(\mbf{p}\Vert \mbf{q}) \ .
        \end{align}
    \end{enumerate}
\end{definition}

The second manner in which we will care about quantum $f$-divergences is in terms of a specific family introduced by Petz \cite{Petz-1985a,Petz-1986a} and have been investigated a great deal subsequently, e.g. \cite{Hiai-2011a,Hiai-2016a,Hiai-2017a,Matsumoto-2018a}. We call these the Petz $f$-divergences, though sometimes they are called the `standard' $f$-divergences \cite{Hiai-2017a}, since technically Petz's `quasi-entropies' include a more general class. 
\begin{definition}\label{def:petz-f-divergence}
    \cite{Hiai-2017a} For $P,Q \in \mrm{Pd}(A)$ the Petz $f$-divergence is defined as 
    \begin{align*}
        \overline{D}_{f}(P \Vert Q) := \Tr[Q^{-1/2}f(L_{P}R_{Q^{-1}})Q^{-1/2}] \ , 
    \end{align*}
    where we use the left and right multiplication operators, $L_{W}(X) = WX$ $R_{W}(X) = XW$ for all $X,W \in \Lin(A)$ Moreover, the divergences are extended to $P,Q \in \Pos(A)$ via continuity:
    \begin{align}
        \ol{D}_{f}(P\Vert Q) := \lim_{\ve \downarrow 0} S_{f}(P + \ve I \Vert Q + \ve I) \ .
    \end{align}
\end{definition}

We note that while the above is the technical definition, they can be handled much more straightforwardly by their equivalence to the $f$-divergence evaluated on special induced probability distributions.
\begin{proposition}\cite{Hiai-2017a}\label{prop:petz-f-with-NS-dists}
    Let $P,Q \in \Pos(A)$ such that they admit spectral decompositions $P = \sum_{x} \lambda_{x} \dyad{e_{x}}$, $Q = \sum_{y} \mu_{y} \dyad{f_{y}}$. Then
    \begin{align}\label{eq:Petz-f-expansion}
        \ol{D}_{f}(P \Vert Q) = \sum_{x: \lambda_{x} > 0} \sum_{y: \lambda_{y} > 0} \mu_{y}f\left(\frac{\lambda_{x}}{\mu_{y}} \right) |\bra{e_{x}}\ket{f_{y}}|^{2} + f(0+)\Tr[(I - P^{0})Q] + f'(+\infty)\Tr[P(I-Q^{0})] \ .
    \end{align}

    Moreover, for quantum states $\rho,\sigma \in \Density(A)$
    \begin{align}\label{eq:Petz-f-to-NS-classical-f}
         \ol{D}_{f}(\rho \Vert \sigma) = D_{f}(\mbf{p}^{[\rho,\sigma]}_{XY}\Vert \mbf{q}^{[\rho,\sigma]}_{XY}) \ , 
    \end{align}
    where
    \begin{align}\label{eq:NS-distributions}
         \mbf{p}^{[\rho,\sigma]}_{XY}(x,y) \coloneqq \lambda_{x}\vert \bra{e_{x}}\ket{f_{y}} \vert^{2} \quad \mbf{q}^{[\rho,\sigma]}_{XY}(x,y) := \mu_{y}\vert \bra{e_{x}}\ket{f_{y}}\vert^{2} \ , 
    \end{align}
    which are known as the Nussbaum-Szko\l a (NS) distributions \cite{Nussbaum-2009a}.
\end{proposition}
We remark that if $\rho \ll \sigma$ (resp.~, $\rho  \ll \gg \sigma$) $\mbf{p}^{[\rho,\sigma]}_{XY} \ll \mbf{q}^{[\rho,\sigma]}_{XY}$ (resp.~$\mbf{p}^{[\rho,\sigma]}_{XY} \ll \gg \mbf{q}^{[\rho,\sigma]}_{XY}$). This follows from the distributions using the same `overlaps' $\vert\bra{e_{x}}\ket{f_{y}}\vert$ as if $\mu_{y} = 0$, then if $\lambda_{x} \neq 0$ and  $\vert \bra{e_{x}}\ket{f_{y}}\vert \neq 0$, then we contradict $\rho \ll \sigma$. \\

The most important property of the Petz $f$-divergences are that they satisfy data processing when $f$ is operator convex.
\begin{fact} 
\cite{Petz-1986a,Hiai-2011a} (See also \cite{Tomamichel-2009a,Tomamichel-2016a,Hiai-2017a}.) If $f$ is operator convex, then $\ol{D}_{f}(\cE(\rho)\Vert \cE(\sigma)) \leq \ol{D}_{f}(\rho \Vert \sigma)$ for all $\rho,\sigma \in \Density(A)$, $\cE \in \Channel(A,B)$.
\end{fact}

\paragraph{Petz \texorpdfstring{$\chi^{2}$}{}-Divergence}
We will make use of the quantum $\ol{\chi}^{2}$-divergence \cite{Temme-2010a}. We begin by defining it on full rank states. That is, for $\rho,\sigma \in \Pd(A)$, we initially define the quantum $\ol{\chi}^{2}$-divergence as 
    \begin{align}\label{eq:quantum-chi-squared-divergence-init}
        \ol{\chi}^{2}(\rho\Vert \sigma) = \Tr[\sigma^{-1}(\rho-\sigma)^{2}] \ .
    \end{align}
We note that \cite{Temme-2010a} in fact considered a whole family of quantum $\chi^{2}$-divergences of which $\ol{\chi}^{2}$ is a special case. We stress that \eqref{eq:quantum-chi-squared-divergence-init} is \textit{not} what is denoted the quantum $\chi^{2}$-divergence in \cite{Hirche-2023a} (See \cite{Hirche-2023a} and \cite[Example 1]{Lesniewski-1999a}). We remark there are other manners of obtaining the $\ol{\chi}^{2}$-divergence \cite[Section 2]{Petz-1998a}. We also note that \eqref{eq:quantum-chi-squared-divergence-init} is in fact the Petz $f$-divergence for $f(x) := x^{2}-1$ at least on full rank density matrices, as may be verified via the following rather direct calculation:
\begin{align*}
    \ol{D}_{x^{2}-1}(P \Vert Q) =& \sum_{x,y} \mu_{y}\left[ \frac{\lambda_{x}^{2}}{\mu_{y}^{2}}-1  \right]\left|\bra{e_{x}}\ket{f_{y}}\right|^{2} \\
    =& \sum_{x,y} \frac{\lambda_{x}^{2}}{\mu_{y}} \left|\bra{e_{x}}\ket{f_{y}}\right|^{2} - \sum_{x,y} \mu_{y}  \left|\bra{e_{x}}\ket{f_{y}}\right|^{2} \\
    =& \sum_{x,y} \frac{\lambda_{x}^{2}}{\mu_{y}} \left|\bra{e_{x}}\ket{f_{y}}\right|^{2} - 1 \\
    =& \Tr[\sigma^{-1}\rho^{2}] + \Tr[\sigma] -2\Tr[\rho] \\
    =& \Tr[\sigma^{-1}(\rho^{2} + \sigma^{2} -\rho\sigma - \sigma\rho)] \\
    =& \Tr[\sigma^{-1}(\rho-\sigma)^{2}] \\
    =& \ol{\chi}^{2}(\rho \Vert \sigma) \ ,
\end{align*}
where the first equality is Proposition~\ref{prop:petz-f-with-NS-dists},  the second is expanding terms, the third uses the normalization of $\sigma$ and that $\{\ket{e_{x}}\},\{\ket{f_{y}}\}$ are orthonormal bases, the fourth uses $\Tr[\rho] = \Tr[\sigma] = 1$, the fifth is a direct calculation of $\Tr[\sigma^{-1}\rho^{2}]$, the sixth is a direct calculation of $\Tr[\sigma^{-1}(\rho^{2} + \sigma^{2} -\rho\sigma - \sigma\rho)]$ using cyclicity of trace, the seventh is again direct calculation, and the last is \eqref{eq:quantum-chi-squared-divergence-init}.
As such, we formally define the Petz $\ol{\chi}^{2}$-divergence. We keep the notation with the bar to both avoid confusion with \cite{Hirche-2023a} and to remind the reader the quantity is induced by the Petz $f$-divergences.
\begin{definition}
      Let $P,Q \in \Pos(A)$. Then we define the Petz $\chi^{2}$-divergence as
      \begin{align}
        \ol{\chi}^{2}(P \Vert Q) := \ol{D}_{x^{2}-1}(P \Vert Q) \ , 
      \end{align}
      where we use the definition of Petz $f$-divergence (Definition~\ref{def:petz-f-divergence}).
\end{definition}

\subsection{Extension of Results}
We now turn to extending our classical results to quantum $f$-divergences. \\

\subsubsection{Pinsker Inequalities for Quantum \texorpdfstring{$f$}{}-Divergences}
We first establish our Pinsker inequalities naturally lift to general quantum $f$-divergences via the data processing inequality. We note this not a new insight. In particular, this is how the Pinsker inequality for Umegaki relative entropy is established (See \cite{Wilde-Book}). The contribution is simply that this now applies to a large class of quantum $f$-divergences. The result is tight or sharp whenever it is classically.
\begin{corollary}\label{cor:quantum-f-divergence-pinsker}
    Let $\mbb{D}_{f}(\rho \Vert \sigma)$ be defined as in Definition~\ref{def:gen-quantum-f-divergence} for some $f$ that is continuously twice differentiable. Then
    $$ \mbb{D}_{f}(\rho \Vert \sigma) \geq \frac{L_{f}}{2} \text{TD}(\rho,\sigma)^{2} \ , $$
    where $L_{f}$ is as defined in Theorems \ref{thm:pinsker-uni} and \ref{thm:pinsker-multi}. Moreover, if the divergence satisfies DPI on positive operators, this extends to $P,Q$ with same trace.
\end{corollary}
\begin{proof}
    This is a straightforward extension of lifting the Pinsker inequality to quantum that has been used previously, e.g. \cite{Wilde-Book}. Note that $\Vert\rho - \sigma\Vert_{1} = \Tr[\Pi_{+}(\rho - \sigma)] + \Tr[(I-\Pi_{+})(\rho-\sigma)]$ where $\Pi_{+}$ is the projector onto the positive eigenspace of $(\rho-\sigma)$. Define the quantum-to-classical channel
    $$ \cE(X) := \Tr[\Pi_{+}X] \otimes \dyad{0} + \Tr[(I-\Pi_{+})X] \otimes \dyad{1} \ . $$
    It follows $\Vert \rho - \sigma \Vert_{1} = \Vert \cE(\rho) - \cE(\sigma) \Vert_{1}$ as may be verified by direct calculation. It follows
    \begin{align*}
        \mbb{D}_{f}(\rho \Vert \sigma) \geq& D_{f}(\cE(\rho) \Vert \cE(\sigma)) \geq \frac{L_{f}}{2}\mrm{TD}(\cE(\rho),\cE(\sigma))^{2} = \frac{L_{f}}{2}\mrm{TD}(\rho,\sigma)^{2} \ ,
    \end{align*}
    where we used the assumed data processing inequality. The moreover statement is simply noting the above proof works identically if $P,Q$ have the same trace and the data processing inequality still holds. This completes the proof.
\end{proof}

\subsubsection{Strong Data Processing and Ergodicity for Petz \texorpdfstring{$f$}{}-Divergences}
As the $f$-divergences naturally lift to the Petz $f$-divergences by \eqref{eq:Petz-f-to-NS-classical-f}, one would expect that we might obtain similar results to those in Section~\ref{sec:classical-Input-Dependent-SDPI}. Indeed, this is the case. We start with the appropriate definitions of quantum contraction coefficients. We then work out the theory of ergodicity for Petz $f$-divergences.

\paragraph{Quantum Contraction Coefficients and Ergodic Definitions}
 We begin with the definition of an input-dependent contraction coefficient. We remark various quantum contraction coefficients have been considered previously, e.g. \cite{Ruskai-1994a,Lesniewski-1999a,Temme-2010a,Hirche-2022a,Hirche-2023-dp,Nuradha-2024b,Cheng-2024a}. 

\begin{definition}\label{def:quantum-contraction-coefficients}
    Let $\cE \in \Channel(A,B)$ be a quantum channel. The \textit{input-independent contraction coefficient}
    \begin{align}
       \eta_{f}(\cE) := \sup_{\sigma,\rho \in \Density(A) \, : \, 0 < \mbb{D}_{f}(\rho||\sigma) < + \infty } \frac{\mbb{D}_{f}(\cE(\rho)||\cE(\sigma)}{\mbb{D}_{f}(\rho||\sigma)}  \ .
    \end{align} The \textit{input-dependent contraction coefficient} is
    \begin{align}\label{eq:quantum-input-dependent-contraction-coeff}
        \eta_{f}(\cE,\sigma) := \sup_{\rho \, : \, 0 < \mbb{D}_{f}(\rho||\sigma) < + \infty } \frac{\mbb{D}_{f}(\cE(\rho)||\cE(\sigma))}{\mbb{D}_{f}(\rho||\sigma)}  \ .
    \end{align}
\end{definition}

\begin{proposition}\label{prop:sub-multiplicativity}
    Consider $\cE \in \Channel(A,B)$, $\cF \in \Channel(B,C)$. Then, assuming $\mbb{D}_{f}$ satisfies data processing,
    \begin{align*}
        \eta_{f}(\cF \circ \cE,\sigma) &\leq \eta_{f}(\cF, \cE(\sigma))\eta_{f}(\cE,\sigma)  \\
        \eta_{f}(\cF \circ \cE) &\leq \eta_{f}(\cF)\eta_{f}(\cE) \ .
    \end{align*}
    The same claims holds for classical $f$-divergences when considering classical channels and classical inputs.
\end{proposition}
\begin{proof}
    The proof is identical to \cite[Lemma 3.6]{Hirche-2022a} except we replace the relative entropy with an arbitrary $f$-divergence that satisfies data processing over the relevant set of objects (classical distributions or quantum states).
\end{proof}

A natural question would be whether for a classical-to-classical channel $\cW$ and distribution $q$, the quantum contraction coefficient reduces to the classical (Definition~\ref{def:classical-contraction-coefficients}). This seems to be unclear. Certainly whenever $\mbb{D}_{f}$ satisfies data processing this is true as the following proposition implies (See Appendix for its proof).
\begin{proposition}\label{prop:classical-contraction-coeff-from-quantum}
    Let $f$ be strictly convex at unity and such that $f'(\infty) = +\infty$. Let $\cW_{X \to Y}$ be a classical-to-classical channel where the input classical space is defined in terms of orthonormal basis $\{\ket{i}\}_{i \in \cX}$. Let $\Delta(Z) := \sum_{i \in \cX} \dyad{i}(Z)\dyad{i}$ be the completely dephasing channel with respect to the given computational basis. Then if $\mbb{D}_{f}$ is a quantum $f$-divergence such that $\mbb{D}_{f}(\Delta (\rho) \Vert \Delta(\sigma)) \leq D(\rho \Vert \sigma)$ for all $\rho,\sigma \in \Density(\mbb{C}^{|\cX|})$, i.e. satisfies the data-processing inequality for the completely dephasing channel, then the quantum definitions of contraction coefficient reduce to the classical definitions.
\end{proposition}
\noindent However, the natural open question would be if Proposition~\ref{prop:classical-contraction-coeff-from-quantum} applies for Petz $f$-divergences when $f$ is not operator convex so that data processing does not hold. This seems unlikely given that the choice of orthonormal basis in Proposition~\ref{prop:classical-contraction-coeff-from-quantum} is in a sense arbitrary, so one would expect one would need $f$ to be contractive under dephasing in every basis. This is equivalent to being contractive under arbitrary ``pinching maps." However, operator convexity is generally equivalent to being contractive under arbitrary ``pinching maps" (See \cite[Theorem V.2.1]{Bhatia-1997}).

In the same way that contraction coefficients were related to scrambling and indecomposability for classical channels, one would expect we need similar ideas for ergodicity in the quantum regime. In the quantum setting, the definitions seem less standard. Perhaps the closest language may be found in \cite{Wolf-2012a}. To capture what we formally want, we will make use of the following, less concrete, definition.
\begin{definition}(See e.g. \cite{Burgarth-2013a})\label{def:mixing} A quantum channel $\cE_{A \to A}$ is called mixing if it admits a unique fixed point state $\pi \in \Density(A)$ such that
\begin{align*}
    \lim_{n \to \infty} \Vert \cE^{n}(\rho) - \pi \Vert_{1} = 0 \quad \forall \rho \in \Density(A) \ .
\end{align*}
Moreover, we say $\cE$ is strongly-mixing if there is $n_{0} \in \mbb{N}$ such that $\cE^{n}(\rho) > 0$ for all $n \geq n_{0}$ and $\rho \in \Density(A)$.
\end{definition}
Some properties of the limiting behaviour of quantum channels is known. For example, it is know that if there exists any time step $n \in \mbb{N}$ such that $\cE^{n}(\rho) \in \Pd(A)$ for all $\rho \in \Density(A)$, then $\cE$ is mixing and it stationary distribution $\pi$ is full rank \cite[Theorem 6.7]{Wolf-2012a}. Such channels are called `primitive' in \cite{Wolf-2012a} (and subsequently \cite{Hirche-2022a}), but seem the natural quantum extension of being both irreducible and aperiodic. We slightly strengthen this correspondence by showing a quantum channel is `primitive' if and only if it is strongly mixing.\footnote{The proof does not rely upon complete positivity, so we actually establish the result for positive, trace preserving maps.}
\begin{lemma}\label{lem:full-rank-unique-means-strongly-mixing}
    A quantum channel $\cE$ is mixing with unique fixed point $\pi$ that is also full rank if and only if it is strongly mixing.
\end{lemma}
\begin{proof}
   ($\leftarrow$) If it is strongly mixing, then the fixed point $\pi$ must be full rank as $\pi = \cE^{n}(\pi)$ for all $n \in \mbb{N}$. Thus, any strongly mixing channel is a mixing channel with unique fixed point $\pi$ that is full rank. \\
   ($\rightarrow$) Let $\cE$ be mixing with unique fixed point $\pi$ that is also full rank. It thus satisfies Item 3 of \cite[Theorem 6.7]{Wolf-2012a}. By \cite[Theorem 6.7]{Wolf-2012a}, there is $n_{0} \in \mbb{N}$ such that for all $\rho \in \Density(A)$, $\cE^{n_{0}}(\rho) > 0$. Then for each $\rho \in \Density(A)$, there exists $\lambda_{\rho} > 0$ and $R_{\rho} \geq 0$ such that $\cE^{n_{0}}(\rho) = \lambda_{\rho} \pi + R_{\rho}$. Note that by the trace condition, $R_{\rho} = (1-\lambda_{\rho})\sigma_{\rho}$ where $\sigma_{\rho} \in \Density(A)$. Now for all $n \geq n_{0}$ and $\rho \in \Density(A)$,
   $$\cE^{n}(\rho) = \cE^{n-n_{0}} \circ \cE^{n_{0}}(\rho) = \cE^{n-n_{0}}(\lambda_{\rho}\pi + (1-\lambda_{\rho})(\sigma_{\rho})) = \lambda_{\rho}\pi + (1-\lambda_{\rho})\cE^{n-n_{0}}(\sigma_{\rho}) \geq \lambda_{\rho}\pi > 0 \ , $$
   where we used $\pi$ is the fixed point and full rank and that $\cE$ is a positive map. This completes the proof.
\end{proof}

\paragraph{Petz \texorpdfstring{$f$}{}-Divergence and Contraction Coefficient Inequalities}
We now establish our Petz $f$-divergence inequalities and the related contraction coefficient inequalities. We begin with the bounds on the Petz $f$-divergences.
\begin{corollary}\label{cor:Petz-f-chi-squared-bounds}
    Let $f$ be twice-differentiable and convex over $(0,\infty)$. Let $\rho,\sigma \in \Density(A)$ such that $\rho \ll \sigma$. If $\vert f''(0) \vert < +\infty$ or $\sigma \ll \rho$, then 
    \begin{align*}
        \kappa_{f}^{\downarrow}(\mbf{p}^{[\rho,\sigma]}_{XY},\mbf{q}^{[\rho,\sigma]}_{XY})\ol{\chi}_{2}(\rho\Vert \sigma) \leq \ol{D}_{f}(\rho\Vert \sigma) \leq \kappa_{f}^{\uparrow}(\mbf{p}^{[\rho,\sigma]}_{XY},\mbf{q}^{[\rho,\sigma]}_{XY})\ol{\chi}_{2}(\rho\Vert \sigma) \ ,
    \end{align*}
    where $\kappa_{f}^{\downarrow},\kappa_{f}^{\uparrow}$ are defined in \eqref{eq:kappa-up-arrow-defn}, \eqref{eq:kappa-down-arrow-defn} respectively and $\mbf{p}^{[\rho,\sigma]}_{XY}, \mbf{q}^{[\rho,\sigma]}_{XY}$ are the Nussbaum-Sko\l a distributions \eqref{eq:NS-distributions}.
\end{corollary}
\begin{proof}
    As $\rho \ll \sigma$, one may verify directly from \eqref{eq:NS-distributions} that $\mbf{p}^{[\rho,\sigma]}_{XY} \ll \mbf{q}^{[\rho,\sigma]}_{XY}$. By \eqref{eq:Petz-f-to-NS-classical-f}, we know $\ol{D}_{f}(\rho \Vert \sigma) = D_{f}(\mbf{p}^{[\rho,\sigma]}_{XY}\Vert \mbf{q}^{[\rho,\sigma]}_{XY})$. Further using $\vert f''(0) \vert <+\infty$ or $\sigma \ll \rho$, we may then apply Theorem~\ref{thm:f-div-chi-squared-bounds} to $D_{f}(\mbf{p}^{[\rho,\sigma]}_{XY}\Vert \mbf{q}^{[\rho,\sigma]}_{XY})$. Finally, we again apply \eqref{eq:Petz-f-to-NS-classical-f} to get back to the Petz quantum $f$-divergence. This completes the proof.
\end{proof}

We highlight that what may be most interesting about the above corollary to quantum information theorists is that it does not depend on $f$ being \textit{operator} convex, which is when the Petz $f$-divergences are known to satisfy the data processing inequality whereas $\ol{\chi}_{2}(\rho \Vert \sigma)$ does satisfy data processing inequality. However, we note that one would have to control $\kappa^{\uparrow}_{f}(\mbf{p}_{XY}^{[\cE(\rho),\cE(\sigma)]})$ to make a proper claim about data processing being satisfied. It follows that, unless $\kappa_{f}^{\uparrow}$ has an input-independent upper bound, this relation does not extend data-processing in any sense to $\ol{D}_{f}$. It seems likely universally bounding $\kappa$ can only be done generically for operator convex functions given the relation between operator monotonicity and operator convexity (See \cite[Chapter V]{Bhatia-1997}). We also note Corollary~\ref{cor:Petz-f-chi-squared-bounds} also implies reverse Pinsker inequalities, but they are in terms of the NS distributions, so we omit them from the main text (See Corollary~\ref{cor:Petz-f-Reverse-Pinsker} in the Appendix).

We now establish our input-dependent contraction coefficient bounds. This will make use of the following straightforward proposition, which is similar to Lemma~\ref{lem:f-div-lb-by-chi-squared}. The proof is provided in the appendix.
\begin{proposition}\label{prop:Petz-f-relations-for-contraction-coeff}
Let $P,Q \in \Pos(A)$ such that $P \ll Q$,
    \begin{enumerate}
        \item $$\ol{\chi}^{2}(P \Vert Q) \leq \wt{\lambda}_{\min}(Q)^{-1} \Vert P - Q \Vert^{2}_{2} \leq  4\wt{\lambda}_{\min}(Q)^{-1}\mrm{TD}(P,Q)^{2} \ . $$ where $\wt{\lambda}_{\min}(Q) := \min_{i \in \lambda_{i}(Q) > 0} \lambda_{i}(Q)$ is the smallest non-zero eigenvalue of $Q$.
        \item Let $f$ be operator convex, then
        \begin{align}
            \ol{D}_{f}(P\Vert Q) \geq  \frac{L_{f} \cdot \wt{\lambda}_{\min}(Q)}{8}\ol{\chi}^{2}(P\Vert Q) \ .
        \end{align}
    \end{enumerate}
\end{proposition}

\begin{proposition}\label{prop:input-dependent-quantum-contraction-coefficients}
Let $f$ be twice continuously differentiable over $(0,+\infty)$, operator convex, satisfy $\vert f''(0)\vert < +\infty$, and such that $L_{f} > 0$. Let $\sigma \in \Density(A)$ such that either (i) $\sigma \in \Pd(A)$ or (ii) $f'(+\infty) = +\infty$. Then for Petz $f$-divergences, $\ol{D}_{f}$, we have the following contraction coefficient bounds:
    \begin{equation}
    \begin{aligned}
        \eta_{f}(\cE,\sigma) 
        \leq \frac{8}{L_{f}\wt{\lambda}_{\min}(\sigma)} \left[\sup_{\substack{\rho \in \Density(A):\\ 0 < D_{f}(\rho||\sigma) < + \infty }} \kappa_{f}^{\uparrow}\left(\mbf{p}^{[\cE(\rho),\cE(\sigma)]}_{XY},\mbf{q}^{[\cE(\rho),\cE(\sigma)]}_{XY}\right)\right] \cdot \eta_{\ol{\chi}^{2}}(\cE,\sigma) \ ,
    \end{aligned}
    \end{equation}
    and
    \begin{equation}
        \eta_{\ol{\chi}^{2}}(\cE,\sigma) 
        \leq 
        \sup_{\rho : \sigma \neq \rho \ll \sigma } \left[\frac{\kappa_{f}^{\uparrow}(\mbf{p}^{[\rho,\sigma]},\mbf{q}^{[\rho,\sigma]})}{\kappa_{f}^{\downarrow}(\mbf{p}^{[\cE(\rho),\cE(\sigma)]},\mbf{q}^{[\cE(\rho),\cE(\sigma)]})}  \right] \cdot \eta_{f}(\cE,\sigma) \ ,
    \end{equation}
    where we are using the Nussbaum-Sko\l a distributions (See \eqref{eq:NS-distributions}).
\end{proposition}
\begin{proof}
    The proof of the upper bound is identical to the proof of Lemma~\ref{lem:classical-input-dependent-contraction-coeff-bounds} except we need to appeal to Corollary~\ref{cor:Petz-f-chi-squared-bounds} and Proposition~\ref{prop:Petz-f-relations-for-contraction-coeff} rather than Theorem~\ref{thm:f-div-chi-squared-bounds} and Lemma~\ref{lem:f-div-lb-by-chi-squared}, so we omit the proof. We do note that we need $f$ to be operator convex to apply Item 2 of Proposition~\ref{prop:Petz-f-relations-for-contraction-coeff}.

    We now turn to proving the lower bound. First note that $f'(+\infty) = +\infty$, so for any $\sigma \in \Density(A)$ such that $\chi^{2}(\rho \Vert \sigma) < +\infty$ satisfies $\rho \ll \sigma$. Moreover, $\ol{\chi}^{2}(\rho \Vert \sigma) = 0$ if and only if $\rho = \sigma$. 
    
    Let $\rho \ll \sigma$. Then,
    \begin{align}\label{eq:petz-contract-lb-step-1}
        \eta_{f}(\cE,\sigma) \geq \frac{\ol{D}_{f}(\cE(\rho)\Vert \cE(\sigma))}{\ol{D}_{f}(\rho\Vert \sigma)} 
        \geq  \frac{\kappa_{f}^{\downarrow}(\mbf{p}^{[\cE(\rho),\cE(\sigma)]},\mbf{q}^{[\cE(\rho),\cE(\sigma)]})}{\kappa_{f}^{\uparrow}(\mbf{p}^{[\rho,\sigma]},\mbf{q}^{[\rho,\sigma]})} \frac{\ol{\chi}^{2}(\cE(\rho)\Vert \cE(\sigma))}{\ol{\chi}^{2}(\rho\Vert \sigma)} \ ,
    \end{align}
    where we used that a contraction coefficient is a supremum in the first inequality, Corollary~\ref{cor:Petz-f-chi-squared-bounds} in the second inequality. Thus,
    \begin{align*}
        \eta_{\ol{\chi}^{2}}(\cE,\sigma) =& \sup_{\rho : 0 < \ol{\chi}^{2}(\rho \Vert \sigma) < + \infty } \frac{\ol{\chi}^{2}(\cE(\rho)\Vert \cE(\sigma))}{\ol{\chi}^{2}(\rho\Vert \sigma)} \\
        \leq &
        \sup_{\rho : \sigma \neq \rho \ll \rho } \left[\frac{\kappa_{f}^{\uparrow}(\mbf{p}^{[\rho,\sigma]},\mbf{q}^{[\rho,\sigma]})}{\kappa_{f}^{\downarrow}(\mbf{p}^{[\cE(\rho),\cE(\sigma)]},\mbf{q}^{[\cE(\rho),\cE(\sigma)]})}  \right] \cdot \eta_{f}(\cE,\sigma) \ , 
    \end{align*}
    where the inequality is just re-ordering \eqref{eq:petz-contract-lb-step-1} and our observations about the feasible set of the supremum. This completes the proof.
\end{proof}

Under extra restrictions on $f$, we may extend the linear contraction bounds of Proposition~\ref{prop:linear-contraction-bounds} to the Petz $f$-divergences as well.
\begin{lemma}\label{lem:Petz-linear-contraction-coeff-bounds}
    Let $f:(0,\infty) \to \mbb{R}$ be operator convex, differentiable at unity, $f(1) = 0$ and $f(0) < +\infty$, and such that $g(t) \coloneq \frac{f(t)-f(0)}{t}$ is concave on $(0,\infty)$. Let $\sigma > 0$. Then,
    \begin{align}
        \eta_{f}(\cE,\sigma) \leq \frac{8}{L_{f}\lambda_{\min}(\sigma)}[f'(1)-f(0)]\eta_{\ol{\chi}^{2}}(\cE,\sigma) \ . 
    \end{align}
\end{lemma}
\begin{proof}
We begin by establishing the following inequalities
    \begin{align}
        \frac{\ol{D}_{f}(\cE(\rho)\Vert \cE(\sigma))}{\ol{D}_{f}(\rho \Vert \sigma)} \leq& \frac{\ol{D}_{f}(\cE(\rho)\Vert \cE(\sigma))}{\frac{L_{f}\lambda_{\min}(\sigma)}{8}\ol{\chi}^{2}(\rho\Vert\sigma)} \\
        =& \frac{8}{L_{f}\lambda_{\min}(\sigma)}\frac{D_{f}(\mbf{p}^{[\cE(\rho),\cE(\sigma)]} \Vert \mbf{q}^{\cE(\rho),\cE(\sigma)]})}{\ol{\chi}^{2}(\rho\Vert\sigma)} \\
        \leq& \frac{8}{L_{f}\lambda_{\min}(\sigma)}\frac{[f'(1)-f(0)]\chi_{2}(\mbf{p}^{[\cE(\rho),\cE(\sigma)]} \Vert \mbf{q}^{\cE(\rho),\cE(\sigma)]})}{\ol{\chi}^{2}(\rho\Vert\sigma)} \\
        =& \frac{8}{L_{f}\lambda_{\min}(\sigma)}\frac{[f'(1)-f(0)]\ol{\chi}_{2}(\cE(\rho) \Vert \cE(\sigma))}{\ol{\chi}^{2}(\rho\Vert\sigma)}
    \end{align} 
    where the first inequality uses Proposition~\ref{prop:Petz-f-relations-for-contraction-coeff}, the second inequality uses \cite[Lemma 6]{Makur-2020a}, and the equalities use the NS distribution (Proposition~\ref{prop:petz-f-with-NS-dists}). Noting no constants depend on $\rho$, we may supremize over $\rho$ to complete the proof.
\end{proof}
We remark the primary difference from Proposition~\ref{prop:linear-contraction-bounds} is the coefficient of $8$ rather than $4$. It is unclear if we can recover the $4$ as the proof method for Lemma~\ref{lem:f-div-lb-by-chi-squared} relies upon $\mbf{p}$ and $\mbf{q}$ commuting if we wrote them as matrices. 

We also highlight an interesting relation to operator convexity in the above lemma. As noted below the proof of \cite[Lemma 6]{Makur-2020a}, one could alternatively require $\frac{f(t)}{t}$ to be a concave function rather than $\frac{f(t)-f(0)}{0}$. It is known that for continuous function $f:(0,\alpha) \to \mbb{R}$, $g(t) = \frac{f(t)}{t}$ is operator monotone on $(0,\alpha)$ if and only if $f(0) \leq 0$ and $f$ is operator convex \cite[Theorem V.2.9]{Bhatia-1997}. For a continuous function $\wt{g}:[0,\infty) \to [0,\infty)$, it is known operator concavity and operator monotonicity are equivalent \cite[Theorem V.2.5]{Bhatia-1997}. The conditions on Lemma~\ref{lem:Petz-linear-contraction-coeff-bounds} seem to suggest they are selecting for the class of functions where these three relations all hold simultaneously. 

\paragraph{Ergodic Claims for Petz \texorpdfstring{$f$}{}-Divergences}
With our contraction coefficient relations established, we may establish our claims about ergodic quantum systems, which show that we may `collapse' down to caring about the $\ol{\chi}^{2}$ contraction coefficient. We stress though that, unlike the classical case, we do not know how to compute $\eta_{\ol{\chi}^{2}}(\cE,\sigma)$, so while we have simplified the problem, we have not resolved the computational aspect.

\begin{theorem}\label{thm:Petz-contraction-coefficient-rate}
     Consider any twice differentiable \textit{operator} convex function $f:(0,\infty) \to \mbb{R}$ such that $f(1)=0$, $f''(1) >0$, and such that $L_{f} > 0$. Let $\cE$ be mixing (Definition~\ref{def:mixing}). If $f'(+\infty) = +\infty$ and either $\vert f''(0) \vert <+\infty$ or $\cE$ is strongly mixing, then any state $\rho \in \Density(A)$ converges to the stationary state at a rate of at most $\eta_{\chi^{2}}(\cE,\pi)$, i.e.
    \begin{align}
        \lim_{n \to \infty} \eta_{f}(\cE^{\circ n},\pi)^{1/n} \leq \eta_{\ol{\chi}^{2}}(\cE,\pi) \ ,
    \end{align}
    Moreover, we know the above bound can be tight given Theorem~\ref{thm:classical-contraction-rate} and Proposition~\ref{prop:classical-contraction-coeff-from-quantum}.
\end{theorem}
\begin{proof}
    The argument is effectively the same as in Theorem~\ref{thm:classical-contraction-rate} except there is the nuance with regards to the Nussbaum-Sko\l a distributions in the $\kappa_{f}^{\uparrow}$ function, which we explain. By definition of a mixing channel, for all $\rho \in \Density(A)$, $\lim_{n \to \infty} \cE^{n}(\rho) = \pi$. It follows 
    \begin{align*}
        \lim_{n \to \infty} \mbf{p}^{[\cE^{n}(\rho),\cE^{n}(\pi)]}_{XY} =  \lim_{n \to \infty} \mbf{p}^{[\cE^{n}(\rho),\pi]}_{XY} =  \mbf{p}^{[\pi,\pi]}_{XY} = \sum_{x,y} \delta_{x,y} \lambda_{x}(\pi) \ket{x,y} = \lim_{n \to \infty} \mbf{q}^{[\cE^{n}(\rho),\cE^{n}(\sigma^{\star})]}_{XY} \ ,
    \end{align*}
    where the first equality is because $\pi$ is the fixed point of $\cE$ by definition of mixing channel, the second is the definition of mixing channel, the third is the definition of the Nussbaum-Sko\l a distribution (See \eqref{eq:NS-distributions}) and the final equality is because the same argument holds for the other Nussbaum-Sko\l a distribution. Thus, we have for all $\rho \in \Density(A)$, 
    $$\lim_{n \to \infty} \kappa_{f}^{\uparrow}\left(\mbf{p}^{[\cE^{n}(\rho),\cE^{n}(\pi)]}_{XY},\mbf{q}^{[\cE^{n}(\rho),\cE^{n}(\pi)]}_{XY}\right) = \kappa_{f}^{\uparrow}(\mbf{p}^{[\pi,\pi]}_{XY},\mbf{p}^{[\sigma^{\star},\sigma^{\star}]}_{XY}) = f''(1) \ , $$
    as follows from the definition of $\kappa_{f}^{\uparrow}$ (See \eqref{eq:kappa-up-arrow-defn}). This addresses the upper bound for Item 2. 
    
    To justify that the rate can be exact, recall that the Petz $f$-divergences are known to satisfy the data processing inequality for operator convex $f$. It follows from Proposition~\ref{prop:classical-contraction-coeff-from-quantum} that for classical-to-classical channel $\cW_{X \to X}$, classical reference state $\sigma_{X} = q_{X}$, and operator convex $f$, the input-dependent contraction coefficient of $\ol{D}_{f}$ for $\cW$ is the same as the contraction coefficient for the corresponding classical divergence $D_{f}$. Thus, the moreover statement in Theorem~\ref{thm:classical-contraction-rate} implies this situation is tight. This completes the proof.
\end{proof}

We make some remarks about the assumptions in the above theorem. First, one may note that one could replace the input dependent contraction coefficient with the input independent contraction coefficient so long as the limiting behaviour of the channel guaranteed all inputs were full rank. While this may appear a relaxation of the assumptions, as noted earlier, this would in fact guarantee the channel converges to a unique, full rank state \cite[Theorem 6.7]{Wolf-2012a}. Therefore, we gain no generality by considering this setting, and, as the input-dependent contraction coefficient can only be smaller than the input independent by definition, we would be loosening our upper bound. Thus Item 2 of Theorem~\ref{thm:Petz-contraction-coefficient-rate} is effectively the strongest result we can expect using just contraction coefficients.

\textbf{Mixing Times} Lastly, we establish mixing times with respect to both trace distance and a set of $f$-divergences. We define mixing time for a mixing channel $\cE$ with unique stationary state $\pi$ to be $\delta$-indistinguishable under distinguishability measure $\Delta$ as 
\begin{align}\label{eq:quantum-mixing-time-def}
    t^{\Delta}_{\text{mix}}(\cE,\delta) \coloneq \min\{n \in \mbb{N}: \max_{\rho \in \Density(A)} \Delta(\cE^{ n}(\rho),\pi) \leq \delta \} \ . 
\end{align}

\begin{proposition}\label{prop:Petz-mixing-times}
Let $\cE$ be a mixing channel with unique stationary point $\pi$. Then, whenever $\eta_{\ol{\chi}^{2}}(\cE,\pi) < 1$,
    \begin{align}
        t^{\text{TD}}_{\text{mix}}(\cE,\delta) \leq \frac{\log(\lambda_{\min}(\pi)/\delta^{2})}{\log(1/\eta_{\ol{\chi}^{2}}(\cE,\pi))}
    \end{align}
Moreover, for $f:(0,\infty) \to \mbb{R}$ be operator convex, differentiable at unity, $f(1) = 0$ and $f(0) < +\infty$, and such that $g(t) \coloneq \frac{f(t)}{t}$ is concave on $(0,\infty)$. Then, whenever $\eta_{\ol{\chi}^{2}}(\cE,\pi) < 1$,
\begin{align}
        t^{D_{f}}_{\text{mix}}(\cE,\delta) \leq \frac{\log(4[f'(1)-f(0)]\lambda_{\min}(\pi)/\delta)}{\log(1/\eta_{\ol{\chi}^{2}}(\cE,\pi))}
    \end{align}
\end{proposition}
\begin{proof}
    We begin with the trace distance mixing times. We obtain
    \begin{align}
        4 \text{TD}(\cE^{n}(\rho),\pi)^{2} \leq \ol{\chi}^{2}(\cE^{n}(\rho) \Vert \pi) \leq \eta_{\ol{\chi}^{2}}(\cE,\pi)^{n}\ol{\chi}^{2}(\rho \Vert \pi) \leq \eta_{\ol{\chi}^{2}}(\cE,\pi)^{n}4\lambda_{\min}(\sigma) \text{TD}(\rho,\pi)^{2} \ ,
    \end{align}
    where we used Corollary \ref{cor:pinsker-chisq} (via Corollary \ref{cor:quantum-f-divergence-pinsker}), the sub-multiplicativity of contraction coefficients (Proposition~\ref{prop:sub-multiplicativity}) and Proposition~\ref{prop:Petz-f-relations-for-contraction-coeff}. Using that $\text{TD}(\rho,\pi) \leq 1$ and solving for this to be smaller than $\delta$ completes the derivation.

    For the $f$-divergences, we have
    \begin{align}
        D_{f}(\cE^{n}(\rho) \Vert \pi) \leq [f'(1)-f(0)]\ol{\chi}^{2}(\cE^{n}(\rho) \Vert \pi) \ ,
    \end{align}
    which makes use of \cite[Lemma 6]{Makur-2020a} and the NS distributions (Proposition~\ref{prop:petz-f-with-NS-dists}). The rest of the derivation is the same as the trace distance case. This completes the proof.
\end{proof}

To the best of our knowledge, no one else has considered mixing times as measured under quantum $f$-divergences. We however note that in \cite{Temme-2010a} the authors established mixing time bounds for a fixed initial state $\rho$ for an arbitrary quantum generalization of the $\chi^{2}$-divergence. Their result is in terms of the eigenvalues of certain maps. To the best of our knowledge, the eigenvalues of these maps are not known to be efficient compute. In contrast, our bounds are independent of the initial quantum state and in terms of the contraction coefficient, which we also do not know how to compute. We note that the contraction coefficients for these quantum $\chi^{2}$-divergences were studied in \cite{Cao-2019a}, which established various properties but did not establish any of them to be efficiently computable (\cite{Cao-2019a} does compute some special cases for qubits using relations in terms of the Pauli basis). To the best of our knowledge the only computable upper bound on contraction for a quantum channel known is from \cite{Hirche-2024a}. An interesting problem is a better understanding of $\ol{\chi}^{2}(\cE,\sigma)$ and how to upper bound it.

\section{Acknowledgements}
I.G. thanks Eric Chitambar, Christoph Hirche, and Marco Tomamichel for helpful discussions. {The authors also thank Jamie Sikora for his emotional support.} This research is supported by the National Research Foundation, Singapore and A*STAR under its Quantum
Engineering Programme (NRF2021-QEP2-01-P06). I.G.~was also supported by NSF Grant No.~2112890 during this project. A.B.~acknowledges the partial support provided by Commonwealth Cyber Initiative (CCI-SWVA) under the 2024 Cyber Innovation Scholars Program and the Kafura Fellowship awarded by the Department of Computer Science at Virginia Tech. 

\bibliography{References.bib}

\appendices

\section{Lower Bounds for \texorpdfstring{$f$}{f}-Divergences}\label{appx:bounds}

\begin{proof}[Proof of Proposition~\ref{prop:example}]
    With our choice of
    \begin{equation*}
        f(t) = \begin{cases}
            \frac{1}{2} t (t-1) & \text{if}\ t \leq 1\ ,\\
            t \ln t - \frac{1}{2}(t-1) & \text{otherwise}\ ,
        \end{cases}
    \end{equation*}
    we have $f(1) = 0$ and $f(0) = 0 < +\infty$ and
    \begin{equation*}
        f''(t) = \begin{cases}
            1 & \text{if}\ t \leq 1\ ,\\
            t^{-1} & \text{otherwise}\ ,
        \end{cases}
    \end{equation*}
    which is continuous but not differentiable.
    We may additionally compute
    \begin{equation*}
        \frac{d^2}{dt^2} \frac{f(t)}{t} = \begin{cases}
            0 & \text{if}\ t \leq 1\ ,\\
            t^{-3}(1-t) & \text{otherwise}\ ,
        \end{cases}
    \end{equation*}
    which is non-positive.
    We may then substitute $f''$ into the right-hand side of \eqref{eq:lambda0} to obtain
    \begin{equation*}
        h_0(x,y) \coloneqq \begin{cases}
            y^{-1} + (1-x)^{-1} & \text{if}\ x \leq y\ ,\\
            x^{-1} + (1-y)^{-1} & \text{otherwise}\ ,
        \end{cases}
    \end{equation*}
    and then compute its gradient as
    \begin{equation*}
        \grad h_0 = 2 \begin{cases}
            \begin{bmatrix} (x-1)^{-2} \\ -y^{-2} \end{bmatrix} \quad \text{if}\ x \leq y\ ,&
            \begin{bmatrix} -x^{-2} \\ (1-y)^{-2} \end{bmatrix} \quad \text{otherwise}\ ,
        \end{cases}
    \end{equation*}
    which has no critical points.
    It remains to evaluate $h_0$ at its boundaries and $x = y$.
    When $y \in \{0,1\}$ we obtain $x^{-1} + 1$ and $1 + (1-x)^{-1}$, which are minimized at $x = 1$ and $x = 0$ where they equal $2$.
    The same reasoning applies when $x \in \{0,1\}$.
    When $x = y$, we get $x^{-1} (1-x)^{-1}$, which is minimized at $x=1/2$ where it equals $4$.
    Thus, $h_0$ is minimized over $x,y \in [0,1]$ at $(0,1)$ and $(1,0)$ where it equals $2$.
\end{proof}

\begin{proof}[Proof of Proposition~\ref{prop:pinsker-renyi-uni}]
    \DeclarePairedDelimiterX{\weird}[2]{\llparenthesis}{\rrparenthesis}{#1,#2}
    With our choice of $f_\alpha$ as in \eqref{eq:fa}, we have $f''(t) = t^{\alpha-2}$.
    Begin by showing \eqref{eq:lambda0}.
    Substituting $f''$ into its right-hand side yields
    \begin{equation*}
        h_0(x,y) \coloneqq x^{\alpha-2} y^{1-\alpha} + (1-x)^{\alpha-2} (1-y)^{1-\alpha}\ ,
    \end{equation*}
    where we note that both terms are non-negative.
    
    \begin{itemize}
        \item Let $\alpha \in [1, \infty)$.
        When $y \leq x$, we have
        \begin{equation*}
            x^{\alpha-2} y^{1-\alpha} \geq x^{-1} \geq 1\ .
        \end{equation*}
        When $y \geq x$, we have $1 - y \leq 1 - x$ and similarly
        \begin{equation*}
            (1-x)^{\alpha-2} (1-y)^{1-\alpha} \geq (1-x)^{-1} \geq 1\ .
        \end{equation*}
        This implies $h_0(x,y) \geq 1$ in this case.
    
        \item Let $\alpha \in (-\infty, 2]$.
        When $x \leq y$, we have
        \begin{equation*}
            x^{\alpha-2} y^{1-\alpha} \geq y^{-1} \geq 1\ .
        \end{equation*}
        When $x \geq y$, we have $1 - x \leq 1 - y$ and similarly
        \begin{equation*}
            (1-x)^{\alpha-2} (1-y)^{1-\alpha} \geq (1-y)^{-1} \geq 1\ .
        \end{equation*}
        This implies $h_0(x,y) \geq 1$ in this case.
    
        \item Let $\alpha = 1$.
        $h_0(x,y) = x^{-1} + (1-x)^{-1}$ is minimized at $x = 1/2$ and equals $4$.

        \item Let $\alpha = 2$.
        $h_0(x,y) = y^{-1} + (1-y)^{-1}$ is minimized at $y = 1/2$ and equals $4$.

        \item Let $a \in (1, 2)$.
        Note that $h_0$ grows positively unbounded when approaching the boundary of $x,y\in[0,1]$.
        For convenience, denote $\weird{a}{b}_+ \coloneqq x^a y^b + (1-x)^a (1-y)^b$ and $\weird{a}{b}_- \coloneqq x^a y^b - (1-x)^a (1-y)^b$.
        Compute the gradient and Hessian of $h_0$.
        \begin{align*}
            \nabla h_0 &= \begin{pmatrix} (\alpha-2)\ \weird{\alpha-3}{1-\alpha}_- \\ (1-\alpha)\ \weird{\alpha-2}{-\alpha}_- \end{pmatrix}\ , \\
            H_{h_0} &= \begin{pmatrix} (\alpha-2)(\alpha-3) & (\alpha-2)(1-\alpha) \\ (1-\alpha)(\alpha-2) & (1-\alpha)(-\alpha) \end{pmatrix} \odot \begin{pmatrix} \weird{\alpha-4}{1-\alpha}_+ & \weird{\alpha-3}{-\alpha}_+ \\ \weird{\alpha-3}{-\alpha}_+ & \weird{\alpha-2}{-1-\alpha}_+ \end{pmatrix}\ .
        \end{align*}
        The determinants of the above matrices in the Hessian are
        \begin{align*}
            (\alpha-3)(\alpha-2)(\alpha-1)(\alpha) - (\alpha-2)^2 (\alpha-1)^2 = -2 (\alpha-2)(\alpha-1) &> 0\ ,\\
            x^{\alpha-4} y^{-1-\alpha} (1-x)^{\alpha-4} (1-y)^{-1-\alpha}\ (x-y)^2 &\geq 0\ .
        \end{align*}
        By the Schur product theorem, $H_{h_0}$ is positive semidefinite for any $x, y \in (0, 1)$, and thus $h_0$ is convex.
        Note that $\weird{a}{b}_- = 0$ when $x = y = 1/2$, which is a critical point and thus a global minimum where $h_0(1/2, 1/2) = 4$.
    \end{itemize}
    
    We now prove \eqref{eq:lambda1} for $\alpha \in [-1, 0]$.
    Substituting $f''$ into its right-hand side yields
    \begin{equation*}
        h_1(x,y) \coloneqq x^{\alpha} y^{-1-\alpha} + (1-x)^{\alpha} (1-y)^{-1-\alpha}\ ,
    \end{equation*}
    where we note that both terms are non-negative.
    Note that we may repeat the arguments above to get the lower bound of $1$ when $\alpha \in (-\infty, 0]$ or $\alpha \in [-1, \infty)$.
    Instead, we obtain the following tighter bound.
    \begin{itemize}
        \item Let $\alpha = -1$.
        $h_1(x,y) = x^{-1} + (1-x)^{-1}$ is minimized at $x = 1/2$ and equals $4$.

        \item Let $\alpha = 0$.
        $h_1(x,y) = y^{-1} + (1-y)^{-1}$ is minimized at $y = 1/2$ and equals $4$.
    
        \item Let $a \in (-1, 0)$, and note that $h_1$ grows positively unbounded when approaching the boundary of $x,y\in[0,1]$.
        Compute the gradient and Hessian of $h_1$.
        \begin{align*}
            \nabla h_1 &= \begin{pmatrix} \alpha\ \weird{\alpha-1}{-1-\alpha}_- \\ (-1-\alpha)\ \weird{\alpha}{-2-\alpha}_- \end{pmatrix}\ , \\
            H_{h_1} &= \begin{pmatrix} \alpha (\alpha-1) & \alpha (-1-\alpha) \\ (-1-\alpha) \alpha & (-1-\alpha)(-2-\alpha) \end{pmatrix} \odot \begin{pmatrix} \weird{\alpha-2}{-1-\alpha}_+ & \weird{\alpha-1}{-2-\alpha}_+ \\ \weird{\alpha-1}{-2-\alpha}_+ & \weird{\alpha}{-3-\alpha}_+ \end{pmatrix}\ .
        \end{align*}
        The determinants of the above matrices are
        \begin{align*}
            (\alpha-1) \alpha (\alpha+1) (\alpha+2) - \alpha^2 (\alpha+1)^2 = -2 \alpha (\alpha+1) &> 0\ ,\\
            x^{\alpha-2} y^{-3-\alpha} (1-x)^{\alpha-2} (1-y)^{-3-\alpha}\ (x - y)^2 &\geq 0\ .
        \end{align*}
        By the Schur product theorem, $H_{h_1}$ is positive semidefinite for any $x, y \in (0, 1)$, and thus $h_1$ is convex.
        Note that $\weird{a}{b}_- = 0$ when $x = y = 1/2$, which is a critical point and thus a global minimum where $h_1(1/2, 1/2) = 4$.
    \end{itemize}
\end{proof}

\begin{proof}[Proof of Proposition~\ref{prop:pinsker-symchisq}]
    With our choice of $f(t) = \frac{(t-1)^2(t+1)}{t}$, we have $f''(t) = 2(1+x^{-3})$, and hence the right-hand side of \eqref{eq:lambda0} is
    \begin{equation*}
        h_0(x,y) \coloneqq \frac{2}{y} + \frac{2y^2}{x^3} + \frac{2}{1-y} + \frac{2(1-y)^2}{(1-x)^3}\ ,
    \end{equation*}
    which grows positively unbounded when approaching the boundary of $x,y\in[0,1]$.
    Its gradient is
    \begin{equation*}
        \grad h_0 = \begin{bmatrix}
            \frac{6(y-1)^2}{(x-1)^4} - \frac{6y^2}{x^4} \\
            \frac{2}{(y-1)^2} - \frac{4(x-1)}{(x-1)^3} - \frac{2}{y^2} + \frac{4y}{x^3}
        \end{bmatrix}\ .
    \end{equation*}
    Setting the top value to zero yields $y = x^2 / (2x^2 - 2x + 1)$.
    Substituting into the bottom value gives
    \begin{equation*}
        \frac{2(2x-1)}{x(1-x)} \left( \frac{2}{2x^2 - 2x + 1} + \left(\frac{2x^2 - 2x + 1}{x(x-1)}\right)^3 \right)\ .
    \end{equation*}
    The latter term has no real roots, and as such the only solution is $2x - 1 = 0$, which gives $x=y=1/2$.
    This is the only critical point in $x,y \in (0,1)$, at which $h_0(1/2, 1/2) = 16$ is the absolute minimum.
\end{proof}

\begin{proof}[Proof of Proposition~\ref{prop:pinsker-mean}]
    With our choice of $f(t) = \left( \frac{t+1}{2} \right) \ln \left( \frac{t+1}{2 \sqrt{t}} \right)$, we have $f''(t) = \frac{1 + x^2}{4 x^2 (1 + x)}$, and hence the right-hand side of \eqref{eq:lambda0} is
    \begin{equation*}
        h_0(x,y) = \frac{1 + \left(\frac{y}{x}\right)^2}{4 (x+y)} + \frac{1 + \left(\frac{1-y}{1-x}\right)^2}{4 (2-x-y)}\ ,
    \end{equation*}
    which grows positively unbounded as $x$ approaches $0$ or $1$.
    At $y \in \{0, 1\}$, we have
    \begin{align*}
        h_0(x,0) &= \frac{1}{4x(1-x)} + \frac{x}{4(2-x)(1-x)^2}\ ,&
        h_0(x,1) &= \frac{1}{4x(1-x)} + \frac{1-x}{4(1+x)x^2}\ ,
    \end{align*}
    which are lower-bounded by $1$ via their former terms.
    Its gradient is
    \begin{equation*}
        \grad h_0 = \frac{1}{4} \begin{bmatrix}
            \frac{2-2x+x^2-2y+y^2}{(x-1)^2(2-x-y)^2} - \frac{2(y-1)^2}{(x-1)^3(2-x-y)} - \frac{2y^2}{x^3(x+y)} - \frac{x^2+y^2}{x^2(x+y)^2}\\
            \frac{2-2x+x^2-2y+y^2}{(x-1)^2(2-x-y)^2} - \frac{2(y-1)}{(x-1)^2(2-x-y)} - \frac{2y}{x^2(x+y)} - \frac{x^2+y^2}{x^2(x+y)^2}
        \end{bmatrix}\ .
    \end{equation*}
    A necessary condition for $\grad h_0 = 0$ is that its two values are equal.
    This simplifies to
    \begin{equation*}
        x - y = 0
        \qquad \text{or} \qquad
        x^3 (1-y) (x+y) = y (1-x)^3 (2-x-y)\ ,
    \end{equation*}
    the latter of which has no solutions in $x,y \in (0,1)$.
    Substituting $x = y$ into $\grad h_0 = 0$, we obtain $x = y = 1/2$.
    This is the only critical point in $x,y \in (0, 1)$, at which $h_0(1/2, 1/2) = 1$ is the absolute minimum.
\end{proof}

\begin{proof}[Proof of Proposition~\ref{prop:pinsker-jeffrey}]
    With our choice of $f(t) = (t-1) \ln t$, we have $f''(t) = t^{-1} + t^{-2}$, and hence \eqref{eq:lambda12} equals
    \begin{equation*}
        h_{1/2}(x,y) = \frac{(x+y)^3}{4 x^2 y^2} + \frac{(2-x-y)^3}{4 (1-x)^2 (1-y)^2}\ ,
    \end{equation*}
    which grows positively unbounded when approaching the boundary of $x, y \in [0, 1]$
    Its gradient is
    \begin{equation*}
        \grad h_{1/2} = \begin{bmatrix}
            \frac{(x+y)^2 (x-2y)}{4 x^3 y^2} - \frac{(2-x-y)^2 (-1-x+2y)}{4 (1-x)^3 (1-y)^2} \\
            \frac{(x+y)^2 (y-2x)}{4 x^2 y^3} - \frac{(2-x-y)^2 (-1-y+2x)}{4 (1-x)^2 (1-y)^3}\ .
        \end{bmatrix}
    \end{equation*}
    The condition $\grad h_{1/2} = 0$ may then equivalently be written as
    \begin{equation}\label{eq:pinsker-jeffrey-grad0}
        (x+y)^2 (1-x)^2 (1-y)^2 \begin{bmatrix} (x-2y)(1-x) \\ (y-2x) (1-y) \end{bmatrix} = (2-x-y)^2 x^2 y^2 \begin{bmatrix} (-1-x+2y) x \\ (-1-y+2x) y \end{bmatrix}\ ,
    \end{equation}
    for which the below condition is necessary.
    \begin{equation*}
        (x-2y)(1-x)(-1-y+2x)y = (y-2x)(1-y)(-1-x+2y)x\ .
    \end{equation*}
    This reduces to
    \begin{equation*}
        x-y = 0
        \quad\text{or}\quad
        x^2 + x + y^2 + y - 4xy = 0\ ,
    \end{equation*}
    the former of which is the line $y = x$ and the latter a hyperbola with no solutions in $x, y \in (0, 1)$.
    All critical points within $x, y \in (0, 1)$ must thus have $y = x$, which we substitute into \eqref{eq:pinsker-jeffrey-grad0} to obtain
    \begin{equation*}
        4 x^3 (1-x)^5 = 4 x^5 (1-x)^3\ .
    \end{equation*}
    The above is satisfied precisely when either $x = 0$, $1-x = 0$, or $x = 1-x$, giving us the critical points $y = x \in \{0, 1/2, 1\}$.
    As we have already noted that $g$ grows unbounded when approaching $y = x \in \{0, 1\}$, it remains to check $x = y = 1/2$.
    At this point, $h_{1/2}(1/2, 1/2) = 8$, which is then the absolute minimum of $h_{1/2}$ over $x, y \in [0, 1]$, and hence a lower bound for it.
\end{proof}

\begin{proof}[Proof of Proposition~\ref{prop:pinsker-renyi-multi}]
    With our choice of $f_{\alpha}$ as in \eqref{eq:fa}, we have $f_\alpha''(t) = t^{\alpha-2}$, and hence \eqref{eq:lambda12} equals
    \begin{equation*}
        h_{1/2}(x,y) = (x+y)^2 \frac{x^{\alpha-2}}{4y^{\alpha+1}} + (2-x-y)^2 \frac{(1-x)^{\alpha-2}}{4(1-y)^{\alpha+1}}\ .
    \end{equation*}
    which grows positively unbounded when approaching the boundary of $x,y \in [0,1]$.
    Its gradient is
    \begin{equation*}
        \grad h_{1/2} = \frac{1}{4} \begin{bmatrix}
            (2-x-y) \frac{(1-x)^{\alpha-2}}{(1-y)^{\alpha+1}} \left(-2 - \frac{(\alpha-2)(2-x-y)}{1-x}\right) - (x+y) \frac{x^{\alpha-2}}{y^{\alpha+1}} \left( -2 - \frac{(\alpha-2)(x+y)}{x} \right)\\
            (2-x-y) \frac{(1-x)^{\alpha-2}}{(1-y)^{\alpha+1}} \left(-2 + \frac{(\alpha+1)(2-x-y)}{1-y}\right) - (x+y) \frac{x^{\alpha-2}}{y^{\alpha+1}} \left( -2 + \frac{(\alpha+1)(x+y)}{y} \right)
        \end{bmatrix}\ .
    \end{equation*}
    The condition $\grad h_{1/2} = 0$ is equivalent to
    \begin{equation}\label{eq:pinsker-renyi-grad0}
        (2-x-y)\frac{(1-x)^{\alpha-2}}{(1-y)^{\alpha+1}} \begin{bmatrix} -2 - \frac{(\alpha-2)(2-x-y)}{1-x} \\ -2 + \frac{(\alpha+1)(2-x-y)}{1-y} \end{bmatrix} = (x+y)\frac{x^{\alpha-2}}{y^{\alpha+1}} \begin{bmatrix} -2 - \frac{(\alpha-2)(x+y)}{x} \\ -2 + \frac{(\alpha+1)(x+y)}{y} \end{bmatrix}\ .
    \end{equation}
    The products of vector cross-terms must be equal, which when $x,y \in (0,1)$ is equivalent to
    \begin{equation*}
        x-y = 0
        \quad\text{or}\quad
        2y(\alpha-2)(\alpha(x-1)+x) + x(\alpha+1)(\alpha(x-2)+2) + y^2(\alpha-2)(\alpha-1) = 0\ .    \end{equation*}
    The only solutions of the latter within $x,y \in [0,1]$ are $x = y \in \{0, 1\}$.
    Thus, $x=y$ is a necessary condition for interior critical points.
    As a side note, we may alternatively recover this constraint by considering $\alpha \partial_x h_{1/2} + (1-\alpha) \partial_y h_{1/2} = 0$.
    Substituting into \eqref{eq:pinsker-renyi-grad0}, we obtain $(1-x) = x$, meaning that $x = y = 1/2$ is the only interior critical point.
    At this point, $h_{1/2}(1/2, 1/2) = 4$, which is hence the absolute minimum of $h_{1/2}$ over $x,y \in [0,1]$.
\end{proof}

\begin{proof}[Proof of Proposition~\ref{prop:lins}]
    When $\theta \in \{0,1\}$, the bound is trivial.
    With our choice of $f(t) = \theta t\ln t - (\theta t + 1 - \theta) \ln (\theta t + 1 - \theta)$, we have $f''(t) = \frac{\theta (1 - \theta)}{t (\theta t + 1 - \theta)}$, and hence \eqref{eq:lambda12} equals
    \begin{equation*}
        h_{1/2}(x,y) = \frac{\theta (1 - \theta)}{4} \left( \frac{(x+y)^2}{x y (\theta(x-y) + y)} + \frac{(2-x-y)^2}{(1-x) (1-y) (1 - \theta(x-y) - y)} \right)\ ,
    \end{equation*}
    which grows positively unbounded when approaching the boundary of $x,y\in[0,1]$.
    Its gradient is
    \begin{equation*}
        \grad h_{1/2} = \frac{\theta (1-\theta)}{4} \begin{bmatrix}
            \frac{(x+y) (x-y-\theta(3x-y))}{x^2 (\theta(x-y) + y)^2} + \frac{(2-x-y) (x-y-\theta(3x-y-2))}{(1-x)^2 (\theta(x-y)-1+y)^2}\\
            \frac{(x+y) (\theta(3y-x)-2y)}{y^2 (\theta(x-y) + y)^2} + \frac{(2-x-y) (\theta(3y-x-2)+2-2y)}{(1-y)^2 (\theta(x-y)-1+y)^2}
        \end{bmatrix}\ .
    \end{equation*}
    A necessary condition for $\nabla h_{1/2} = 0$ is $(1-\theta) \partial_x h_{1/2} - \theta \partial_y h_{1/2} = 0$, which simplifies to
    \begin{equation*}
        x - y = 0
        \quad\text{or}\quad
        \frac{x+y}{x^2 y^2} + \frac{2-x-y}{(1-x)^2 (1-y)^2} = 0\ ,
    \end{equation*}
    the latter of which has no solutions in $x,y \in (0, 1)$.
    Substituting the $x=y$ constraint into $\partial_x h_{1/2} = 0$ reveals $x = y = 1/2$ as the only critical point in $x,y \in (0,1)$.
    At this point, $h_{1/2}(1/2,1/2) = 4 \theta (1-\theta)$, which is hence the absolute minimum of $h_{1/2}$ over $x,y \in (0,1)$.
\end{proof}

\begin{proof}[Proof of Proposition~\ref{prop:js}]
    With our choice of $f(t) = \frac{1}{2}(t \ln t - (t+1) \ln(\frac{t+1}{2}))$, we have $f''(t) = \frac{1}{2t^2 + 2t}$, and hence \eqref{eq:lambda12} equals
    \begin{equation*}
        h_{1/2}(x,y) = \frac{x+y}{8xy} + \frac{2-x-y}{8(1-x)(1-y)}\ ,
    \end{equation*}
    which grows positively unbounded when approaching the boundary of $x,y \in [0,1]$.
    Its gradient is
    \begin{equation*}
        \grad h_{1/2} = \frac{1}{8} \begin{bmatrix} \frac{2x-1}{x^2(x-1)^2} \\ \frac{2y-1}{y^2(y-1)^2} \end{bmatrix}\ .
    \end{equation*}
    Setting $\grad h_{1/2} = 0$ yields $2x-1 = 2y-1 = 0$, meaning that $x=y=1/2$ is the sole critical point.
    At this point, $h_{1/2}(1/2, 1/2) = 1$, which is the absolute minimum.
\end{proof}

\section{Extra Results for Strong Data Processing and Petz \texorpdfstring{$f$}{f}-Divergences}

\begin{proof}[Proof of Proposition~\ref{prop:Petz-f-relations-for-contraction-coeff}]
    To establish the first item, 
    \begin{align*}
        \ol{\chi}^{2}(P \Vert Q) = \Tr[Q^{-1}(P-Q)^{2}] \leq & \Tr[\lambda_{\min}(Q)^{-1}Q^{0}(P-Q)^{2}] \\
        \leq& \wt{\lambda}_{\min}(Q)^{-1} \Tr[(P-Q)^{2}] \\
        = & \wt{\lambda}_{\min}(Q)^{-1} \Vert P - Q \Vert^{2}_{2} \leq  4\wt{\lambda}_{\min}(Q)^{-1}\mrm{TD}(P,Q)^{2} \ ,
    \end{align*}
    where the first inequality $Q \preceq \lambda_{\min}(Q)^{-1} Q^{0}$, the second inequality uses $P \ll Q$, and  the second equality is using $(P-Q)$ is Hermitian and the definition of Schatten $2$-norm, and the last inequality uses the definition of trace distance and the Schatten $1$-norm. \\

    To establish the second item, we apply Corollary~\ref{cor:quantum-f-divergence-pinsker} for the Petz $f$-divergences, which holds for operator convex $f$ as this guarantees the DPI holds for all positive semidefinite operators \cite{Hiai-2017a}. We then apply Item 1 of this proposition, so we obtain
    $$ \ol{D}_{f}(P\Vert Q) \geq  \frac{L_{f}}{2}\mrm{TD}(P,Q)^{2} \geq \frac{L_{f} \cdot \wt{\lambda}_{\min}(Q)}{8}\ol{\chi}^{2}(P\Vert Q) $$
    This completes the proof.
\end{proof}

\begin{proof}[Proof of Proposition~\ref{prop:classical-contraction-coeff-from-quantum}]
    We focus on the input-independent case. First, note that the quantities in Definition~\ref{def:quantum-contraction-coefficients} can only be larger than if we restricted to classical distributions as they are supremums. Thus, it suffices to show we can upper bound these quantities while restricting to classical distributions. Now as $\cW$ is a classical-to-classical channel, it admits a Kraus operator representation $\{W_{Y|X}(y|x)\vert y \rangle \langle x \vert\}_{x,y}$ \cite{Wilde-Book}. As may be verified by direct calculation, for all $\rho \in \Density(A)$, $\cW(\rho) = \cW(\Delta(\rho))$ where $\Delta$ completely dephases $\rho$ in the classical basis, i.e. $\Delta(Z) := \sum_{x \in \cX} \dyad{x}Z\dyad{x}$ for all $Z \in \Lin(A)$. Thus, we have
    \begin{align}
        \frac{\mbb{D}_{f}(\cW(\rho)\Vert \cW(\sigma))}{\mbb{D}_{f}(\rho \Vert \sigma)} =& \frac{\mbb{D}_{f}(\cW(\Delta(\rho))\Vert \cW(\Delta(\sigma)))}{\mbb{D}_{f}(\rho \Vert \sigma)} \nonumber \\
        =& \frac{\mbb{D}_{f}(\cW(\Delta(\rho))\Vert \cW(\Delta(\sigma)))}{\mbb{D}_{f}(\Delta(\rho) \Vert \Delta(\sigma))}\frac{\mbb{D}_{f}(\Delta(\rho)\Vert \Delta(\sigma))}{\mbb{D}_{f}(\rho \Vert \sigma)} \nonumber \\
        \leq& \frac{\mbb{D}_{f}(\cW(\Delta(\rho))\Vert \cW(\Delta(\sigma)))}{\mbb{D}_{f}(\Delta(\rho) \Vert \Delta(\sigma))} \ , \label{eq:classical-contraction-coeff-step-1}
    \end{align}
    where the inequality uses our assumption that $\mbb{D}_{f}(\Delta(\rho) \Vert \Delta(\sigma)) \leq \mbb{D}(\rho \Vert \sigma)$. Now note that $\Delta(\rho),\Delta(\sigma)$ are classical distributions always. Thus,
    \begin{align*}
        \sup_{\substack{\rho,\sigma \in \Density(A) \\ 0 < \mbb{D}_{f}(\rho||\sigma) < + \infty }} \frac{\mbb{D}_{f}(\cW(\rho)\Vert \cW(\sigma))}{\mbb{D}_{f}(\rho \Vert \sigma)} 
        = & \sup_{\substack{\rho,\sigma \in \Density(A) \\ 0 < D_{f}(\rho||\sigma) < + \infty \\  \Delta(\rho) \neq \Delta(\sigma) }} \frac{\mbb{D}_{f}(\cW(\rho)\Vert \cW(\sigma))}{\mbb{D}_{f}(\rho \Vert \sigma)}\\ 
        \leq & \sup_{\substack{\rho,\sigma \in \Density(A) \\ 0 < \mbb{D}_{f}(\rho||\sigma) < + \infty \\ 
        \Delta(\rho) \neq \Delta(\sigma) }} \frac{\mbb{D}_{f}(\cW(\Delta(\rho))\Vert \cW(\Delta(\sigma)))}{\mbb{D}_{f}(\Delta(\rho) \Vert \Delta(\sigma))} \\ 
        = & \sup_{\substack{\rho,\sigma \in \Density(A): \\ 0 < \mbb{D}_{f}(\rho||\sigma) < + \infty  \\
        \Delta(\rho) = \mbf{p} \neq \mbf{q} =\Delta(\sigma) }} \frac{\mbb{D}_{f}(\cW(\mbf{p})\Vert \cW(\mbf{q}))}{\mbb{D}_{f}(\mbf{p} \Vert \mbf{q})} \\
        \leq & \sup_{\substack{\mbf{p},\mbf{q} \in \cP(\cX): \\ 0 < D_{f}(\mbf{p}||\mbf{q}) < + \infty}} \frac{D_{f}(\cW(\mbf{p})\Vert \cW(\mbf{q}))}{D_{f}(\mbf{p} \Vert \mbf{q})} \ ,
    \end{align*}
    where we now explain the steps. The first equality is noting that if $0 < \mbb{D}_{f}(\rho \Vert \sigma)$ but $\Delta(\rho) = \Delta(\sigma)$ then $D_{f}(\cW(\rho)\Vert \cW(\sigma)) = D_{f}(\cW(\Delta(\rho)) \Vert \cW(\Delta(\sigma)) = 0$, so these cases are either suboptimal or the channel $\cW$ is a replacer channel (traces out the input and outputs a specific state) and could be achieved with any pair of states. We thus may remove the cases where $\Delta(\rho) = \Delta(\sigma)$ from the supremum without loss of generality. The first inequality is \eqref{eq:classical-contraction-coeff-step-1}. The second equality is just defining $\mbf{p},\mbf{q}$. The second inequality is as follows. First, as $f'(\infty) = +\infty$, we may conclude $\rho \ll \sigma$ if it is to satisfy $D_{f}(\rho \Vert \sigma) < + \infty$. Thus, the supremum is over $\Delta(\rho) \ll \Delta(\sigma)$, which guarantees $+\infty > D_{f}(\Delta(\rho) \Vert \Delta(\sigma)) = D_{f}(\mbf{p}\Vert \mbf{q})$. Second, note that $D_{f}(\mbf{p}\Vert \mbf{q})=0$ if and only if $\mbf{p} = \mbf{q}$ by our strict convexity assumption on $f$ (See Item 2 of Proposition~\ref{fact:f-div-properties}), so we have also guaranteed $D_{f}(\mbf{p}\Vert \mbf{q})>0$ is not a strengthening in the last inequality. This completes the proof for the input-independent case. In the input-dependent case, the argument goes through the same way so long as one fixes the reference state $\sigma$ to be classical $q_{X}$.
\end{proof}

\begin{corollary}\label{cor:Petz-f-Reverse-Pinsker}
    Let $\rho,\sigma \in \Density(A)$ such that $\rho \ll \sigma$. Then 
    \begin{align}
    \ol{D}_{f}(\rho \Vert \sigma) 
    \leq \frac{\kappa_{f}^{\uparrow}(\mbf{p}^{[\rho,\sigma]}_{XY},\mbf{q}^{[\rho,\sigma]}_{XY})}{2q^{[\rho,\sigma]}_{\min}} \left \Vert\mbf{p}^{[\rho,\sigma]}_{XY}-\mbf{q}^{[\rho,\sigma]}_{XY} \right\Vert_{2}^{2} 
    \leq \frac{2\kappa_{f}^{\uparrow}(\mbf{p}^{[\rho,\sigma]}_{XY},\mbf{q}^{[\rho,\sigma]}_{XY})}{\wt{q}^{[\rho,\sigma]}_{\min}}\mrm{TD}\left(\mbf{p}_{XY}^{[\rho,\sigma]},\mbf{q}_{XY}^{[\rho,\sigma]}\right)^{2} \ ,
    \end{align}
    where $\kappa_{f}^{\uparrow}(\cdot,\cdot)$ is defined in \eqref{eq:kappa-up-arrow-defn}, $\mbf{p}^{[\rho,\sigma]}_{XY}, \mbf{q}^{[\rho,\sigma]}_{XY}$ are the Nussbaum-Sko\l a distributions \eqref{eq:NS-distributions} and $\wt{q}^{[\rho,\sigma]}_{\min} := \min_{i \in \supp(\mbf{q}^{[\rho,\sigma]})} q_{i}^{[\rho,\sigma]}$.
\end{corollary}
\begin{proof}
    By the same argument as the previous proof, we have $\ol{D}_{f}(\rho \Vert \sigma) =  D_{f}(\mbf{p}^{[\rho,\sigma]}_{XY} \Vert \mbf{q}^{[\rho,\sigma]}_{XY})$. We then apply Corollary~\ref{cor:Reverse-Pinsker}. Finally, note $\left \Vert \wt{\mbf{p}}^{[\rho,\sigma]}_{XY} - \wt{\mbf{q}}^{[\rho,\sigma]}_{XY} \right \Vert_{2} = \left \Vert \mbf{p}^{[\rho,\sigma]}_{XY} - \mbf{q}^{[\rho,\sigma]}_{XY} \right \Vert_{2}$ as $\mbf{p}^{[\rho,\sigma]}_{XY} \ll \mbf{q}^{[\rho,\sigma]}_{XY}$ and similarly for the $1$-norm. This completes the proof.
\end{proof}

\end{document}